%% file: main.tex
\documentclass{article}
\usepackage{physics}
\usepackage{amsmath,amssymb,amsthm}
\usepackage[]{qcircuit}
\usepackage{graphicx}
\usepackage{bm}
\usepackage{caption}
\usepackage{subcaption}
\usepackage{tikz}
\usepackage[colorlinks=true,citecolor=blue]{hyperref}

\usepackage{breakurl}

\usepackage[capitalise,noabbrev]{cleveref}
\crefname{equation}{}{}
\crefname{enumi}{}{}

\usepackage[]{color}
\usepackage[natbib=true, style=authoryear-comp, maxbibnames=8, maxcitenames=1, uniquelist=false, citetracker=true, backend=biber,sortcites=false]{biblatex}				
\AtEveryCitekey{\ifciteseen{}{\defcounter{maxnames}{8}}} 
\DeclareNameAlias{sortname}{last-first}
\bibliography{stats.bib, biblio.bib}

\usepackage{siunitx}
\usepackage{todonotes}

\providecommand{\keywords}[1]
{
  \small	
  \textbf{\textit{Keywords---}} #1
}

\input{commands}

\title{On sampling determinantal and Pfaffian\\ point processes
on a quantum computer}
\date{}
\author{R\'emi Bardenet\footnote{Corresponding author: \texttt{remi.bardenet@univ-lille.fr}}, Michaël Fanuel, and Alexandre Feller\\
Universit\'e de Lille, CNRS, Centrale Lille\\
UMR 9189 - CRIStAL, F-59000 Lille, France
}

\begin{document}
\maketitle

\begin{abstract}
    \input{abstract}
\end{abstract}
    
\keywords{
    Determinantal and Pfaffian point processes, fermionic systems, quantum circuits, Givens rotations. 
}
 
\tableofcontents


\section{Introduction}
\label{s:introduction}
\input{introduction}

\section{Determinantal and Pfaffian point processes}
\label{s:DPPs}
\input{dpp.tex}

\section{From qubits to fermions}
\label{s:qubits_to_fermions}
\input{qubits_to_fermions}

\section{From fermions to Pfaffian point processes}
\label{s:fermions_to_dpps}
\input{fermions_to_dpps}
\section{Quantum circuits to sample DPPs and PfPPs}
\label{s:quantum_circuits}
\input{quantum_circuits}

%
\section{Numerical experiments}
\label{s:experiments}
\input{experiments}


\section{Discussion}
\label{s:discussion}
\input{discussion}

\section*{Declarations}

\subsection*{Competing interests}

The authors declare to have no competing interests related to this work.

\subsection*{Availability of data and materials}

The source code used in the paper is publicly available on the GitHub page
\url{https://github.com/For-a-few-DPPs-more/quantum-sampling-DPPs}.

\subsection*{Acknowledgements}
We acknowledge support from ERC grant Blackjack (ERC-2019-STG-851866) and ANR AI chair Baccarat (ANR-20-CHIA-0002).
Furthermore, we acknowledge the use of IBM Quantum services for this work. 
The views expressed are those of the authors, and do not reflect the official policy or position of IBM or the IBM Quantum team.

\appendix
\input{appendix}

\printbibliography
\end{document}

%% file: commands.tex
\makeatletter
\newrobustcmd*{\setmincitenames}{\numdef\blx@mincitenames}
\makeatother
\def\citeus{\setmincitenames{2}\citep{BFBDHRRSW22Sub}\setmincitenames{1}}
\def\i{\mathrm{i}}
\def\e{\mathrm{e}}
\def\bI{\mathbf{I}}
\def\bu{\mathbf{u}}
\def\bU{\mathbf{U}}
\def\bK{\mathbf{K}}
\def\bX{\mathbf{X}}
\def\bV{\mathbf{V}}
\def\bQ{\mathbf{Q}}
\def\bH{\mathbf{H}}
\def\bn{\mathbf{n}}

\def\bG{\mathbf{G}}
\def\bA{\mathbf{A}}
\def\bK{\mathbf{K}}
\def\bR{\mathbf{R}}
\def\bSigma{\mathbf{\Sigma}}
\def\bLambda{\mathbf{\Lambda}}
\def\bzero{\mathbf{0}}
\newcommand{\ibm}[1]{\texttt{#1}}

\def\blambda{\bm{\lambda}}
\def\d{\mathrm{d}}
\def\Diag{\mathrm{Diag}}
\def\pf{\mathrm{Pf}}
\def\pfPP{\mathrm{PfPP}}
\def\pfK{\mathbb{K}}
\def\bS{\mathbf{S}}
\def\ph{\mathrm{ph}}
\def\bF{\mathbf{C}}

\def\cO{\mathcal{O}}

\def\adjoint{*}

\def\spec{\mathrm{spec}}
\def\sig{\sigma}

\providecommand{\keywords}[1]{\textbf{\textit{Keywords ---}} #1}

\newtheorem{theorem}{Theorem}[section]
\newtheorem{lemma}[theorem]{Lemma}
\newtheorem{corollary}[theorem]{Corollary}
\newtheorem{proposition}[theorem]{Proposition}
\newtheorem{definition}[theorem]{Definition}
\newtheorem{example}[theorem]{Example}

\newtheorem{remark}[theorem]{Remark}

\DeclareMathOperator{\Ker}{\mathrm{Ker}}
\DeclareMathOperator{\sign}{\mathrm{sign}}
\DeclareMathOperator{\E}{\mathbb{E}}

\newcommand{\corr}[1]{\textcolor{black}{#1}}

%% file: abstract.tex
DPPs were introduced by Macchi as a model in quantum optics the 1970s. 
Since then, they have been widely used as models and subsampling tools in statistics and computer science.
Most applications require sampling from a DPP, and given their quantum origin, it is natural to wonder whether sampling a DPP on a quantum computer is easier than on a classical one.
We focus here on DPPs over a finite state space, which are distributions over the subsets of $\{1,\dots,N\}$ parametrized by an $N\times N$ Hermitian kernel matrix. 
Vanilla sampling consists in two steps, of respective costs $\mathcal{O}(N^3)$ and $\mathcal{O}(Nr^2)$ operations on a classical computer, where $r$ is the rank of the kernel matrix. 
A large first part of the current paper consists in explaining why the state-of-the-art in quantum simulation of fermionic systems already yields quantum DPP sampling algorithms.
We then modify existing quantum circuits, and discuss their insertion in a full DPP sampling pipeline that starts from practical kernel specifications. 
The bottom line is that, with $P$ (classical) parallel processors, we can divide the preprocessing cost by $P$ and build a quantum circuit with $\mathcal{O}(Nr)$ gates that sample a given DPP, with depth varying from $\mathcal{O}(N)$ to $\mathcal{O}(r\log N)$ depending on qubit-communication constraints on the target machine.
We also connect existing work on the simulation of superconductors to Pfaffian point processes, which generalize DPPs and would be a natural addition to the machine learner's toolbox.
\corr{In particular, we describe ``projective'' Pfaffian point processes, the cardinality of which has constant parity, almost surely.}
Finally, the circuits are empirically validated on a classical simulator and on 5-qubit IBM machines. 

%% file: introduction.tex
Determinantal point processes (DPPs) were introduced in the thesis of \cite{Mac72}, recently translated and reprinted as \citep{Mac17}.
Macchi's motivation was the design of probabilistic models for free fermions in quantum optics; see the preface of \citep{Mac17} for a history of DPPs, and \citeus{} for an extended discussion of the links between free fermions and DPPs.
DPPs have known another surge of interest since the 90s for their application to random matrix theory \citep{Joh05}.
More recently, they have been adopted as models or sampling tools in fields like spatial statistics \citep*{LaMoRu15}, Monte Carlo methods \citep{BaHa20}, machine learning \citep*{KuTa12}, or numerical linear algebra \citep{DeMa20}.
In the latter two fields, the considered DPPs are often \emph{finite}, in the sense that a DPP is a probability measure over subsets of a ground set of finite cardinality $N\gg 1$. 
Such a finite DPP is specified by an $N\times N$ matrix called its kernel matrix, which we assume here to be Hermitian.

In machine learning as in numerical linear algebra, it is crucial to be able to sample from the considered finite DPPs.
For instance, a famous DPP is the subset of edges of a uniform spanning tree in a connected graph \citep{Pemantle91}.
Sampling these uniform spanning trees is a necessary step for building the randomized preconditioners of Laplacian systems in \citep{KySo18}.
As another example, DPPs have been used as randomized summaries of large collections of items, such as a corpus of texts.
Sampling the corresponding DPP then allows to extract a small representative subset of sentences \citep[Section 4 and references therein]{KuTa11}.
Other machine learning applications include constructing coresets \citep{TBA19}, kernel matrix approximation \citep{DeKaMa20,FSS21} and feature extraction for linear regression \citep{BeBaCh20b}.

Much research has thus been devoted to sampling finite DPPs on a classical computer, either exactly or approximately. 
The default generic exact sampler is the `HKPV' sampler \citep{HKPV06}.
To fix ideas, when applied to a \emph{projection} DPP, i.e., a DPP that puts all its mass on subsets of some fixed cardinality $r\leq N$, and assuming the kernel matrix is given in diagonalized form, HKPV has complexity $\cO(Nr^2)$.
For DPPs on graphs such as uniform spanning trees, there also exist dedicated algorithms, such as the loop-erased random walks of \citet{Wilson96}, with an expected number of steps equal to the mean commute time to a chosen root node of the graph.

Given that DPPs are originally a model in quantum electronic optics, and are still the default mathematical object used to describe a quantum physical system known as \emph{free fermions} \citep{DLMS19}, it is natural to ask whether finite DPPs can be sampled more efficiently on a quantum computer than on a classical computer. 
Somewhat implicitly, the question has actually already been tackled in a string of physics papers whose goal is the more ambitious quantum simulation of fermionic systems \citep{OGKL01,WHWCNT15,KMWGACGB18,JSKSB18}. 
In a reverse cross-disciplinary direction, and still rather implicitly, the quantum algorithms therein are reminiscent of parallel QR decompositions, a key topic in numerical algebra \citep{SaKu78,DGHL12}.
While finishing this work, we also realized that in the computer science literature, and independently of the aforementioned physics works, \citet{KerPra22} recently proposed similar quantum algorithms to sample projection DPPs, as a building block for quantum data analysis pipelines.
For our purpose, their main contribution is a quantum circuit with depth logarithmic in $N$, when \citet{JSKSB18} only discuss depths linear in $N$.

On our side, motivated by applications of finite DPPs in data science, we initiated in \citeus{} a programme of reconnection of DPPs to their physical fermionic roots, to foster cross-disciplinary research between mathematics, computer science, and physics on the topic, even if our languages and lore are quite different.
In particular, physicists have developed tools for the analysis and the construction of fermionic systems that we would like to apply to DPPs in data science, without reinventing the wheel. 
The current paper is part of this programme, and sums up our understanding of what the state of the art in physics tells us on sampling finite DPPs, after we follow in the footsteps of \citep{Mac72} and map a given DPP to a fermionic density operator.
This cross-disciplinary, self-contained survey is our first contribution. 

As an example of what our disciplines can bring to each other, our second contribution is to relate the Pfaffian point processes as defined by \citet{koshida21} --~a generalization of DPPs that is natural in theoretical physics but has not yet been used in data science~-- to a quantum algorithm by \citet{JSKSB18} for solving the Bogoliubov-de Gennes Hamiltonian.
As another example of the fertility of cross-disciplinary work, after we make the link between the quantum circuits of \citep{WHWCNT15,JSKSB18} and parallel QR decompositions \citep{DGHL12}, many new variants of the quantum circuits in \citep{JSKSB18} become immediately available, adapting to a range of qubit-communication and hardware constraints. 
In particular, we exhibit a variant of the quantum circuits in \citep{JSKSB18} with the same dimensions as the best circuit in \citep{KerPra22}. 

Overall, our conclusions on quantum DPP sampling are that if a projection kernel is given in diagonalized form $\bK = \bQ^*\bQ$, with $\bQ\in\mathbb{C}^{r\times N}$ a matrix with orthonormal rows, one can build quantum circuits that sample DPP($\bK$) with $\cO(rN)$ one- and two-qubit gates, and depth depending on what hardware constraints we put on which qubits can be jointly operated. 
Acting only on neighbouring qubits, depth is $\cO(N)$ \citep{WHWCNT15,JSKSB18}, while acting on arbitrary pairs of qubits can take the depth down to $\cO(r\log N)$; see our variant of \citep{JSKSB18} in \cref{s:parallel_QR} and the logarithmic depth Clifford loaders of \cite{KerPra22}.
Such depths (i.e., the largest number of gates applied to any single qubit) favourably compare to the time complexity $\cO(Nr^2)$ of the classical HKPV algorithm, or the expected complexity in $\cO(Nr + r^3 \log r)$ of the randomized version of HKPV in \citep{DCMW19,BTA22}.

That being said, diagonalizing $\bK$ on a classical computer as a preprocessing step remains a $\cO(N^3)$ bottleneck, or at least $\cO(Nd^2)$ in the common case where the diagonalization of $\bK$ can be reduced to the SVD of an $N\times d$ matrix. 
This bottleneck thus seems to cancel the advantage of using a quantum circuit \emph{if one insists} on starting with $\bK$ stored on a classical computer. 
Yet, while the projection kernel $\bK$ may not be available in diagonalized form, it is common in data science applications \citep{KySo18,BeBaCh20b} to specify it implicitly, as a set of vectors spanning its range. 
As noted by \cite{BTA22}, using a (classical) parallel QR algorithm and $P$ processors, we can reduce the classical preprocessing step to $\cO(Nd^2/P)$ flops.
Importantly for our quantum pipeline, we discuss here how to further reuse the intermediate steps of this preprocessing in the design of the quantum circuit to apply next. 
This yields a hybrid parallel/quantum algorithm to sample projection DPPs.
Compared to the classical HKPV sampler, our pipeline thus provides a linear speedup.
Compared to the expected complexity of the randomized classical algorithm discussed in \citep{DCMW19,BTA22}, we show a gain in the sampling step, but we arguably share the same bottleneck of classical parallel QR preprocessing.
Finally, the necessity for classical preprocessing may disappear in the future, once $\bK$ can be assumed to be initially available as a quantum state, stored on a quantum computer.

The rest of the paper is organized as follows.
In \cref{s:DPPs}, we define DPPs and one of their generalizations, Pfaffian PPs (PfPPs). 
In \cref{s:qubits_to_fermions}, we introduce the vocabulary of quantum field theory, at the basis of the connection between PfPPs and free fermions. 
By sticking to the case of a finite-dimensional state space, we can avoid technical difficulties and provide a rigorous, stand-alone introduction, mostly following \citep{Nie05}. 
Section~\ref{s:fermions_to_dpps} is devoted to building a Hamiltonian starting from a  DPP or a PfPP, so that a simple quantum measurement yields a sample from the corresponding point process.
In Section~\ref{s:quantum_circuits}, we show how \cite{WHWCNT15,JSKSB18} build circuits to simulate the fermionic systems corresponding to our point processes. 
Our presentation insists on the implicit links with parallel QR algorithms, which allow us to introduce variants of the circuits with smaller complexity under assumptions on the qubit communication constraints of the target machine.
\corr{Still in Section~\ref{s:quantum_circuits}, we show that the projective PfPPs that we simulate generate sets of points with a fixed  parity of their cardinality.}
Finally, we investigate in Section~\ref{s:experiments} the implementation of the circuits with the library Qiskit \citep{Qiskit}, and the noise when running the circuits on 5-qubit IBMQ machines \citep{IBMQ}.

Appendix~\ref{s:mathematical_details} contains a few detailed proofs that we extracted from the main body of the paper. 
Appendix~\ref{s:gates_details} contains a discussion on gate details to implement the basic operations introduced in the circuits of \cref{s:quantum_circuits}.
Appendix~\ref{s:errors_in_quantum_computers} is an overview of the sources of error in current quantum computers and their orders of magnitude.


\paragraph*{Notations.} The complex conjugate of a complex number $z$ is denoted by $\overline{z}$. 
Similarly,  $\overline{\mathbf{M}}$ denotes the entrywise complex conjugate of a matrix $\mathbf{M}$.
The transpose of $\mathbf{M}$ reads $\mathbf{M}^\top$ and its Moore-Penrose pseudo-inverse is $\mathbf{M}^+$.
\corr{For two Hermitian matrices $\mathbf{M}$ and $\mathbf{N}$ of the same size, we write $\mathbf{M}\preceq \mathbf{N}$ if for all complex vectors $\mathbf{v}$ we have $\mathbf{v}^* \mathbf{M} \mathbf{v} \leq  \mathbf{v}^* \mathbf{N}\mathbf{v}$}.
The adjoint of an operator $A$ is written $A^*$.
Also, we denote the canonical basis elements of $\mathbb{C}^N$ by $\mathbf{e}_k$, $1\leq k \leq N$. 

\corr{For any positive integer} $k$, we write $[k] \triangleq \{1,\dots,k\}$.
\corr{The matrix obtained by selecting the first $k$ columns on $\mathbf{M}$ is denoted by $\mathbf{M}_{:, [k]}$.
Similarly, the $\ell$th column of $\mathbf{M}$ reads $\mathbf{M}_{:, \ell}$.}
A point process on $[N]$ is a probability measure over subsets of $[N]$.
\corr{When talking about a point process $Y$, we use $\mathbb{P}$ and $\mathbb{E}_Y$ to denote the corresponding probability and expectation, like in $\mathbb{P}(\{1,2\}\subseteq Y)$ or $\mathbb{E}_Y f(Y)$.}
Finally, the sigmoid function is $\sig(x) = 1/(1+\exp(-x))$.

%% file: dpp.tex
In this section, we introduce \emph{discrete} determinantal point processes (DPPs), and refer to \citep{KuTa12} for their elementary properties.
We also introduce Pfaffian point processes (PfPPs, \citealp{Rains2000,Soshnikov03}), a generalization of DPPs that has not yet been used in machine learning, to the best of our knowledge.
As we shall see in \cref{s:PfPP_from_QM}, both DPPs and PfPPs naturally appear when modeling physical particles known as fermions. 

\subsection{Determinantal point processes}
\corr{A determinantal point process is determined by the so-called inclusion probabilities.}
\begin{definition}[DPP]
Let $\bK \in \mathbb{C}^{N\times N}$. 
A random subset $Y \subseteq [N]$ is drawn from the DPP of marginal kernel $\bK$, denoted by $Y\sim \mathrm{DPP}(\bK)$ if and only if
\begin{equation}
  \label{eq:def_dpp}
  \forall S \subseteq [N],\quad \mathbb{P}(S \subseteq Y) = \det (\bK_{S}),
  \end{equation}
  where $\bK_{S} = [\bK_{i,j}]_{i,j \in S}$. 
  We take as convention $\det (\bK_{\emptyset}) = 1$.
\end{definition}
\corr{Note that the matrix $\bK$ is called the \emph{marginal kernel} since it determines the inclusion probabilities of subsets of items, in the same way the adjective \emph{marginal} is used for the distribution of a subset of random variables in probability theory.
In particular, the one-item inclusion probabilities are given by the diagonal of the kernel, namely $\mathbb{P}(\{i\} \subseteq Y) = \bK_{i,i}$ for all $i \in [N]$.
For all pairs $\{i,j\}$, $|\bK_{i,j}|$ is interpreted as the similarity of $i$ and $j$ -- similar items having a low probability to be jointly sampled; see \cref{e:proba_DPP_pair} below for more details.}
\emph{A priori}, it is not obvious that a given complex matrix $\bK$ defines a DPP. 
\begin{theorem}[\citealp{Mac72}, \citealp{Sos00}]
  \label{th:macchi_soshnikov}
  When $\bK$ is Hermitian, existence of $\mathrm{DPP}(\bK)$ is equivalent to the spectrum of $\bK$ being included in $[0,1]$.
\end{theorem}
\corr{As a first comment, note that if $\bK$ is a Hermitian kernel associated with a DPP, the complex conjugate kernel $\overline{\bK}$ defines a DPP with the same law.
This is because the eigenvalues of any principal submatrix of $\bK$ are real.}
Second, as a particular case of \cref{th:macchi_soshnikov}, when the spectrum of $\bK$ is included in $\{0,1\}$, we call $\bK$ a \emph{projection} kernel, and the corresponding DPP a \emph{projection} DPP. 
Letting $r$ be the number of unit eigenvalues of its kernel, samples from a projection DPP have fixed cardinality $r$ with probability 1 \citep[Lemma 17]{HKPV06}.
In applications, projection kernels of rank $r$ are often available in one of two forms: either
\begin{equation}
  \label{e:projection_kernel_gram_schmidt}
  \bK = \bA (\bA^\adjoint \bA)^{+} \bA^\adjoint,
\end{equation}
where $\bA \in \mathbb{C}^{N \times M}$ is any matrix with rank $r\leq \min(N, M)$,
or in diagonalized form
\begin{equation}
  \label{e:projection_kernel_diagonalized}
  \bK = \bU \bU^\adjoint,
\end{equation}
where $\bU\in \mathbb{C}^{N\times r}$ has orthonormal columns.
We give an example application for each form.

\begin{example}[Uniform spanning trees]
    \label{ex:USTs}
    Consider a finite connected graph $G$ with $M$ vertices and $N$ edges, encoded by its oriented edge-vertex incidence matrix $\bA\in\{-1,0,1\}^{N\times M}$. 
    There are a finite number of spanning trees of $G$, and we draw one uniformly at random. 
    The edges in that random tree correspond to a subset $Y$ of the indices $[N]$ of the rows of $\bA$.
    It turns out \citep{Pemantle91} that $Y$ is a projection DPP with kernel \eqref{e:projection_kernel_gram_schmidt}.
\end{example}

Uniform spanning trees are useful in many contexts, e.g.\ to build preconditioners for certain linear systems \citep[Section 5]{KySo18}. 
Another example of application of DPPs is column-subset selection.

\begin{example}[Column subset selection]
    \label{ex:CSS}
    \cite{BeBaCh20b} propose to select $k$ columns of a ``fat" matrix\footnote{Note our different notation compared to \citep{BeBaCh20b}, who use $N$ for the number of rows.} $\bX \in\mathbb{R}^{n\times N}$, $N\gg n$, using the projection DPP with rank-$k$ kernel 
  \begin{equation}
      \label{e:ayoub_kernel}
      \bK = \bV_{:,[k]}\bV_{:,[k]}^\top,
  \end{equation}
  where $\bX = \bU \mathbf{\Sigma} \bV^\top$ is the singular value decomposition of $\bX$, \corr{and $\bV_{:,[k]}$ is the matrix given by the first $k$ columns of $\bV$}. 
  This is an example of DPP with a kernel specified by \eqref{e:projection_kernel_diagonalized}.
  \cite{BeBaCh20b} prove that the projection of $\bX$ onto the subspace spanned by the selected columns is essentially an optimal low-rank approximation of $\bX$. 
  This ensures statistical guarantees in sketched linear regression.
\end{example}

Because we assume that the kernel is Hermitian, a DPP can be seen as a \emph{repulsive} distribution, in the sense that for all distinct $i,j\in [N]$,
\begin{align}
  \mathbb{P}(\{i,j\} \subseteq Y) &= \bK_{i,i} \bK_{j,j} - \bK_{i,j}\overline{\bK_{i,j}}\label{e:proba_DPP_pair}\\
  &= \mathbb{P}(\{i\} \subseteq Y)\mathbb{P}(\{j\} \subseteq Y) - \vert\bK_{i,j}\vert^2\nonumber\\
  &\leq \mathbb{P}(\{i\} \subseteq Y)\mathbb{P}(\{j\} \subseteq Y).\nonumber
\end{align}

This repulsiveness enforces diversity in samples, and is particularly adequate in applications where a DPP is used to extract a small diverse subset of a large collection of $N$ items.
Beyond column subset selection, this diversity is natural in machine learning tasks such as summary extraction \citep*{KuTa12} or experimental design \citep{DLM20,PoiBa23}.

\subsection{Pfaffian point processes}
\label{s:PfPP_basics}
Similarly to the determinant, the Pfaffian of a $2k\times 2k$  matrix is a polynomial of the matrix entries
    \begin{align*}
        \pf(\mathbf{A}) = \frac{1}{2^k k!}\sum_{\sigma } \mathrm{sgn}(\sigma) \mathbf{A}_{\sigma(1)\sigma(2)}\dots \mathbf{A}_{\sigma(2k-1)\sigma(2k)}. 
    \end{align*}
  It is easy to see that the Pfaffian of $\mathbf{A}$ is equal to the Pfaffian of its antisymmetric part $(\mathbf{A}-\mathbf{A}^\top)/2$.
  For $\mathbf{A}$ skew-symmetric, the definition simplifies to
    \begin{align*}
        \pf(\mathbf{A}) = \sum_{\sigma \text{ contraction}} \mathrm{sgn}(\sigma) \mathbf{A}_{\sigma(1)\sigma(2)}\dots \mathbf{A}_{\sigma(2k-1)\sigma(2k)}. 
    \end{align*}
Recall that a contraction of order $m$ ($m$ even) is a permutation such that $\sigma(1)< \sigma(3) < ... < \sigma(m - 1)$, and $\sigma(2i - 1) < \sigma(2i)$ for $i \leq m/2$.
To relate to determinants, note that the Pfaffian of a skew-symmetric matrix $\mathbf{A}$ of even size satisfies $\det \mathbf{A} = (\pf \mathbf{A})^2$.

\begin{definition}[PfPP]
  Let $\pfK:[N]\times [N] \rightarrow \mathbb{C}^{2\times 2}$ satisfy $\pfK(i,j)^\top = -\pfK(j,i)$ for all $1\leq i,j\leq N$.
  A random subset $Y \subseteq [N]$ is drawn from the PfPP of marginal kernel $\pfK$, denoted by $Y\sim \pfPP(\pfK)$ if and only if
  \begin{equation}\label{eq:def_pfpp}
  \forall S \subseteq [N],\quad \mathbb{P}(S \subseteq Y) = \pf (\bK_S),
  \end{equation}
  where $\bK_S = [\pfK(i,j)]_{i,j\in S}$ is a complex matrix made of $|S|$ blocks of size $2\times 2$. 
\end{definition}
Sufficient conditions on $\pfK$ for the existence of $\pfPP(\pfK)$ were given by \citet[Theorem 1.3]{Kargin2014} when 
$\big(\begin{smallmatrix}
  0 & -1\\
  1 & 0
\end{smallmatrix}\big)\pfK(i,j)$ can be mapped to a self-adjoint quaternionic kernel taking values in the set of $2\times 2$ complex matrices.
Later \citet[Theorem 2.3]{Kassel15} gave an equivalent of the Macchi-Soshnikov Theorem~\ref{th:macchi_soshnikov} for this type of processes; see \citep[Theorem 7.6 and Proposition 7.11]{kaLe22}.
This class of Pfaffian PPs was also studied by \citet{BDQ21} in the continuous setting.

More recently, \citet{koshida21} gave another sufficient condition for the existence of a Pfaffian point process on a discrete ground set, which is well-suited to the processes considered in our paper.
The PfPPs of \citet{koshida21} correspond to the case where $\big(\begin{smallmatrix}
  0 & 1\\
  1 & 0
\end{smallmatrix}\big)\pfK(i,j)$ is a self-adjoint complex kernel.
The intersection of the classes of PfPPs studied by \citet{Kargin2014} and \citet{koshida21} is simply the set of PfPPs for which the $2\times 2$ matrix $\pfK(i,j)$ has a vanishing diagonal, i.e., DPPs with Hermitian kernels; see \cref{ex:vanish_diagonal} below.

Before going further, we introduce a few useful notations. 
Consider a $2N\times 2N$ complex matrix
$
    \bS,
$
viewed as made of four $N\times N$ blocks.
Define the following transformation of $\bS$, called here \emph{particle-hole transformation}, which consists in taking the complex conjugation and exchanging blocks along diagonals, i.e.
\[
    \ph(\bS) =     
    \bF
    \overline{\bS}
    \bF, \text{ with } \bF = 
    \begin{pmatrix}
      \mathbf{0} & \bI\\
      \bI & \mathbf{0}
  \end{pmatrix}.
\]
\begin{proposition}[\cite{koshida21}]\label{prop:existence_Koshida}
  Let $\bS = \left(\begin{smallmatrix}
    \bS_{11} & \bS_{12}\\
    \bS_{21} & \bS_{22}
  \end{smallmatrix}\right)$ be a Hermitian $2N\times 2N$ matrix such that $0 \preceq \bS \preceq \bI$ and $\ph(\bS) = \bI - \bS$.
  There exists a Pfaffian point process with the marginal kernel
  \[
    \pfK_{\corr{\bS}}(i,j) = \begin{pmatrix}
      \bS_{21}(i,j) & \bS_{22}(i,j)\\
      \bS_{11}(i,j) - \delta_{ij} & \bS_{12}(i,j)
    \end{pmatrix}, \quad 1\leq i,j \leq N.
  \]
\end{proposition}
A few remarks are in order. 
First, the properties of $\bS$ allow to simplify the expression of the kernel in
\begin{align}
  \pfK_{\corr{\bS}}(i,j) = \left(\begin{matrix}
    \bS_{21}(i,j) & \bS_{22}(i,j)\\
    -\bS_{22}(j,i) & \overline{\bS_{21}}(j,i)
  \end{matrix}\right), \label{eq:K_pfaffian_kernel}
\end{align}
where $\bS_{21}$ is skew-symmetric and $\bS_{22}$ is Hermitian.
\corr{Second, for a matrix $\bS$ as in \cref{prop:existence_Koshida}, we easily see that $\pfK_{\corr{\bS}}$ and $\pfK_{\corr{\overline{\bS}}}$ yield a PfPP with the \emph{same} law.
This is a consequence of the identity
\begin{align}
  \pfK_{\overline{\bS}}(i,j) = -\begin{pmatrix}
    0 & 1\\
    1 & 0
  \end{pmatrix}
  \pfK_{\bS}(i,j)
  \begin{pmatrix}
    0 & 1\\
    1 & 0
  \end{pmatrix}.\label{e:pf_kernel_conjugate}
\end{align}
When no ambiguity is possible, we suppress the dependence on $\bS$ for simplifying expressions. 
Third}, DPPs with Hermitian kernels appear as particular instances of the PfPPs of \cref{prop:existence_Koshida}.
\begin{example}[vanishing diagonal]\label{ex:vanish_diagonal}
  Let $\bS$ satisfy the conditions of \cref{prop:existence_Koshida}, and let $\mathbb{K}_{\corr{\bS}}$ be the corresponding Pfaffian kernel.
  If $\bS_{21} = 0$, $Y \sim \mathrm{PfPP}(\mathbb{K}_{\corr{\bS}})$ is distributed according to $\mathrm{DPP}(\bS_{22})$.
\end{example}
\corr{Fourth}, for $Y \sim \pfPP(\pfK_{\corr{\bS}})$ and $i\neq j$, the $2$-point correlation function is
  \begin{align}
    \mathbb{P}(\{i,j\} \subseteq Y) &= \bS_{22}(i,i)\bS_{22}(j,j) - |\bS_{22}(i,j)|^2 + |\bS_{21}(i,j)|^2\nonumber\\
    &= \mathbb{P}(\{i\} \subseteq Y)\mathbb{P}(\{j\} \subseteq Y) - |\bS_{22}(i,j)|^2 + |\bS_{21}(i,j)|^2.
    \label{e:2-point-correlation-PfPP}
  \end{align}
Compared with DPPs with Hermitian kernels, Equation~\eqref{e:2-point-correlation-PfPP} suggests that a Pfaffian point process as in Proposition~\ref{prop:existence_Koshida} is less repulsive than the related determinantal process $\mathrm{DPP}(\bS_{22})$  -- an intuition for this fact is given in \cref{s:PfPP_from_QM} below.
Relatedly, note that the $1$-point correlation functions of $\mathrm{DPP}(\bS_{22})$ and  $\pfPP(\pfK_{\corr{\bS}})$ for $\mathbb{K}_{\corr{\bS}}$ given in \cref{eq:K_pfaffian_kernel} are the same. 
In particular, the expected cardinality of $Y\sim \pfPP(\pfK_{\corr{\bS}})$ is simply $\E | Y | = \Tr(\bS_{22})$. 

We end this section with a result about sample parity, which will be further discussed later in \cref{s:PfPP_from_QM}.
\cref{lem:PfPP_parity} below gives the expected parity of the number of samples of a PfPP.
Since we are not aware of any similar statement in the literature, we also provide a short proof of \cref{lem:PfPP_parity}.
\begin{lemma}[Parity of PfPP samples]\label{lem:PfPP_parity}
  Let $\mathbb{K}$ be the kernel of a Pfaffian point process on $[N]$ and let $\mathbb{J}$ be a $2N\times 2N$ block matrix such that
  $
    \mathbb{J}(i,j) 
    =
    \delta_{ij}
     \left(\begin{smallmatrix}
      0 & 1\\
      -1 & 0
    \end{smallmatrix}\right).
  $
  The expected parity of the cardinality of a sample of $Y\sim \pfPP(\pfK)$ is
  $
    \E_{Y}(-1)^{|Y|} = \pf\left(\mathbb{J} - 2 \mathbb{K}\right).
  $
\end{lemma}
\begin{proof}
  We shall prove a slightly more general result, which relies on the identity 
  $$
    \pf(\mathbb{J} + \mathbb{M}) = \sum_{S\subseteq [N]}\pf\left(\mathbf{M}_S\right),
  $$ 
  which holds for any skew-symmetric block matrix $\mathbb{M}$; see \citep[Section 8]{Rains2000}. 
  Note that the term corresponding to $\emptyset$ (and equal to $1$) is included in the sum.
  
  Let $\alpha \neq 1$.
  We have 
  \begin{align*}
    \E_{Y}(\alpha^{|Y|})  = \sum_{S\subseteq [N]}(\alpha -1)^{|S|} \pf(\bK_S) = \sum_{S\subseteq [N]}\pf\left((\alpha -1)\bK_S\right) = \pf\left(\mathbb{J} + (\alpha - 1) \mathbb{K}\right),
  \end{align*}
  where the first equality can be derived from \citep[Sec. 2.3]{koshida21}.
  This the desired result if we take $\alpha = -1$.
\end{proof}
Below, in \cref{rem:projective_S_for_PfPP}, we will see that samples of a PfPP associated with a projective $\mathbf{S}$ have a fixed parity.

%% file: qubits_to_fermions.tex
The content of this section is standard; see e.g., the reference textbook \citep{NiCh10} for quantum computing basics and \citep{Nie05} for the Jordan-Wigner transform. 
We also refer to \citeus{}, which presents all the basic elements required in this section in the context of optical measurements and the resulting point processes.

\subsection{Models in quantum physics}
A quantum model is given by $(i)$ a Hilbert space $(\mathbb{H},\braket{\cdot}{\cdot})$ called the \emph{state space}, and $(ii)$ a collection of self-adjoint operators $\mathbb{H}\to \mathbb{H}$ called \emph{observables}, of which one particular observable $H:\mathbb{H}\rightarrow \mathbb{H}$ is singled out and called the \emph{Hamiltonian}.
Let $\psi \neq 0$ be an element of $\mathbb{H}$. 
All elements of the form $z \psi$ for a complex $z\neq 0$ represent the same quantum state, called a \emph{pure} quantum state, as opposed to more general states to be defined later.
To simplify expressions, it is conventional to only consider elements $\psi$ of unit norm, and to denote a unit-norm pure state by the ``ket'' $\ket{\psi}$, keeping in mind that, as long as $|z|=1$, all vectors $z\ket{\psi}\in\mathbb{H}$ represent the same state.
The corresponding ``bra" $\bra{\psi}$ is the linear form $\ket{x}\mapsto \braket{\psi}{x}$. 

\subsubsection{An observable and a state define a random variable}
Henceforth, we assume that $\mathbb{H}$ has finite dimension $d$. 
Take an observable $A$. 
By the spectral theorem, $A$ can be diagonalized in an orthonormal basis, say with eigenpairs $(\lambda_i, u_i)$ with $1\leq i \leq d$.
For simplicity, we momentarily assume that all the eigenvalues of $A$ have multiplicity $1$.
Together with a state $\ket\psi$, the observable $A = \sum_i \lambda_i u_i u_i ^*$ describes a random variable $X_{A,\psi}$ on $\spec(A) = \{\lambda_1, \dots, \lambda_d\}$, through 
\begin{equation}
    \label{e:born_pure_state}
    \mathbb{P}(X_{A, \psi} = \lambda_i) = \vert\braket{\psi}{u_i}\vert^2.
\end{equation}

When modeling statistical uncertainty on a state, like when describing the noisy output of an experimental device, physical states are not modelled as unit-norm vectors of $\mathbb{H}$, but rather as positive trace-one operators.
To see how, we first map $\ket{\psi}$ to the rank-one projector $\rho = \ketbra{\psi}$.
Then the distribution \eqref{e:born_pure_state} can be equivalently defined as
\begin{equation}
    \label{e:born}
    \mathbb{P}(X_{A, \rho} = \lambda_i) = \tr\left[ \rho 1_{\{\lambda_i\}}(A) \right], 
\end{equation}
where for any $f:\mathbb{R}\rightarrow \mathbb{R}$, we have $f(A) = \sum_i f(\lambda_i) u_i u_i ^*.$
In particular, the expectation of $X_{A, \ketbra{\psi}}$ is $\expval{A}{\psi}$.
Note that Definition \eqref{e:born} generalizes to operators $A$ with eigenvalues with arbitrary multiplicity, and to states $\rho$ beyond projectors. 
In particular, for $\mu$ a probability measure on $\mathbb{H}$, 
\begin{equation}
    \label{e:mixed_state}
    \rho = \E_{\psi\sim\mu} \ketbra{\psi}
\end{equation}
still defines a probability measure on the spectrum of $A$ through \eqref{e:born}. 
The expectation of that distribution, also called the \emph{expectation value} of operator $A$, is $\expval{A}_\rho \triangleq \tr \rho A$.
In physics, the association \eqref{e:born} of a state-observable pair $(\rho, A)$ with the random variable $X_{A,\rho}$ is known as \emph{Born's rule}.

Any $\rho$ that is not a rank-one projector is called a \emph{mixed} state, by opposition to rank-one projectors, which are \emph{pure} states.
Mixed states like \eqref{e:mixed_state} are commonly used to describe any uncertainty in the actual state of an experimental entity.\footnote{Statisticians would say \emph{epistemic} uncertainty, i.e., imprecise knowledge of the state.}

\subsubsection{Commuting observables and a state define a random vector}
When all pairs of a set of observables $A_1, \dots, A_p$ commute, then these observables can be diagonalized in the same orthonormal basis $(u_i)$, and \eqref{e:born} can be naturally generalized to describe a random vector $(X_{A_j,\rho})$ of dimension $p$, with values in the Cartesian product 
$$
    \spec(A_1) \times \dots \times \spec(A_p).
$$
More precisely, the law of $(X_{A_j,\rho})$ is given by
\begin{equation}
    \label{e:vector_born}
    \mathbb{P} (X_{A_1,\rho} = x_1, \dots, X_{A_p,\rho} = x_p) = \tr\left[ \rho 1_{\{x_1\}}(A_1) \dots 1_{\{x_p\}}(A_p) \right],
\end{equation}
see e.g. \citep{BVJ07} for an introduction.
In this paper, we will associate a Pfaffian point process to a particular mixed state and a set of commuting observables, which can respectively be efficiently prepared and measured on a quantum computer.
To define these objects, we first need to explain how physicists build Hamiltonians of fermionic systems.

\subsection{The canonical anti-commutation relations}
The Hamiltonian and its structure are often the key part in specifying a model, much like the factorization of a joint distribution in a probabilistic model.
In the case of fermions, a family of physical particles that includes electrons, Hamiltonians are typically built as polynomials of fermionic creation-annihilation operators, i.e., operators that satisfy the so-called \emph{canonical anti-commutation relations} (CAR).
\begin{definition}[CAR]
    Let $\mathbb{H}$ be a Hilbert space.
    The operators $c_j:\mathbb{H}\rightarrow \mathbb{H}$, $j=1, \dots, p$ and their adjoints are said to satisfy the canonical anti-commutation relations if
    \begin{align}
        \label{e:CAR}
        \tag{CAR}
        \lbrace c_i,c_j \rbrace = \lbrace c^\adjoint_i,c^\adjoint_j \rbrace  = 0
        \qquad \text{and} \qquad
        \lbrace c_i,c^{\adjoint}_j \rbrace = \delta_{ij} \mathbb{I},
    \end{align}
    where $\{u,v\} \triangleq uv+vu$ is the anti-commutator of operators $u$ and $v$.    
\end{definition}

Assuming existence\footnote{\cref{def:JWT} gives a realization of \cref{e:CAR}; see \cref{p:jw}.} for a moment, and limiting ourselves to a finite dimensional Hilbert space $\mathbb{H}$ of dimension $d=2^N$, one can say many things on $\mathbb{H}$ from the fact that there are $N$ operators $c_1, \dots, c_N$ satisfying \eqref{e:CAR}.
On that topic, we recommend reading \citep{Nie05}, from which we borrow the following lemma.

\begin{lemma}[Fock basis; see e.g. \citealp{Nie05}]
    Let $\mathrm{dim} \mathbb{H} = 2^N$, $N\in\mathbb{N}^*$, and assume that $c_1, \dots, c_N$ are distinct operators on $\mathbb{H}$ that satisfy \eqref{e:CAR}.
    First, there is a vector $\ket{\emptyset}\in\mathbb{H}$, called the \emph{vacuum}, which is a simultaneous eigenvector of all $c_i^\adjoint c_i$, $i=1, \dots, N$, always with eigenvalue $0$.
    Second, for $\bn =(n_1, \dots, n_N)\in\{0,1\}^N$, consider 
    $$
    \ket{\bn} \triangleq \prod_{i=1}^N (c_i^\adjoint)^{n_i} \ket{\emptyset}.
    $$
    Then $\mathcal{B}_{\mathrm{Fock}} = (\ket{\bn})$ is an orthonormal basis of $\mathbb{H}$. Third, for all $1\leq i \leq N$,
    \begin{equation}
        c_i^\adjoint c_i \ket{\bn} = n_i \ket{\bn},
        \label{e:nonzero_terms}
    \end{equation}
    and, for $i\neq j$,
    \begin{equation} 
        c_i^\adjoint c_j \ket{\bn} =\pm n_j (1-n_i) \ket{\widetilde\bn},
        \label{e:zero_terms}
    \end{equation}
    where $\widetilde n_i =1$, $\widetilde n_j =0$, and $\widetilde n_k = n_k$ for $k\neq i,j$.
    \label{l:representation_theory}
\end{lemma}
The basis built in the lemma in called the Fock basis. 
Its construction depends on the choice of the operators $c_1, \dots, c_N$. 
When there is a risk of confusion, we shall thus further denote $\ket{\bn}$ and $\ket{\emptyset}$ as $\ket{\bn_c}$ and $\ket{\emptyset_c}$, respectively.
Second, because applying $c_i$ to a vector of the Fock basis has the effect of zeroing the $i$th component if it was $1$, and mapping to zero otherwise, we call the $c_i$'s \emph{annihilation operators}.
Similarly, we call their adjoints \emph{creation operators}.

\subsection{Fermionic operators acting on qubits}
Henceforth, we let $\mathbb{H} \triangleq (\mathbb{C}^2)^{\otimes N}$.
This is the state space describing $N$ qubits, of dimension $2^N$.
A qubit corresponds to any physical system, the state of which is described by one out of two levels.
We associate these two levels with a distinguished orthonormal basis $(\ket{0}, \ket{1})$ of $\mathbb{C}^2$, commonly called the \emph{computational basis}.

Consider the so-called \emph{Pauli operators} on $\mathbb{C}^{2}$, given in the computational basis as
\begin{equation}
	\bI = \begin{pmatrix} 1 & 0\\ 0 & 1\end{pmatrix},\quad
	\sigma_x = \begin{pmatrix} 0 & 1\\ 1 & 0\end{pmatrix}, \quad
	\sigma_y = \begin{pmatrix} 0 & -\i\\ \i & 0\end{pmatrix}, \quad
	\sigma_z = \begin{pmatrix} 1 & 0\\ 0 & -1\end{pmatrix}.
	\label{e:pauli_matrices}
\end{equation}
Note that $\sigma_x^2 = \sigma_y^2 = \sigma_z^2= \bI$, and that $\sigma_x\sigma_y = \i\sigma_z$.

\begin{definition}[{The Jordan-Wigner (JW) transformation}] \label{def:JWT}
    Define the JW annihilation operators $a_j$ on $\mathbb{H}$, for $j\in\{1, \dots N\}$, by their matrices in the computational basis
    \begin{equation}
        \label{e:jordan_wigner}
        \underbrace{\sigma_z \otimes \dots \otimes \sigma_z}_{j-1\text{ times}} \otimes \left(\frac{\sigma_x+\i \sigma_y}{2}\right) \otimes \bI \otimes \dots \otimes \bI.
    \end{equation}
    The creation operator $a_j^*$, adjoint of $a_j$, has matrix 
    $$
        \underbrace{\sigma_z \otimes \dots \otimes \sigma_z}_{j-1\text{ times}} \otimes \left(\frac{\sigma_x-\i \sigma_y}{2}\right) \otimes \bI \otimes \dots \otimes \bI.
    $$
    \label{d:jw}
\end{definition}

It is easy to check the following properties of the Jordan-Wigner operators.
\begin{proposition}
    \label{p:jw}
    Let $a_1, \dots, a_N$ be the JW operators from~\cref{d:jw}.
    They satisfy \eqref{e:CAR}. 
    Moreover, the so-called \emph{number} operators $N_j= a_j^\adjoint a_j$ have the simple matrix form
    $$ 
        \underbrace{\bI \otimes \dots \otimes \bI}_{j-1\text{ times}} \otimes \frac{\bI - \sigma_z}{2}  \otimes \bI \dots\otimes \bI.
    $$
    Finally, the Fock basis $\ket{\bn}$ obtained by setting $c_j = a_j$ in~\cref{l:representation_theory} is actually the computational basis, i.e.,
    $$
    \ket{\bn} = \ket{n_1}\otimes \dots \otimes \ket{n_N}.
    $$
\end{proposition}
There are alternative ways to define fermionic operators on $\mathbb{H}$ using Pauli matrices, like the parity or Bravyi-Kitaev transformations \citep{BraKit2002}; or the Ball-Verstraete-Cirac transformation \citep{Ball2005,VerCir2005}. 
In particular, and at the expense of simplicity, the Bravyi-Kitaev transformation yields operators with smaller \emph{weight} than Jordan-Wigner: only $\mathcal{O}(\log N)$ terms are allowed to differ from the identity in the equivalent of \eqref{e:jordan_wigner}.

%% file: fermions_to_dpps.tex
We first show in \cref{s:DPP_from_QM} how to build any discrete DPP with Hermitian kernel from a mixed state corresponding to a particle number preserving Hamiltonian and a set of commuting observables. 
This connection between DPPs and quasi-free states was observed by
\citet{olshanski2020}.
Then we do the same for a class of discrete PfPPs in \cref{s:PfPP_from_QM} associated to Hamiltonians without particle number conservation.

\subsection{Building a DPP from a quantum measurement} \label{s:DPP_from_QM}
Let $\bH\in\mathbb{C}^{N\times N}$ be Hermitian, and consider the operator 
\begin{align}
H &= \sum_{i,j=1}^N c_i^\adjoint \bH_{ij} c_j,\label{eq:quad_Hamiltonian_vanilla}
\end{align}
which we think of as a Hamiltonian acting on $\mathbb{H}$.
This Hamiltonian preserves the particle number, i.e., $H$ commutes with the number operator $\sum_i c_i^\adjoint c_i$.
Let $(\lambda_k, \bu_{k})_{1\leq k \leq N}$ be the eigenpairs of $\bH$ where  the eigenvalues satisfy $\lambda_1 \leq \dots \leq \lambda_N$. 
The Hamiltonian \cref{eq:quad_Hamiltonian_vanilla} can thus be rewritten in ``diagonalized" form
\begin{align*}
H = \sum_{i,j=1}^N  c_i^\adjoint \left( \sum_{k=1}^N \lambda_k \mathbf u_{k} \mathbf u_k^*\right)_{ij} c_j= \sum_{i,j,k=1}^N  \lambda_k c_i^\adjoint \mathbf u_{ki} \mathbf u_{kj}^* c_j= \sum_{k=1}^N  \lambda_k b_k^\adjoint b_k,
\end{align*}
where we defined the operators 
\begin{align}
    b_k = \sum_{j=1}^N \mathbf u_{kj}^* c_j,
    \label{e:new_annihilation_operators}
\end{align}
which can be checked to satisfy \cref{e:CAR}.

In what follows, we consider in \cref{s:L-ensemble} DPPs for which the correlation kernel has a spectrum in $[0,1)$, referred to as \emph{L-ensembles} in the literature.
Next, in \cref{s:projDPP}, we discuss projection DPPs.

\subsection{The case of L-ensembles}
\label{s:L-ensemble}

Below, \cref{prop:DPP_from_rho} states that there is a natural DPP associated with a Hamiltonian like \eqref{eq:quad_Hamiltonian_vanilla}, whose marginal kernel $\mathbf{K}$ is obtained by applying a sigmoid to the spectrum of $\bH$; formally $\mathbf{K} = \sigma(-\beta \bH)$ with $\beta >0$ being a parameter called \emph{inverse temperature}.
\begin{proposition}[DPP kernel by taking the sigmoid]\label{prop:DPP_from_rho}
    Let $\beta>0$ and $\mu \geq 0$, $H$ and $(b_k)$ be respectively defined by \eqref{eq:quad_Hamiltonian_vanilla} and \eqref{e:new_annihilation_operators}.
    Consider the mixed state
    \begin{align}
        \label{e:gaussian}
        \rho =\frac{1}{Z} \e^{ -\beta \left( H - \mu \sum_{j=1}^N b_j^\adjoint   b_j \right) },
    \end{align} 
    where the normalization constant $Z$ ensures that $\tr \rho = 1$.
    For $i\in[N]$, the observable $N_i=c_i^\adjoint c_i$ is a projector, so that the random variable $X_{N_i,\rho}$ associated to $N_i$ and the state \eqref{e:gaussian} by \eqref{e:born} takes values in $\{0,1\}$.    
    Moreover, all $N_i$'s commute pairwise, and thus define a joint Boolean vector $(X_{N_i,\rho})_{i\in[N]}$ through \eqref{e:vector_born}.
    Consider the point process 
    $$
        S = \{i\in[N]: X_{N_i,\rho}=1\},
    $$
    corresponding to jointly observing all operators $N_i$ in the state \eqref{e:gaussian}, and reporting the indices corresponding to a $1$.
    Then $S$ is determinantal with correlation kernel 
    $$
        \bK = \bU\, \Diag\left( \frac{\e^{-\beta (\lambda_k-\mu)}}{1 + \e^{-\beta (\lambda_k-\mu)}}\right)\, \mathbf U^\adjoint,
    $$
    where $\bU$ and $(\lambda_k)_k$ are obtained by the diagonalization $\bH = \bU\, \Diag\left( \lambda_k\right) \bU^\adjoint$.
\end{proposition}
\begin{proof}
    All number operators $N_i = c_i^\adjoint c_i$ commute pairwisely and have spectrum in $\{0,1\}$; see~\cref{l:representation_theory}.
    Consequently, their joint measurement indeed describes a random binary vector $X=(X_{N_i, \rho})$. 
    By Born's rule \eqref{e:vector_born}, the correlation function of the point process corresponding to the indices of the $1$s in $X$ are given by
    $$
        \mathbb{P} \left ( \{i_1, \dots, i_k\}\subseteq X \right ) = \tr \left(\rho N_{i_1}\dots N_{i_k}\right) \triangleq \expval{N_{i_1}\dots N_{i_k}}_\rho.
    $$
    Because of the anti-commutation relations \eqref{e:CAR}, an explicit computation known as Wick's theorem implies that for any $i_1,\dots,i_k\in [N]$,
    \begin{equation}
        \label{e:wick}
        \expval{N_{i_1}\dots N_{i_k}}_\rho = \det \left(\expval{c_{i_m}^\adjoint c_{i_n}}_\rho\right)_{m,n}.
    \end{equation}
    Wick's theorem is a standard result in quantum physics; see e.g.\ \citep[Section 3.5]{BFBDHRRSW22Sub} for a rewriting with our notations of one of the canonical references \citep{CTDL77}.

    Now, Equation \eqref{e:wick} implies that the point process $X$, consisting of simultaneously measuring the occupation of all qubits using $(N_i)$, is determinantal with correlation kernel
    \begin{align}
        \mathbf{K}_{ij} &= \expval{c_{i}^\adjoint c_{j}}_\rho.
        \label{e:kernel_intermediate}
    \end{align}
    In order to provide an explicit computation of $\mathbf{K}_{ij}$, we use~\cref{l:representation_theory} with the operators $b_1, \dots, b_N$, to obtain a basis $(\ket{\bn}) = (\ket{\bn}_b)$ of eigenvectors of all operators of the form $b_i^\adjoint b_i$, $i=1, \dots, N$.
    We now proceed to computing the trace in \eqref{e:kernel_intermediate} by summing over that basis.
    We write
    \begin{align*}
        \mathbf{K}_{ij} &= \tr \left[\rho c_{i}^\adjoint c_{j}\right] \\
        &= \tr \left[\rho c_{i}^\adjoint c_{j} \sum_{\bn} \ketbra{\bn}\right]\\
        & = Z^{-1} \sum_{\bn} \expval{\e^{-\beta \sum_{p} (\lambda_p-\mu) b_p^\adjoint b_p} c_i^\adjoint c_j}{\bn}.
    \end{align*}
    Now note that $c_j = \sum_{\ell=1}^N \bu_{j\ell}b_\ell$, so that $c_i^\adjoint c_j = \sum_{k,\ell=1}^N \bu^*_{ik} \bu_{j\ell}  b_k^\adjoint b_\ell$, and
    \begin{align*}
        \mathbf{K}_{ij} &= Z^{-1}\sum_{k,\ell=1}^N \bu^*_{ik} \bu_{j\ell}  \sum_{\bn} \expval{\e^{-\beta \sum_{p} (\lambda_p-\mu) b_p^\adjoint b_p} b_k^\adjoint b_\ell}{\bn}.
    \end{align*}
    Because of \eqref{e:zero_terms} and the orthonormality of the basis $(\ket{\bn})$, for $k\neq \ell$,
    $$ 
    \expval{\e^{-\beta \sum_{p} (\lambda_p-\mu) b_p^\adjoint b_p} b_k^\adjoint b_\ell}{\bn}=0.
    $$ 
    Therefore,
    \begin{align}
        \mathbf{K}_{ij} &= Z^{-1}\sum_{k=1}^N \bu^*_{ik} \bu_{jk}  \sum_{\bn} \expval{\e^{-\beta \sum_{p} (\lambda_p-\mu) b_p^\adjoint b_p} b_k^\adjoint b_k}{\bn} \label{e:kernel}\\
        &= Z^{-1}\sum_{k=1}^N \bu^*_{ik} \bu_{jk}  \sum_{\bn} n_k \e^{-\beta \sum_p (\lambda_p-\mu)n_p}. \label{e:kernel_new_intermediate}
    \end{align}
    Now, we pause and compute the normalization constant $Z$ of $\rho$.
    Starting from $\tr\rho=1$, we write
    \begin{align*} 
        Z = \tr \e^{-\beta \sum_{p} (\lambda_p-\mu) b_p^\adjoint b_p}
        = \sum_{\bn} \expval{\e^{-\beta \sum_{p} (\lambda_p-\mu) b_p^\adjoint b_p}}{\bn}
        = \sum_{\bn} \e^{-\beta \sum_p (\lambda_p-\mu)n_p}.
    \end{align*}
    Rewriting the sum as a product, we obtain
    \begin{align}
        Z &= \prod_{p=1}^N \sum_{n_p \in \{0,1\}} \e^{-\beta (\lambda_p-\mu)n_p} = \prod_{p=1}^N (1+ \e^{-\beta (\lambda_p-\mu)}).
        \label{e:constant_intermediate}
    \end{align}
    Now, explicitly writing the dependence of $Z=Z(\blambda)$ to $\blambda = (\lambda_1, \dots, \lambda_N)$, Equations~\eqref{e:kernel_new_intermediate} and \eqref{e:constant_intermediate} together yield 
    \begin{align}
        \mathbf{K}_{ij} &= \sum_{k=1}^N \bu^*_{ik} \bu_{jk}  \frac{\sum_{\bn} n_k \e^{-\beta \sum_p (\lambda_p-\mu)n_p}} {\sum_{\bn} \e^{-\beta \sum_p (\lambda_p-\mu)n_p}}\\
        &= \sum_{k=1}^N \bu^*_{ik} \bu_{jk} \left[-\frac{1}{\beta}\frac{\d}{\d\lambda_k} \log Z(\blambda)\right]\\
        &= \sum_{k=1}^N \bu^*_{ik} \bu_{jk} \frac{\e^{-\beta (\lambda_k-\mu)}}{1 + \e^{-\beta (\lambda_k-\mu)}},
        \label{e:kernel_final}
    \end{align}
    for any $1\leq i,j \leq N$.
    \corr{Thus, we find $\bK = \overline{\bU}\, \Diag\left( \frac{\e^{-\beta (\lambda_k-\mu)}}{1 + \e^{-\beta (\lambda_k-\mu)}}\right)\, \overline{\mathbf U}^\adjoint$. 
    Since  $\overline{\bK}$ and $\bK$ define the same DPP, we have the desired result.}
\end{proof}

\begin{corollary}[Hamiltonian by taking the logit]\label{c:sampling_a_DPP}
    Let $\bK = \mathbf{V} \mathbf{D} \mathbf{V}^\adjoint$ be a Hermitian DPP kernel with $0\prec \mathbf{D}\prec \bI$. 
    Then Proposition~\ref{prop:DPP_from_rho} yields a DPP with kernel $\mathbf{K}$ provided we choose $\mathbf{V} = \bU$ and the eigenvalues $\blambda$ so that
    $$
        \beta(\lambda_i-\mu) = \log \frac{1-d_i}{d_i}.
    $$
    Furthermore, assuming $d_1 \geq \dots \geq d_N \geq 0$, and $d_{r} < \mu < d_{r+1}$, $\bK$ converges to the projection kernel onto the first $r$ columns of $\mathbf{V}$. 
\end{corollary}
\begin{proof}
    The first part is a consequence of \eqref{e:kernel}. 
    The second statement is maybe easiest to see in terms of Frobenius norm, to wit $\Vert \mathbf{A}\Vert^2 = \sum_{i=1}^N \sigma_i^2(\mathbf{A})$, where $\sigma_i(\mathbf{A})$ are the singular values of $\mathbf{A}$.
    If $\mathbf{P}$ denotes the projector into the first $r$ columns of $\mathbf{V}$, then the absolute values of the eigenvalues of the Hermitian matrix $\mathbf{K}-\mathbf{P}$ are its singular values, and all of them converge to $0$.
\end{proof}

\subsection{The case of projection DPPs}\label{s:projDPP}

With the notation of \cref{s:L-ensemble}, we note that $(\ket{\bn}) = (\ket{\bn}_b)$ is a basis of eigenvectors of $\rho$, and that
$$
\rho\ket{\bn} = \frac{\e^{-\beta \sum_p (\lambda_p - \mu) n_p}}{\sum_{\bn'} \e^{-\beta \sum_p (\lambda_p - \mu)n'_p}} \ket{\bn}.
$$ 
In particular, the only eigenvalue that does not vanish when $\beta\rightarrow \infty$ is that of $\ket{1,\dots, 1, 0, \dots, 0}$, where the $1$s occupy the first $r$ components, and $r\in\{1, \dots, N\}$ is such that with $\lambda_{r} < \mu < \lambda_{r+1}$. 
Indeed, the ratio of the eigenvalue of $\ket{1,\dots, 1, 0, \dots, 0}$ with that of any other eigenvector diverges to $+\infty$, while all eigenvalues remain in $[0,1]$ and the trace is fixed to $1$.
Thus, denoting $\ket{\psi} = \ket{1,\dots, 1, 0, \dots, 0}$, 
$$
\rho \rightarrow \ketbra{\psi}
$$
in Frobenius norm as $\beta\rightarrow \infty$. 
In particular, 
$$
\tr N_{i_1} \dots N_{i_k} \rho \rightarrow \tr N_{i_1} \dots N_{i_k} \ketbra{\psi}
$$
as $\beta\rightarrow \infty$. 
Combined with the second statement of \cref{c:sampling_a_DPP}, we know that the correlation functions of the point process corresponding to preparing $\ketbra{\psi}$ and measuring the occupation of all qubits and recording where the $1s$ occur is the projection DPP of kernel $\mathbf{P} = \mathbf{V}_{:,[r]} \mathbf{V}^*_{:,[r]}$.
\corr{Recall that $\mathbf{V}_{:,[r]}$ is the matrix obtained by selecting the $r$ first columns of $\mathbf{V}$.}
\begin{remark}[Beyond Gaussian states]
    We have shown that the kernels of projection DPPs are obtained as limits of kernels associated with Gaussian states \cref{e:gaussian}.
    Actually, a DPP kernel can be associated with a density matrix of a quasi-free state generalizing the Gaussian state, for which Wick's theorem also holds; see e.g.\ \citep[Section 1.3]{koshida21}.
    In particular, every quasi-free state is a convex combination of pure states. 
    We give now a few details by following \citep{DFP08}.
    Let $\mathcal{C} \subseteq [N]$ and let  $\ket{\psi_\mathcal{C}} = \prod_{i\in \mathcal{C}} c_i^\adjoint \ket{\emptyset}$.
    Let $\bK$ be a Hermitian matrix with eigenvalues $(\nu_p)$ in $[0,1]$.
    The density matrix associated with $\bK$ is the quasi-free state
    \[
        \rho = \sum_{\mathcal{C}\subseteq [N]} \alpha_\mathcal{C} \ketbra{\psi_\mathcal{C}} \text{ with } \alpha_\mathcal{C} = \prod_{p\in \mathcal{C}}\nu_p \prod_{q\in [N]\setminus\mathcal{C}}(1-\nu_q). 
    \]
    \corr{Furthermore,  in \cref{a:proof_without_wick}, we establish for the interested reader the determinantal formula for $\expval{N_{i_1}\dots N_{i_k}}_\rho$ for a projective $\rho$ without resorting to Wick's theorem as in the proof of \cref{prop:DPP_from_rho}}.
\end{remark} 
\subsection{Building a PfPP from the BdG Hamiltonian} \label{s:PfPP_from_QM}
We shall see that Pfaffian point processes appear using slightly more sophisticated Hamiltonians. 
Unlike Hamiltonian \eqref{eq:quad_Hamiltonian_vanilla}, the Hermitian operator
\begin{align}
    H = \frac{1}{2}\sum_{i,j=1}^N \mathbf{M}_{ij} (c_i^\adjoint c_j - c_j c_i^\adjoint) + \frac{1}{2}\sum_{i,j=1}^N \left( \bm{\Delta}_{ij} c_i^\adjoint c_j^\adjoint + \bm{\Delta}^*_{ij} c_i c_j\right), \label{e:H_general}
\end{align}
with complex matrices $\mathbf{M} = \mathbf{M}^\adjoint$ and $\bm{\Delta}$, does not commute with the total number operator $N = \sum_i c_i^* c_i$.
Physicists say that $H$ does not ``preserve" the total number of particles.

Note that, due to the anti-commutation relation $c_i^\adjoint c_j^\adjoint = - c_j^\adjoint c_i^\adjoint$,  we can simply redefine $\bm{\Delta}$ so that $\bm{\Delta} = - \bm{\Delta}^\top$.
The Hamiltonian~\eqref{e:H_general} becomes the so-called Bogoliubov-de Gennes (BdG) Hamiltonian
\begin{align}
    H_{\mathrm{BdG}} = \frac{1}{2}\sum_{i,j=1}^N \mathbf{M}_{ij} (c_i^\adjoint c_j - c_j c_i^\adjoint) + \frac{1}{2}\sum_{i,j=1}^N \left( \bm{\Delta}_{ij} c_i^\adjoint c_j^\adjoint - \overline{\bm{\Delta}_{ij}} c_i c_j\right), \label{e:H_BdG}
\end{align} 
which is a model for superconductors; see e.g.\ \citep{SRFL08} for a modern overview.
\begin{remark}[Fermion parity]\label{rem:fermion_parity}
    \corr{The Hamiltonian $H_{\mathrm{BdG}}$ commutes with the parity operator $(-1)^{\sum_{i=1}^N c_i^* c_i}$; see e.g.\ \citep{grabsch2019pfaffian}. 
    Therefore, the eigenfunctions of $H_{\mathrm{BdG}}$ are also eigenfunctions of the parity operator.}
\end{remark}
We now investigate the point process associated to the occupation numbers of the Gaussian state
$
     Z^{-1} \exp(-\beta H_{\mathrm{BdG}}), 
$
for  $\beta >0$, as we did for \eqref{e:gaussian}.
It is convenient to stack the operators $c_1, \dots, c_N$ and their adjoints in column vectors, and write $\mathbf{c} = [c_1 \dots c_N ]^\top$ and $\mathbf{c}^\adjoint = [c_1^\adjoint \dots c_N^\adjoint]^\top$.
We can then rewrite the Hamiltonian more compactly as
\begin{align}
    H_{\mathrm{BdG}} = \frac{1}{2}\Big(\begin{smallmatrix}
        \mathbf{c}^\adjoint\\
        \mathbf{c}
    \end{smallmatrix}\Big)^\adjoint \mathbf{H}_{\mathrm{BdG}}
    \Big(\begin{smallmatrix}
        \mathbf{c}^\adjoint\\
        \mathbf{c}
    \end{smallmatrix}\Big)
     \text{ with } \mathbf{H}_{\mathrm{BdG}} = \begin{pmatrix}
        -\overline{\mathbf{M}} & -\overline{\bm{\Delta}}\\
        \bm{\Delta} & \mathbf{M}
    \end{pmatrix}.\label{eq:H_matrix}
\end{align}
By construction, the matrix $\bH_{\mathrm{BdG}}$ obeys the \emph{particle-hole symmetry}, namely 
$
     \bF\overline{\bH}_{\mathrm{BdG}}\bF = -\bH_{\mathrm{BdG}}
$
where we recall the definition of the involution
$
\bF = \big(\begin{smallmatrix}
    \bm{0} &\bI\\
    \bI & \bm{0}
\end{smallmatrix}\big)
$.
The name \emph{particle-hole symmetry} comes from the fact that 
$$
    \bF\Big(\begin{smallmatrix}
    \mathbf{c}^\adjoint\\
    \mathbf{c}
\end{smallmatrix}\Big) = \Big(\begin{smallmatrix}
    \mathbf{c}\\
    \mathbf{c}^\adjoint
\end{smallmatrix}\Big)
$$
flips the role of creation and annihilation operators.

Before going further, we introduce a group of transformations preserving \cref{e:CAR} that are used to diagonalize the Bogoliubov-de Gennes Hamiltonian; see \citep[Section 18.4.3]{Moore2014}.
\begin{definition}[orthogonal complex matrix]\label{def:orthogonal_complex}
     A complex invertible $2N\times 2N$ matrix $\mathbf{W}$ satisfying $\mathbf{W}^\top \bF \mathbf{W} = \bF$ is called an \emph{orthogonal complex matrix}.
     The group of orthogonal complex matrices is denoted by $O(\bF,\mathbb{C})$.
\end{definition}
For any orthogonal complex matrix $\mathbf{W}$ and any $N$ operators $c_1, \dots, c_N$ satisfying \cref{e:CAR}, another set of creation-annihilation operators satisfying \cref{e:CAR} is given by the so-called \emph{Bogoliubov transformation}
\begin{align}
    \Big(\begin{smallmatrix}
            \bm{b}^\adjoint\\
            \bm{b}
        \end{smallmatrix}\Big)
        = \bm{W}         \Big(\begin{smallmatrix}
            \mathbf{c}^\adjoint\\
            \mathbf{c}
        \end{smallmatrix}\Big),\label{e:Bog_transform}
\end{align}
and by further requiring that $\mathbf{W}$ is \emph{unitary},  $b_k^\adjoint$ is the adjoint of $b_k$ for all $ k \in [N]$.
Hence, in what follows, we only consider transformations $\mathbf{W}\in O(\bF,\mathbb{C})\cap U(2N)$.
%
\begin{lemma}\label{lem:unitary_Bog}
    A unitary orthogonal complex matrix is of the form 
        $   \mathbf{W}
            =
            \big(\begin{smallmatrix}
                \mathbf{U} & \mathbf{V} \\
                \overline{\mathbf{V}} & \overline{\mathbf{U}}
            \end{smallmatrix}\big),
        $
        with $\mathbf{U} \mathbf{U}^\adjoint + \mathbf{V} \mathbf{V}^\adjoint = \bI$ and $\mathbf{U} \mathbf{V}^\top + \mathbf{V} \mathbf{U}^\top = \bm{0}$.
        Furthermore, $\det \mathbf{W} \in \{-1,1\}$.
\end{lemma}
Next, following \citet{JSKSB18}, we diagonalize \eqref{e:H_BdG} by a  convenient change of variables, which consists in finding a suitable Bogoliubov transformation.
The upshot is that we have the decomposition
\begin{equation}
    \mathbf{H}_{\mathrm{BdG}} = \mathbf{W}^*     \begin{pmatrix}
        \Diag(-\epsilon_k) & 0\\
        0 & \Diag(\epsilon_k)
    \end{pmatrix}
    \mathbf{W},\label{e:eigen_decomp_HBDG}
\end{equation}
as shown in \cref{lem:diag_H_BdG} below.
\begin{lemma}[Diagonalization of BdG Hamiltonian]\label{lem:diag_H_BdG}
    Let $\bm{\Omega} = \frac{1}{\sqrt{2}} \Big(\begin{smallmatrix}
        \bI & \bI\\
        \i \bI & -\i \bI
    \end{smallmatrix}\Big)$ and define $\mathbf{A} = -\i \overline{\bm{\Omega}} \Big(\begin{smallmatrix}
        \bm{\Delta} & \mathbf{M}\\
        -\overline{\mathbf{M}} & -\overline{\bm{\Delta}} 
    \end{smallmatrix}\Big) \bm{\Omega}^\adjoint$.
    The following statements hold.
    \begin{enumerate}
        \item[(i)]
        $
            H_{\mathrm{BdG}} = \frac{1}{2} \i\bm{f}^\top \mathbf{A} \bm{f}
        $
        where $\bm{f} = \bm{\Omega} \Big(\begin{smallmatrix}
            \mathbf{c}^\adjoint\\
            \mathbf{c}
        \end{smallmatrix}\Big)$ and $\mathbf{A}$ is real skew-symmetric. 
        \item[(ii)] There exists a real orthogonal matrix $\bm{R}$ and a real vector $\bm{\epsilon} = [\epsilon_1, \dots, \epsilon_N]^\top$ \corr{such that $0\leq \epsilon_1\leq  \dots\leq \epsilon_N$} and
        $
            \bm{R}\mathbf{A}\bm{R}^\top = 
            \Big(\begin{smallmatrix}
                0 & \Diag(\epsilon_k)\\
                -\Diag(\epsilon_k) & 0
            \end{smallmatrix}\Big).
        $
        \item[(iii)]
        Another set of creation-annihilation operators satisfying \cref{e:CAR} is given by 
        $
        \Big(\begin{smallmatrix}
                \bm{b}^\adjoint\\
                \bm{b}
            \end{smallmatrix}\Big)
            = \bm{W}         \Big(\begin{smallmatrix}
                \mathbf{c}^\adjoint\\
                \mathbf{c}
            \end{smallmatrix}\Big)
        $,
        where $\bm{W} = \bm{\Omega}^\adjoint \bm{R} \bm{\Omega} \in O(\bF,\mathbb{C}) \cap U(2N)$.
        \item[(iv)] We have the diagonalization
        $
            H_{\mathrm{BdG}} = \frac{1}{2}     
            \Big(\begin{smallmatrix}
                \mathbf{b}^\adjoint \\
                \mathbf{b}
            \end{smallmatrix}\Big)^\top
            \Big(\begin{smallmatrix}
                \mathbf{0} & \Diag(\epsilon_k)\\
                -\Diag(\epsilon_k) & \mathbf{0}
            \end{smallmatrix}\Big)
            \Big(\begin{smallmatrix}
                \mathbf{b}^\adjoint \\
                \mathbf{b}
            \end{smallmatrix}\Big).
        $
    \end{enumerate}
\end{lemma}
We refer to \citep{JSKSB18} for a proof sketch.

Now, we leverage these results to compute the expectation value of bilinears under $\rho_{\mathrm{BdG}}$.
\cref{lem:bilinears_BdG} states a result analogous to~\cref{prop:DPP_from_rho}, and its proof, given in \cref{proof:bilinears_BdG}, relies on similar techniques such as Wick's theorem. 
\begin{lemma} \label{lem:bilinears_BdG}
    Let $\bm{\epsilon} = [\epsilon_1, \dots, \epsilon_N]^\top$ be the eigenvalues of $H_{\mathrm{BdG}}$ \corr{such that $0\leq \epsilon_1\leq  \dots\leq \epsilon_N$} as given by~\cref{lem:diag_H_BdG} and let 
    $
        \rho  = \frac{1}{Z} \exp( -\beta H_{\mathrm{BdG}}).
    $
    We have
    \begin{align*}
        \mathbf{S} 
        \triangleq
        \begin{pmatrix}
            (\left\langle c_i c^\adjoint_j \right\rangle_{\rho})_{i,j} &(\left\langle c_i c_j\right\rangle_{\rho})_{i,j} \\
            (\left\langle c^\adjoint_i c^\adjoint_j\right\rangle_{\rho})_{i,j} & (\left\langle c^\adjoint_i c_j\right\rangle_{\rho})_{i,j}
        \end{pmatrix} 
        = \overline{\mathbf{W}}^{\adjoint}
        \begin{pmatrix}
             \Diag(\sig(\beta \epsilon_k))& \mathbf{0}\\
            \mathbf{0} & \Diag(\sig(-\beta \epsilon_k))
        \end{pmatrix}
        \overline{\mathbf{W}},
    \end{align*}
    where $\sigma(x)$ is the sigmoid function; see Section~\ref{s:introduction}.
    Furthermore, we have that $\bS = \sigma(-\beta \overline{\mathbf{H}_{\mathrm{BdG}}})$ and $\bS$ satisfies $0\preceq \mathbf{S} \preceq \bI$, as well as $\mathbf{C}\overline{\mathbf{S}}\mathbf{C} = \bI - \mathbf{S}$.
\end{lemma}
\corr{Note that the process $\pfPP(\mathbb{K}_{\bS})$ with the $\bS$ matrix given in \cref{lem:bilinears_BdG} has the same law as $\pfPP(\mathbb{K}_{\overline{\bS}})$ in the light of the behaviour of this type of Pfaffian kernels under complex conjugation; see \cref{e:pf_kernel_conjugate}.}

For convenience, we introduce the following notations, for $1\leq i,j\leq N$,
\[
        \begin{pmatrix}
            \left\langle c_i c^\adjoint_j\right\rangle_{\rho} & \left\langle c_i c_j\right\rangle_{\rho} \\
            \left\langle c^\adjoint_i c^\adjoint_j\right\rangle_{\rho}& \left\langle c^\adjoint_i c_j\right\rangle_{\rho}
        \end{pmatrix} 
        \triangleq
        \begin{pmatrix}
            \delta_{ij}-\mathbf{K}^T_{ij} &-\overline{\mathbf{P}_{ij}} \\
            \mathbf{P}_{ij} &\mathbf{K}_{ij} 
        \end{pmatrix},
\]
where  $\mathbf{0} \preceq \mathbf{K}\preceq \bI$ is Hermitian and 
$\mathbf{P}$ is skew-symmetric.
By definition, we have 
\[  
    \mathbf{S}
    = 
    \begin{pmatrix}
         \bI -\overline{\mathbf{K}}& \mathbf{P}^\adjoint\\
        \mathbf{P} & \mathbf{K}
    \end{pmatrix}.
\]
The particle-hole transformation amounts to replacing in $\mathbf{S}$ each $c_i$ by $c_i^\adjoint$ and vice-versa.
As a consequence of \eqref{e:CAR}, $\mathbf{S}$ satisfies $\mathbf{C}\overline{\mathbf{S}}\mathbf{C} = \bI - \mathbf{S}$; see \cref{s:PfPP_basics}.

\Cref{prop:Pfaffian_PP} can be seen as a generalization to mixed states of the analysis of \citet{TerDiVi02}, and restates the results of \citet{koshida21}.

\begin{proposition}\label{prop:Pfaffian_PP}
    Let $N_i = c_i^\adjoint c_i$ for $1\leq i\leq N$ and let $\rho  = Z^{-1} \exp( -\beta H_{\mathrm{BdG}})$.
    For any $i_1,\dots,i_k\in [N]$, we have
    \begin{equation}\label{eq:correlation_Pfaffian}
        \expval{N_{i_1}\dots N_{i_k}}_{\rho} = \pf \left(\pfK(i_m,i_n)\right)_{1\leq m,n \leq k},
    \end{equation}
    where each block of the above $2k \times 2k$ matrix is given in terms of the $2\times 2$-matrix-valued kernel
    \begin{align}
        \pfK(i,j) = 
        \begin{pmatrix}
            \mathbf{P}_{i j} & \mathbf{K}_{i j}\\
            - \mathbf{K}_{j i} & -\overline{\mathbf{P}_{i j}}
        \end{pmatrix}. \label{eq:Pfaffian_kernel}
    \end{align}
    The latter satisfies $\pfK(i,j)^\top = -\pfK(j,i)$ for $1\leq i,j \leq N$.
\end{proposition}
The proof of this result, given in \cref{proof:Pfaffian_PP}, is also based on Wick's theorem.

More can also be said about sample parity.
\begin{proposition}[Sample parity \corr{and Majorana quadratic form}]\label{prop:Pfaffian_sample_parity}
    In the setting of \cref{lem:bilinears_BdG}, for $Y \sim \pfPP(\mathbb{K})$, \corr{we have that}
    \begin{equation}
        \E_Y (-1)^{|Y|} = \det(\mathbf{W})\prod_{k = 1}^N \tanh(\beta \epsilon_k/2).\label{e:expected_parity_tanh}
    \end{equation}
    Assuming that $\epsilon_k >0$ for all $k$, an alternative expression for the expectation of the parity is obtained by using the identity
    $
        \det \mathbf{W} = \sign \pf(\mathbf{A}_{M}).
    $
    Here, the skew-symmetric matrix $\mathbf{A}_{M}$ is the quadratic form of the Majorana representation of the Hamiltonian, namely
    \[
    H_{\mathrm{BdG}} = \frac{\i}{2} \bm{\gamma}^\top \mathbf{A}_{M} \bm{\gamma},
    \]
    where $\gamma_{2k-1} = \frac{1}{\sqrt{2}}(c_k^*+c_k)$ and $\gamma_{2k} = \frac{\i}{\sqrt{2}}(c^*_k - c_k)$ for $k\in [N]$.
\end{proposition}
The proof of \cref{prop:Pfaffian_sample_parity} given in \cref{proof:Pfaffian_sample_parity} relies on the expected parity formula of \cref{lem:PfPP_parity}.
Incidentally, note that this result implies Kitaev's formula \citep[Eq.\ (19)]{Kitaev01} for the parity of a one-dimensional chain of fermions in a limit case that we discuss in \cref{rem:projective_S_for_PfPP}\corr{; see also \citep{grabsch2019pfaffian} for a detailed study of Pfaffian formulae for fermion parity}.

%% file: quantum_circuits.tex
Armed with the connections between PfPPs and fermions in Section~\ref{s:fermions_to_dpps}, it remains to connect fermions and quantum circuits.
Quantum circuits are briefly introduced in Section~\ref{s:what_is_a_circuit} for self-containedness.
In Section~\ref{s:circuit_proj_DPP}, we describe the quantum circuit of \cite{WHWCNT15}, later modified by \citet{JSKSB18}, that corresponds to a projection DPP given the diagonalized form of its kernel.
In Section~\ref{s:parallel_QR}, we depart from \citep{JSKSB18} and highlight the connections of their construction with a classical parallel QR algorithm in numerical algebra \citep{SaKu78}.
Consequently, we propose to take inspiration from more recent parallel QR algorithms, such as \citep{DGHL12}, to construct circuits with different constraints on the communication between qubits. 
In particular, we recover circuits with depth the shortest depth reported in \citep{KerPra22}, but with QR-style arguments rather than sophisticated data loaders.
Moreover, and maybe more importantly for DPP sampling, while sampling from a DPP with \emph{non-diagonalized} kernel remains limited by the initial (classical) diagonalization step, we argue that if one chooses the right avatar in the available distributed QR algorithms, we can even give a hybrid pipeline of a classical parallel and a quantum algorithm to sample some projection DPPs with \emph{non-diagonalized} kernel, with a linear speedup compared to the vanilla classical DPP sampler of \citet{HKPV06}.
The covered DPPs include practically relevant cases such as the uniform spanning trees and column subsets of Examples~\ref{ex:USTs} and \ref{ex:CSS}.
In Sections~\ref{s:reduction_mixture_DPPs} and \ref{s:dilation_argument}, we give two standard arguments, respectively due to \citep{HKPV06} and \citep{lyons2003}, to reduce the treatment of (non-projection) Hermitian DPPs to projection DPPs.
This concludes our treatment of DPPs.
In Section~\ref{s:quantum_circuit_PfPPs}, we go back to connecting point processes to the work of \citet{JSKSB18}, showing how one can use their circuit corresponding to the BdG Hamiltonian to sample a Pfaffian PP.

%
\subsection{Quantum circuits}
\label{s:what_is_a_circuit}
We refer to \citep[Part II]{NiCh10} for a description of quantum circuits and the basic building blocks.
In short, in the quantum circuit model of quantum computation, one describes a computation by the initialization of a set of $N$ qubits, a sequence of unitary operators among a small set of physically-implementable operators called \emph{gates}, and a physically-implementable observable called \emph{measurement}.
Not all gates can be implemented on every quantum computer, but software development kits like Qiskit \citep{Qiskit} allow the user to define a quantum circuit using a large enough set of gates.
The latter usually include any tensor product of identity matrices and Pauli matrices \eqref{e:pauli_matrices} (the so-called $X$, $Y$, and $Z$ gates, respectively corresponding to $\sigma_x$, $\sigma_y$, and $\sigma_z$ in \eqref{e:pauli_matrices}) and a few two-qubit gates such as the ubiquitous CNOT.
Qiskit then ``transpiles" the resulting circuit into a sequence of actually-implementable gates for a given machine. 
For instance, with the current technology, not all qubits can be jointly operated on a given machine, e.g., one may not be able to apply a gate that acts on two qubits that are physically too far from each other in the actual quantum machine, or only be able to do so with a significant chance of error.
Assuming the transpiling process preserves the dimensions of the circuit, one measures the complexity of a quantum circuit by quantities such as its total number of gates and its depth, i.e., the largest number of elementary gates applied to any single qubit.

Finally, in order to judge the possibility of sampling DPPs \emph{today}, and be able to estimate what we can do in the future as hardware and software improve, it is important to keep in mind the main sources of error and their order of magnitude in current quantum hardware; see Appendix~\ref{s:errors_in_quantum_computers}. 

\subsection{The case of a projection DPP}
\label{s:circuit_proj_DPP}
Consider again $(a_j)$ to be the Jordan-Wigner operators, which satisfy \eqref{e:CAR}, and $(\ket{\bn})$ to be the corresponding Fock basis.
Let $r\in\mathbb{N}^*$, $\bQ\in\mathbb{C}^{r\times N}$ have orthonormal rows, and $b_j^\adjoint = \sum_{\ell=1}^N q_{j\ell} a_\ell^\adjoint$ for $1\leq j\leq r$. 
Let further 
\begin{equation}
    \label{e:slater_determinant}
    \ket{\psi} = b_1^\adjoint \dots b_r^\adjoint \ket{\emptyset}.
\end{equation}
From Section~\ref{s:projDPP}, we know that simultaneously measuring $N_i = a_i^*a_i$, $i=1, \dots, N$ on $\rho = \ketbra{\psi}$ samples the projection DPP with kernel  $\bK = \bQ^\adjoint \bQ$. 
Measuring $N_i$ is easy, because the Fock basis of the JW operators is the computational basis, see \cref{p:jw}, and measurement in the computational basis is a basic operation on any quantum computer. 
For the same reason, implementing $a_1^* \dots a_r^*\ket\emptyset$ simply amounts to initializing the first $r$ qubits to $\ket{1}$, and the rest to $\ket{0}$.
It remains to be able to prepare the state \eqref{e:slater_determinant}, where the $b_i$'s intervene, not the $a_i$'s.
This is done by a procedure akin to the QR algorithm by successive Givens rotations;
a standard reference on QR algorithms and matrix computations in general is \citep{GoVa12}. 

\subsubsection{QR by Givens rotations}
\label{s:unitary_U_slater}
\begin{proposition}[\citealp{WHWCNT15,KMWGACGB18,JSKSB18}]
    \label{p:jiang}
    There is a unitary operator $\mathcal{U}(\bQ)$ on $\mathbb{H}$ such that 
    $$
        b_1^\adjoint \dots b_r^\adjoint \ket{\emptyset} = \mathcal{U}(\bQ) a_1^* \dots a_r^*\ket\emptyset,
    $$
    and $\mathcal{U}(\bQ)$ is a product of unitaries corresponding to elementary quantum gates.
\end{proposition}

\begin{proof}
    The Givens rotation with parameters $\theta,\phi$ and indices $\ell^1, \ell^2\in[N]$ is the unitary matrix
\begin{equation}
    \bG = \bG_{\ell^1,\ell^2}(\theta,\phi) = \mathbf{P}
    \begin{pmatrix}
        \bm{\gamma} & \mathbf{0}\\
        \mathbf{0} & \mathbb{I}
    \end{pmatrix}
    \mathbf{P}^{-1}
    \text{ with }
    \bm{\gamma} = 
    \begin{pmatrix} \cos\theta & \e^{-\i \phi}
        \sin\theta \\ -\sin\theta \e^{\i \phi}& \cos\theta \end{pmatrix},
    \label{e:givens}
\end{equation}
where $\mathbf{P}$ is the matrix of the permutation $(1\ \ell^1)(2\ \ell^2)$, allowing to select the vectors on which the rotation is applied.\footnote{The parametrization of Givens rotations slightly differs from \citep{JSKSB18}.}
Choosing $\theta, \phi$ relevantly, one can make sure that $\bQ\bG^*$ has a zero in position $(\ell^1, \ell^2)$, while all columns other than $(\ell_1, \ell_2)$ are left unchanged.
Iteratively choosing $(\ell^1, \ell^2)$, we can sequentially introduce zeros in $\bQ$ by multiplying it by $n_{G}$ Givens rotations $\bG_1, \dots,  \bG_{n_G}$, never changing an entry that we previously set to zero, until 
\begin{equation}
    \bQ \bG_1^* \dots \bG_{n_G}^* =  
    \begin{pmatrix}
        \bLambda \vert \mathbf{0}
    \end{pmatrix},
    \label{e:Givens_chain}
\end{equation}
where $\bLambda$ is an $r\times r$ diagonal unitary matrix, and the right block is the $r\times (N-r)$ zero matrix.

Now, to each Givens rotation $\mathbf{G} = \mathbf{G}_{\ell^1,\ell^2}(\theta,\phi)$, we associate a unitary operator $\mathcal{G} = \mathcal{G}_{\ell^1, \ell^2}(\theta, \phi)$ on $\mathbb{H}$. 
We call $\mathcal{G}$ a \emph{Givens operator}, as it realizes the Givens rotation $\mathbf{G}$ by a conjugation of creation operators, i.e.
\begin{align*}
    \mathcal{G} a_{\ell^1}^\adjoint\mathcal{G}^\adjoint &= \cos \theta a_{\ell^1}^\adjoint + \e^{-\i\phi}\sin \theta a_{\ell^2}^\adjoint,\\
    \mathcal{G} a_{\ell^2}^\adjoint\mathcal{G}^\adjoint &= -\e^{\i\phi}\sin \theta a_{\ell^1}^\adjoint +  \cos \theta a_{\ell^2}^\adjoint,\\
    \mathcal{G} a_{i}^\adjoint\mathcal{G}^\adjoint &= a_{i}^\adjoint, \quad \forall i \notin\{\ell^1, \ell^2\}.
\end{align*}
We refer to Appendix~\ref{a:more_about_Bogoliubov} for an explicit construction.
For future convenience, we note that this action can be compactly summarized if one stacks up the creation operators in a vector $\mathbf{a}^\adjoint = (a_1^\adjoint \cdots a_N^\adjoint)^\top$, and writes 
\begin{align}
    \mathcal{G} \mathbf{a}^\adjoint \mathcal{G}^\adjoint = \bG \cdot \mathbf{a}^\adjoint,\label{e:givens_operator_definition}
\end{align}
with $\cdot$ defined by matrix-vector multiplication.

Finally, consider the unitary
\begin{align}
    \mathcal{U}(\bQ) \triangleq \mathcal{G}_{n_G} \dots \mathcal{G}_1,
    \label{e:Givens_factorization_on_H}
\end{align}
where $\mathcal{G}_i$ is the Givens operator corresponding to $\bG_i$ in \eqref{e:Givens_chain}.
Then, by construction and up to a phase factor, for all $1\leq k \leq r$,
\[
    \mathcal{U}(\bQ) a_k^\adjoint \mathcal{U}(\bQ)^\adjoint
    = b_k^\adjoint, \text{ and } \mathcal{U}(\bQ) \ket{\emptyset} = \ket{\emptyset}.
\]
In particular, 
\begin{align*}
    \ket{\psi} = b_1^\adjoint \dots b_r^\adjoint \ket{\emptyset} &= 
    \mathcal{U}(\bQ) a_1^\adjoint \mathcal{U}(\bQ)^* \dots \mathcal{U}(\bQ) a_r^\adjoint\mathcal{U}(\bQ)^* \ket{\emptyset}\\
    &= 
    \mathcal{U}(\bQ) a_1^\adjoint \dots a_r^\adjoint\ket{\emptyset},
\end{align*}
again up to a phase factor, which is irrelevant in the resulting state $\ketbra{\psi}$.

%

\end{proof}

The operator $\mathcal{U}(\bQ)$ in Proposition~\ref{p:jiang} is indeed a product of elementary two-qubit gates, because any Givens operator $\mathcal{G}_{\ell^1,\ell^2}(\theta,\phi)$ can be implemented as such, as first put forward by \citep{WHWCNT15}.
We discuss gate details in Appendix~\ref{s:gates_details}.

To see how many gates are required for a given DPP kernel, we need to discuss a final degree of freedom: the Givens chain of rotations \eqref{e:givens} is not unique.
This is where constraints on which qubits can be jointly acted upon enter the picture.



\subsubsection{Parallel QR algorithms and qubit operation constraints}\label{s:parallel_QR}

The product of Givens rotations in \eqref{e:givens} can give rise to a quantum circuit of short depth if there are subproducts of rotations that can be performed on disjoint pairs of qubits. 
Independently of quantum computation, this is exactly the same kind of constraint that numerical algebraists have been studying in a long string of works on parallel QR factorization;\footnote{Note that, unlike the common QR decomposition, the matrix $R$ obtained here has more structure: it contains only one non-zero element per row.
To some extent, the decomposition we discuss is close to a singular value decomposition; see the end of \cref{s:parallel/quantum_proj_DPPs} for a short discussion.} see \citep{DGHL12} and references therein.

As a first example, after a preprocessing phase, the preparation of \eqref{e:slater_determinant} in \citep[Section III]{JSKSB18} implicitly implements a parallel QR algorithm known as Sameh-Kuck \citep{SaKu78}.
More precisely, in a preprocessing phase, \cite{JSKSB18} zero out $r(r-1)/2$ entries in the upper-right corner of $\bQ$ by pre-multiplying $\bQ$ by a product of (complex) Givens rotations.
Let us write this product of unitary matrices by $\mathbf{V}$.
This preprocessing requires at most\footnote{
    There might be zeros that appear out of collinearity, and thus do not need to be forced.
} 
$r(r-1)/2$ Givens rotations, which are implemented thanks to a classical computer.
To fix ideas, when $r = 3$ and $N = 6$, this results in a matrix of the following form
\begin{equation}
    \bV \bQ = \begin{pmatrix}
        * & * & *  & * & 0 & 0\\
        * & * & *  &* & * & 0\\
        * & * & * & * & * & *
    \end{pmatrix}.
    \label{e:jiang's_example}
\end{equation}
Note that replacing $\bQ$ by $\bV\bQ$ does not change the kernel of the resulting DPP since $\bV^\adjoint \bV = \bI$.
Now, \cite{JSKSB18} apply QR to $\bV \bQ$ as in the proof of Proposition~\ref{p:jiang}, applying Givens rotations on disjoint pairs of neighbouring columns, as they become available. 
Rather than giving a formal description of the algorithm, available in \citep{SaKu78}, we follow \cite{JSKSB18} and depict its application on Example \eqref{e:jiang's_example}.  
We use the convenient notation of \cite{JSKSB18}, i.e.\ bold characters to show the most recently actively updated entries of the matrix, and underlined characters to show entries that automatically result from the rows of $\bQ$ being orthonormal. 
The successive parallel rounds of the algorithm are then
\begin{align}
    \bV\bQ 
    & =  
    \begin{pmatrix}
        * & * & *  & * & 0 & 0\\
        * & * & *  &* & * & 0\\
        * & * & * & * & * & *
    \end{pmatrix}
    \rightarrow
    \begin{pmatrix}
        * & * & *  & \bzero & 0 & 0\\
        * & * & *  &* & * & 0\\
        * & * & * & * & * & *
    \end{pmatrix}\nonumber\\
    &\rightarrow
    \begin{pmatrix}
        * & * & \bzero  & 0 & 0 & 0\\
        * & * & *  &* & \bzero & 0\\
        * & * & * & * & * & *
    \end{pmatrix}
    \rightarrow
    \begin{pmatrix}
        \underline{\lambda_1} & \bzero & 0  & 0 & 0 & 0\\
        \underline{0} & * & *  &\bzero & 0 & 0\\
        \underline{0} & * & * & * & * & \bzero
    \end{pmatrix}\nonumber\\
    &\rightarrow
    \begin{pmatrix}
        \lambda_1 & 0 & 0  & 0 & 0 & 0\\
        0 & \underline{\lambda_2} & \bzero  & 0 & 0 & 0\\
        0 & \underline{0} & * & * & \bzero & 0
    \end{pmatrix}
    \rightarrow
    \begin{pmatrix}
        \lambda_1 & 0 & 0  & 0 & 0 & 0\\
        0 & \lambda_2 & 0  & 0 & 0 & 0\\
        0 & 0 & \underline{\lambda_3} & \bzero & 0 & 0
    \end{pmatrix}.\label{e:QR_nnb_Jiang}
\end{align}
The factors $\lambda_i$ are of unit modulus, as in the construction behind Proposition~\ref{p:jiang}.
The upshot is that only a lower triangular matrix remains whose only non-zero entries are on the diagonal, while the other entries of the lower triangle automatically vanish due to row orthonormality.

The resulting circuit, an example of which is given in Figure~\ref{f:dpp_circuit}, requires $\cO(Nr)$ gates, and has depth $\cO(N)$.
The $\mathcal{O}(r^2)$ preprocessing, which is optional and could also be performed in the quantum circuit, allows getting rid of a handful of (necessarily faulty) quantum gates at a small classical cost. 
Finally, constraining Givens operators to act on neighbouring qubits, or equivalently, Givens rotations to act on neighbouring columns, is practically relevant if the actual quantum computer at disposal favours two-qubit gates acting on neighbouring qubits.

As a second extreme example, assume that we have no constraint on which pairs of qubits can be jointly operated.
We can then perform parallel QR on $\bQ$ using only $\cO(r\log N)$ parallel rounds, simply by acting on all available disjoint pairs in a single row until all but one entry are zeroed, keeping in mind that the leftmost $r$ entries of each row will be automatically updated because of orthonormality constraints.
On an example with $N=8$ and $r=3$, this would yield
\begin{align*}
    \bQ 
    &\rightarrow 
    \begin{pmatrix}
        * & \bzero & *&  \bzero & *  & \bzero & * & \bzero\\
        * & * & * & * & * & * & * & *\\
        * & * & * & * & * & * & * & *\\
    \end{pmatrix}
    \rightarrow 
    \begin{pmatrix}
        * & 0 & \bzero &  0 & *  & 0 & \bzero & 0\\
        * & * & * & * & * & * & * & *\\
        * & * & * & * & * & * & * & *\\
    \end{pmatrix}\\
    &\rightarrow 
    \begin{pmatrix}
        \underline{\lambda_1} & 0 & 0 &  0 & \bzero  & 0 & 0 & 0\\
        \underline{0} & * & * & * & * & * & * & *\\
        \underline{0} & * & * & * & * & * & * & *\\
    \end{pmatrix}
    \rightarrow
    \begin{pmatrix}
        \underline{\lambda_1} & 0 & 0 &  0 & 0  & 0 & 0 & 0\\
        0 & * & * & \bzero & * & \bzero & * & \bzero\\
        0 & * & * & * & * & * & * & *\\
    \end{pmatrix}\\
    &\rightarrow 
    \begin{pmatrix}
        {\lambda}_1 & 0 & 0 &  0 & 0  & 0 & 0 & 0\\
        0 & * & \bzero & 0& * & 0 & \bzero & 0\\
        0 & * & * & * & * & * & * & *\\
    \end{pmatrix}
    \rightarrow
    \begin{pmatrix}
        {\lambda}_1 & 0 & 0 &  0 & 0  & 0 & 0 & 0\\
        0 & \underline{\lambda_2} & 0 & 0& \bzero & 0 & 0 & 0\\
        0 & \underline{0} & * & * & * & * & * & *\\
    \end{pmatrix}\\        
    &\rightarrow 
    \begin{pmatrix}
        {\lambda}_1 & 0 & 0 &  0 & 0  & 0 & 0 & 0\\
        0 & \lambda_2 & 0 & 0& 0 & 0 & 0 & 0\\
        0 & 0 & * & \bzero & * & \bzero & * & \bzero\\
    \end{pmatrix}        
    \rightarrow
    \begin{pmatrix}
        {\lambda}_1 & 0 & 0 &  0 & 0  & 0 & 0 & 0\\
        0 & \blambda_2 & 0 & 0& 0 & 0 & 0 & 0\\
        0 & 0 & * & 0 & * & 0 & \bzero & 0\\
    \end{pmatrix}\\
    &\rightarrow
    \begin{pmatrix}
        {\lambda}_1 & 0 & 0 &  0 & 0  & 0 & 0 & 0\\
        0 & \lambda_2 & 0 & 0& 0 & 0 & 0 & 0\\
        0 & 0 & \underline{\lambda_3} & 0 & \bzero & 0 & 0 & 0\\
    \end{pmatrix}.       
\end{align*}
The resulting quantum circuit has $\cO(Nr)$ gates as the one of \cite{JSKSB18}, but a depth of only $\cO(r\log_2 N)$.
This depth is similar to the shortest depth obtained by \cite{KerPra22} with a similar tree-like structure called ``parallel Clifford loaders".

Just like in parallel QR, we argue that the choice of the chain \eqref{e:givens} of rotations should depend on the particular hardware constraints, like which qubits we allow to be jointly operated.
Parallel QR algorithms allow quite arbitrary dependency structures, see e.g. \citep[Section 4]{DGHL12}.
While initially motivated by communication- or storage constraints in parallel classical computing, these dependency structures are actually tailored to designing quantum circuits to prepare states like \eqref{e:slater_determinant} depending on a given quantum architecture.
Taking the analogy with parallel QR even one step further, we now show that the initial bottleneck of diagonalizing the kernel of a DPP can be combined with the design of the quantum circuit in a single run of a parallel QR algorithm.

\begin{figure}
    \includegraphics[width=\textwidth]{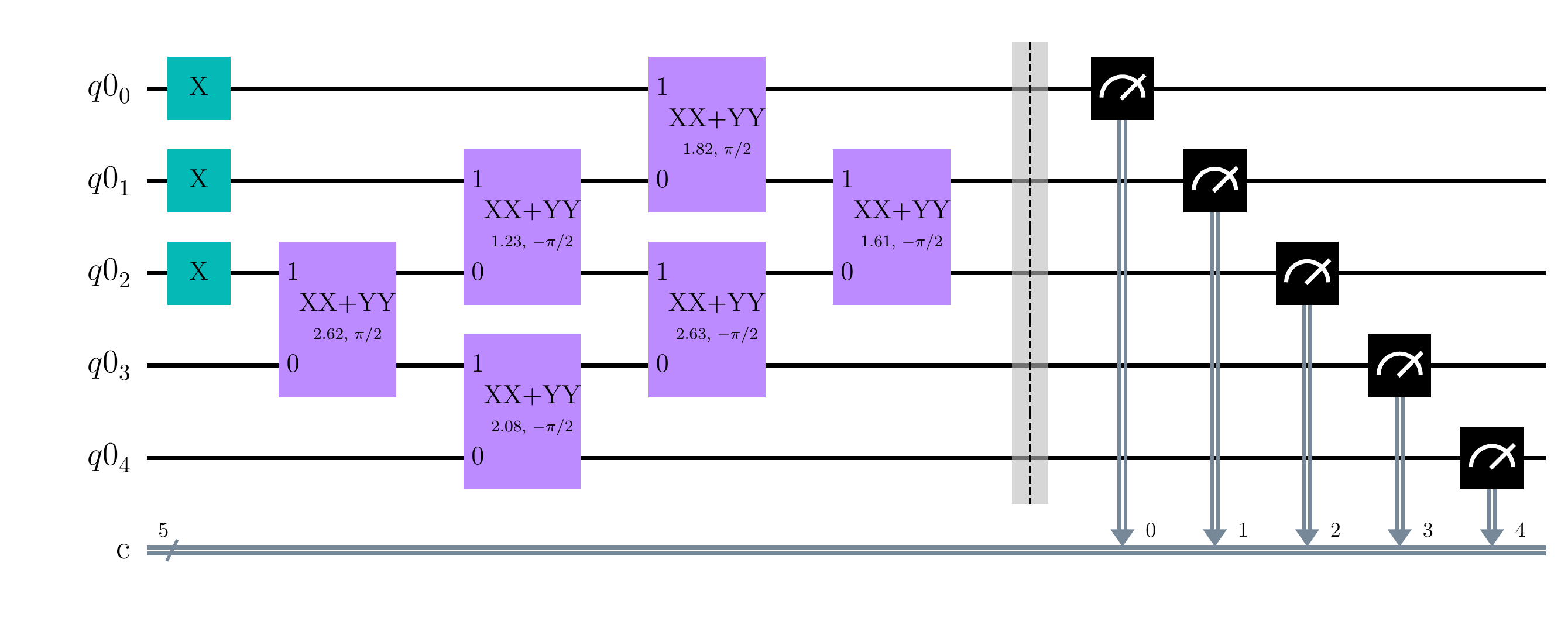}
    \caption{A circuit sampling a DPP with projection kernel of rank $r = 3$, with $N = 5$ items.
    On the left-hand side, Pauli $X$ gates are used to create fermionic modes on the three first qubits.
    Note also the parallel QR Givens rotations on neibouring qubits indicated by parametrized $XX+YY$ gates.
    See \cref{s:gates_details} for gate details.
    On the right-hand side, measurements of occupation numbers are denoted by black squares.
    }
    \label{f:dpp_circuit}
\end{figure}

%
\subsubsection{A hybrid parallel/quantum algorithm for projection DPPs}\label{s:parallel/quantum_proj_DPPs}

The link of Givens-based quantum circuits with parallel QR algorithms suggests pipelines to sample from projection DPPs, even when the kernel is not yet in diagonalized form. 

\begin{proposition}
    Assume that we have access only to $\bA \in \mathbb{R}^{d\times N}$, and that we want to sample from DPP($\bK$), where $\bK = \bV_{:,[r]} \bV_{:,[r]}^\adjoint$ and $\bV$ is defined by the singular value decomposition $\bA = \bU \bSigma \bV^\adjoint$.
    Then, given $P$ classical processors, for a run time in $\cO(Nd^2/P)$ up to logarithmic factors in $P$, we can design a quantum circuit with depth $\cO(N\log P/P)$ and $\cO(Nd)$ gates, such that simultaneous measurements in the computational basis at the output of the circuit yield a sample of DPP($\bK$).
\label{p:hybrid_sampler}
\end{proposition}

Three comments are in order. 
First, specifying the kernel $\bK$ as in Proposition~\ref{p:hybrid_sampler}, where $\bK$ is the projection kernel onto the first $r$ principal components of a full-rank rectangular matrix $\bA$, is common in practice.
In particular, it covers the column subset selection of Example~\ref{ex:CSS} and the uniform spanning forests of Example~\ref{ex:USTs}.
Second, for simplicity we have omitted the additional $\cO(r^2)$ preprocessing introduced by \cite{JSKSB18}.
Third, at least if we neglect physical sources of noise in the quantum circuit (see Appendix~\ref{s:errors_in_quantum_computers}), the bottleneck in the hybrid sampler remains the initial classical cost $\cO(Nd^2/P)$.
Compared with the vanilla classical approach, we have gained a linear speedup thanks to parallelization. 
Note that it is not clear how such a linear speedup can be gained in the vanilla classical algorithm, but that it can be gained in non-quantum randomized variants \citep{DCMW19,BTA22}.
Dequantizing Proposition~\ref{p:hybrid_sampler}, to obtain a non-quantum algorithm with a similar cost, is actually an interesting question, which we leave to future work. 

\begin{proof}
To prove Proposition~\ref{p:hybrid_sampler}, we first use a Givens-based parallel QR algorithm to compute $\bA^\adjoint = \bQ\bR$; see \citep{DGHL12} and references therein for a recent entry. 
One can incorporate here communication constraints that will turn into qubit-communication constraints in the final algorithm.
For simplicity, we assume $N\gg d$ and use the parallel tall-skinny QR of \citep[Section 2.1]{DGHL12}, with $P\leq N/d$ processors. 
Without going into details, for $P$ a power of two, the idea is to partition $\bA^\adjoint$ into $P$ equal $N/P\times d$ blocks, perform QR for each block using a user-chosen QR algorithm, and for $\log P$ stages, group the resulting $R$ matrices in pairs and perform QR. 
The run time is $\cO(Nd^2/P)$, up to logarithmic factors, which is a linear speedup compared to classical, non-parallel QR. 
Now, perform the singular value decomposition of the $d\times d$ matrix $\bR = \tilde\bU \tilde\bSigma \tilde\bV^\adjoint$; this costs $\cO(d^3)$.
Then 
$$
    \bA^\adjoint\bA = \bQ \bR\bR^\adjoint \bQ^\adjoint = (\bQ\tilde\bU) \tilde\bSigma^2 (\bQ\tilde\bU)^\adjoint,
$$
so that the first $r$ principal components of $\bA$ are $\bQ\tilde\bU_{:,[r]}$.
Essentially, we have used the QR algorithm to compute the principal directions: this detour is valuable because we simultaneously $(i)$ benefit from parallelization and $(ii)$ obtain $\bQ$ as a product of Givens rotations by construction.
If we now decompose $\tilde\bU_{:,[r]}$ as a product of Givens rotations, say using the Sameh-Kuck algorithm like in \citep{JSKSB18}, we paid a (classical) preprocessing cost of $\cO(Nd^2/P + d^3)$, up to logarithmic factors in $P$, to obtain all the information needed to create a quantum circuit with Givens gates that samples DPP($\bK$) starting from the knowledge of $\bA$.
The depth depends on the blockwise QR algorithm used in the initial parallel tall-skinny QR step. 
If we use Givens rotations à la Sameh-Kuck to enforce only nearest-neighbour qubit interactions, we obtain a depth of $\cO(N\log P/P)$, and a number of gates of $\cO(Nd)$. 
The limiting factor of the overall pipeline remains the QR factorization of $\bA$, in $\cO(Nd^2/P)$ flops, but we have gained a linear speedup.
\end{proof}
\subsection{Reducing general DPPs to mixture DPPs}\label{s:reduction_mixture_DPPs}
In Section~\ref{s:circuit_proj_DPP}, we gave quantum circuits to sample projection DPPs. 
If the kernel $\bK$ is not that of a projection, a standard argument by \citet{HKPV06} allows writing the corresponding DPP as a statistical mixture of projection DPPs, thus yielding a quantum sampler for DPP($\bK$) with a further classical preprocessing step, the latter implementing the choice of the mixture component.

More precisely, assume $\bK$ is available in diagonalized form
\begin{align}
    \bK = \bU\, \Diag\left( \nu_k\right)\, \mathbf U^\adjoint,
    \label{eq:correlation_kernel_DPP_sampling}
\end{align}
with $\nu_1, \dots, \nu_N \in [0,1]$.
Then sampling DPP($\bK)$ can be done in two steps.
First, sample $N$ independent Bernoulli random variables with respective parameters $\nu_1, \dots, \nu_N$ on a classical computer.
Let $\mathcal{C}\subseteq [N]$ be the set of indices of successful Bernoulli trials.
Second, measure all the occupation numbers $N_i$ for $1\leq i \leq N$ in the quantum circuit of \cref{s:circuit_proj_DPP}; in other words, sample from the projection DPP of correlation kernel $\mathbf{Q}^* \mathbf{Q}$ with $\mathbf{Q} = (\bU_{:\mathcal{C}})^*$, where $\bU_{:\mathcal{C}}$ is the set of columns indexed by $\mathcal{C}$.

To conclude this section, we note how one can arrive at this mixture construction by inspecting the quantum state ``above'' DPP($\bK$).
Consider the factorization of the state \cref{e:gaussian} obtained thanks to the diagonalization of $H$, i.e.,
\begin{align}
    \rho = \prod_{k = 1}^N \frac{\exp(-\beta (\lambda_k-\mu) b_k^\adjoint b_k)}{1 + \exp(-\beta (\lambda_k-\mu))},\label{eq:rho_DPP_as_product}
\end{align}
where we used the formula \cref{e:constant_intermediate} for the normalization constant.
Note that the eigenvalues of the correlation kernel \cref{eq:correlation_kernel_DPP_sampling} are obtained by taking the sigmoid of the Hamiltonian eigenvalues, i.e., $\nu_k = \sigma(-\beta(\lambda_k-\mu))$, in the light of \cref{prop:DPP_from_rho}.
By a simple inspection, the factor $k \in [N]$ in the product \cref{eq:rho_DPP_as_product}, is given by 
\[
    \frac{\exp(-\beta (\lambda_k-\mu) b_k^\adjoint b_k)}{1 + \exp(-\beta (\lambda_k-\mu))} = (1-\nu_k)\times \pi_{\Ker b_k} + \nu_k \times \pi_{\Ker b_k^\adjoint},
\]
where $\pi_{\Ker b_k}$  (resp. $\pi_{\Ker b_k^\adjoint}$) is the orthogonal projector onto the null space of $b_k$ (resp. $b_k^\adjoint$).
Inspecting Born's rule reveals that $\tr \rho N_{i_1}\dots N_{i_\ell}$ is equal to the correlation function of a statistical mixture of projection DPPs, each coming with a weight of the form 
$$
    \prod_{i\in I} \nu_i \prod_{j\notin I} (1-\nu_j) 
$$
for some $I\subset[N]$. 
These weights correspond to the independent Bernoulli trials introduced by \cite{HKPV06} and discussed immediately after \cref{eq:correlation_kernel_DPP_sampling}. 

\subsection{Dilating a general DPP to a projection DPP}
\label{s:dilation_argument}
In principle, there exists another strategy to sample from a general DPP by only using a quantum circuit such as described in \cref{s:circuit_proj_DPP} -- thus bypassing the first step of the sampling algorithm consisting of the classical sampling of Bernoulli random variables.
Any DPP on the ground set $\{1,\dots,N\}$ can be dilated to a projection DPP on $\{1, \dots, N, N+1, \dots, 2N\}$ by the transformation
\[
    \bK 
    \mapsto
    \bK_{\rm dil}
    = 
    \begin{pmatrix}
        \bK & (\bK(\bI-\bK))^{1/2}\\
        (\bK(\bI-\bK))^{1/2} & \bI - \bK
    \end{pmatrix},
\]
see \citep[Section 8]{lyons2003}.
The points of a sample of DPP($\bK_{\rm dil}$) which belong to $\{1,\dots,N\}$ yield a sample of DPP($\bK$).
\subsection{Sampling the PfPP corresponding to the BdG Hamiltonian}
\label{s:quantum_circuit_PfPPs}
We now turn to PfPPs. 
As above, we can write the mixed state of the Bogoliubov-de Gennes Hamiltonian at inverse temperature $\beta > 0$ as 
\begin{equation}
    \frac{\exp(-\beta H_{\mathrm{BdG}})}{Z} = \prod_{k=1}^N \frac{\exp(-\beta \epsilon_k b_k^\adjoint b_k)}{1 + \exp(-\beta \epsilon_k)}, \label{e:mixture_pfaffian}
\end{equation}
where $\epsilon_k$ is the energy of the eigenmode $k$.
Hence, in complete analogy with the case of DPPs in \cref{s:reduction_mixture_DPPs}, the product in \cref{e:mixture_pfaffian} entails a statistical mixture formula.
Therefore, sampling the corresponding Pfaffian point process amounts to proceed as follows:
\begin{enumerate}
    \item sample $N$ independent Bernoullis with success probability $\sigma(-\beta \epsilon_k)$ on a classical computer, to obtain the set of successful indices $\mathcal{C}\subseteq \{1, \dots, N\}$,
    \item jointly measure the occupation numbers $N_i = a_i^\adjoint a_i$ for all  $1\leq i\leq N$ in the pure state $ \prod_{i\in \mathcal{C}} b_{i}^\adjoint \ket{\emptyset_b}$ with a quantum circuit,
\end{enumerate}
where $\ket{\emptyset_b}$ is the joint Fock vacuum of the $b_k$'s.
Note that the preparation of $ \prod_{i\in \mathcal{C}} b_{i}^\adjoint \ket{\emptyset_b}$ is discussed in \cref{s:double_Givens}.

The output of this sampling algorithm is the set of indices for which the measure has given an occupation number equal to $1$ rather than equal to $0$.
Again, this can be done easily thanks to the representation of the occupation numbers $N_i$'s in \cref{p:jw}.
\begin{remark}[Projective $\bS$]\label{rem:projective_S_for_PfPP}
    In the light of \cref{prop:Pfaffian_PP}, the Pfaffian point process associated with the pure state $    \prod_{i\in \mathcal{C}}b_i^\adjoint\ket{\emptyset_b}$ is determined by the
    orthogonal projection matrix
    \begin{equation}
        \mathbf{S}(\mathcal{C}) = \overline{\mathbf{W}}^{\adjoint}
        \begin{pmatrix}
             \mathbf{I}_{\bar{\mathcal{C}}}& \mathbf{0}\\
            \mathbf{0} & \mathbf{I}_{\mathcal{C}}
        \end{pmatrix}
        \overline{\mathbf{W}},\label{e:S(C)}
    \end{equation}
    where $\mathbf{I}_{\mathcal{C}}$ is the diagonal matrix with entry $(i,i)$ equal to $1$ if $i\in \mathcal{C}$ and zero otherwise, whereas $ \mathbf{I}_{\bar{\mathcal{C}}} = \mathbf{I} - \mathbf{I}_{\mathcal{C}}$. 
    Notice that $\mathbf{C}\overline{\mathbf{S}}\mathbf{C} = \bI - \mathbf{S}$ holds.

    We now shortly discuss a simple consequence of \cref{prop:Pfaffian_sample_parity} in the case of the zero temperature limit a the PfPP, denoted by $Y_\beta$ and associated with the density matrix $\rho  = Z^{-1} \exp( -\beta H_{\mathrm{BdG}})$, to highlight the role of the Majorana quadratic form of $H_{\mathrm{BdG}}$, namely $\mathbf{A}$, see \cref{lem:diag_H_BdG}.
    This PfPP is associated with the matrix
        \begin{align*}
        \mathbf{S} 
        = \overline{\mathbf{W}}^{\adjoint}
        \begin{pmatrix}
             \Diag(\sig(\beta \epsilon_k))& \mathbf{0}\\
            \mathbf{0} & \Diag(\sig(-\beta \epsilon_k))
        \end{pmatrix}
        \overline{\mathbf{W}}.
    \end{align*}
    For simplicity, assume now that the eigenvalues of $H_{\mathrm{BdG}}$ are such that: $\epsilon_k < 0$ for all $k\in \mathcal{C}$ and $\epsilon_k > 0$ otherwise\footnote{An alternative strategy uses a chemical potential as in \cref{s:projDPP}.}; and take the limit of $\rho$ for $\beta\to +\infty$.
    
    The correlation functions of $Y_\beta$ tend to the correlation functions of the process $Y_\infty$ associated to the orthogonal projection $\mathbf{S}(\mathcal{C})$ given in \cref{e:S(C)}.
    Meanwhile, the expectation of the parity w.r.t.\ $Y_\beta$ tends to
    \[
        \E_{Y_\infty} (-1)^{|Y_\infty|} = (-1)^{|\mathcal{C}|} \sign \pf(\mathbf{A}) = (-1)^{|\mathcal{C}|} \det \mathbf{W},
    \]
   \corr{in the light of \cref{lem:diag_H_BdG}.}
   Since $\det \mathbf{W} \in \{-1,1\}$, the right-hand side of this equality is equal to either $1$ or $-1$. 
    But $(-1)^{|Y_\infty|}$ can only take values in $\{-1,1\}$, so that the parity of the cardinality of $Y_\infty$ is almost surely fixed.

    In particular, when $\mathcal{C} = \emptyset$, the Pfaffian process $Y_\infty$ is determined by the lowest energy eigenvector of $H_{\mathrm{BdG}}$, namely $\ket{\emptyset_b}$, and its expected parity is $\sign \pf(\mathbf{A})$, which corresponds to the formula derived by \citet[Eq.\ (19)]{Kitaev01} in the setting of a topological superconductor in one dimension.
    \corr{Note the direct connection with \cref{rem:fermion_parity} about the fermionic parity of eigenvectors of $H_{\mathrm{BdG}}$.}
    Thus, in this case, the parity is odd if and only if $\bm{W} = \bm{\Omega}^\adjoint \bm{R} \bm{\Omega}$ is such that the (real) orthogonal matrix $\bm{R}$ is in the connected component of $O(2N)$ not containing the identity element.
\end{remark}
\subsubsection{QR decomposition with double Givens rotations}\label{s:double_Givens}
We now describe how to adapt the discussion of \cref{s:circuit_proj_DPP} to the case where the particle number is not conserved.
To simplify notation, we suppose that $\mathcal{C} = \{1, \dots, r\}$, and we prepare the state
\begin{align}
    b_1^\adjoint\dots b_r^\adjoint\ket{\emptyset_b}, \label{e:fermionic_gaussian_state_W}
\end{align}
where the creation-annihilation operators $b_k$'s are obtained from the $a_k$'s thanks to a Bogoliubov transformation\footnote{Recall that the $a_k$'s are the Jordan-Wigner operators representing the $c_k$'s in \cref{e:Bog_transform}.} \cref{e:Bog_transform} with a unitary orthogonal complex matrix $\mathbf{W}$
given in \cref{lem:diag_H_BdG}. 
Due to particle number non-conservation, the Fock vacuum of the $b_i$'s, denoted by $\ket{\emptyset_b}$ and given below, is not annihilated by the $a_i$'s. 
Also, recall that the Bogoliubov transformation is given by a matrix of the form
\begin{equation}
    \label{e:w_matrix}
    \mathbf{W} =\begin{pmatrix}
        \overline{\mathbf{W}}_1 & \overline{\mathbf{W}}_2 \\
        \mathbf{W}_2 & \mathbf{W}_1
    \end{pmatrix},
\end{equation}
where the blocks are such that this transformation is unitary; see \cref{def:orthogonal_complex}.
Explicitly, as explained in \citep[Section IV]{JSKSB18}, the full transformation is determined by the lower blocks of $\mathbf{W}$, which encode the expression of $\mathbf{b}$ in terms of $\mathbf{a}^\adjoint$ and $\mathbf{a}$.
Now, we give the following result that we formalize and adapt from \citep[Equation (31)]{JSKSB18} \corr{by including cases with fewer than $N$ particle-hole transformations which were not explicitly considered by \citet{JSKSB18}}.
\begin{lemma}\label{lem:lower_block}
    There exists a $2N \times 2N$ orthogonal complex matrix $\mathbf{O}$ and an $N\times N$ unitary matrix $\bV$ such that
    \[
    \bV \begin{pmatrix}
        \mathbf{W}_2  \vert \mathbf{W}_1
    \end{pmatrix} \mathbf{O}^\adjoint = \begin{pmatrix}
        \mathbf{0}  \vert \mathbf{D}
    \end{pmatrix},
    \]
    where $\mathbf{D} = \Diag(z_k)$ is a diagonal matrix with $|z_k| = 1$ for $1\leq k \leq N$.
    Furthermore, the matrix $\mathbf{O}$ is here the following product of Givens rotations \emph{and} particle-hole transformations
\begin{align}
    \mathbf{O} =  \mathbf{B}^{\corr{\mathtt{m}_{N}}}\bm{M}_{N-1}\mathbf{B}^{\corr{\mathtt{m}_{N-1}}}  \dots \mathbf{B}^{\corr{\mathtt{m}_{2}}}\bm{M}_1\mathbf{B}^{\corr{\mathtt{m}_{1}}},\label{eq:unitary_U_slater_ph}
\end{align}
the \corr{terms} of which we now explain. 
The matrix $\bm{M}_i$ is a product of \corr{at most} $N-i$ double Givens rotation matrices. 
The double Givens rotation $\bm{\Gamma}$ associated with $(\theta, \phi)$ and vector indices $1\leq \ell^1 < \ell^2\leq N$ is defined as
\[
    \bm{\Gamma} =
    \begin{pmatrix}
        \bG_{\ell^1,\ell^2}(\theta,\phi) & \mathbf{0}\\
        \mathbf{0} & \overline{\bG}_{\ell^1,\ell^2}(\theta,\phi)
    \end{pmatrix}
    \text{ with $\bG$ defined in \cref{e:givens}}.
\]
The matrix $\mathbf{B}$ in \eqref{eq:unitary_U_slater_ph} is a permutation matrix exchanging the last vector of the first block with the last vector of the second block.
Formally, it is given by the so-called \emph{particle-hole} matrix
\[
    \mathbf{B} = \begin{pmatrix}
        \bI - \mathbf{e}_N \mathbf{e}_N^\top & \mathbf{e}_N \mathbf{e}_N^\top\\
        \mathbf{e}_N \mathbf{e}_N^\top & \bI - \mathbf{e}_N \mathbf{e}_N^\top
    \end{pmatrix}.
\]
\corr{The Booleans $(\mathtt{m}_1, \dots, \mathtt{m}_N) \in \{0,1\}^N$ simply indicate whether the corresponding \emph{particle-hole} matrices appear or not in the decomposition \cref{eq:unitary_U_slater_ph}.}
Finally, note that because we use the Jordan-Wigner transform, $\mathbf{B}$ can be implemented by a Pauli $X$ gate.
\end{lemma}
Let us give a few details that sketch the proof of \cref{lem:lower_block}.
On the one hand, the matrix $\bV$ is a product of single Givens rotations applied to the rows of 
$
(\begin{smallmatrix}
    \mathbf{W}_2 & \mathbf{W}_1
\end{smallmatrix}),
$
to yield an upper triangle of zeros in the left block. 
For instance, when $N=3$, and assuming for simplicity that $\mathbf{W}_2$ is full-rank,\footnote{
    An example illustrating a case where $\mathbf{W}_2$ is rank-deficient is given in \cref{app:rank_defficient_ph}.
}we obtain 
\begin{align}
\bV\begin{pmatrix}
        \mathbf{W}_2  \vert \mathbf{W}_1
    \end{pmatrix}
    =
    \left(\begin{array}{ccc|ccc}
        0 & 0 & \square   & * & * & \underline{0}\\
        0 & \square & * &* & * & * \\
        \square & * & * & * & * & *
    \end{array}\right),\label{e:Givens_first_step}
\end{align}
where we have highlighted the diagonal of the first block as squares. 
The full-rank assumption implies that each square in \eqref{e:Givens_first_step} corresponds to a nonzero element.
This in turn implies that there is a zero at the top right corner of the second block, as we now explain.
Denote by $\mathbf{r}_{1,i}$ and $\mathbf{r}_{2,i}$, the $i$th \emph{row} vector of the first and second block, respectively. 
Thanks to  \cref{lem:unitary_Bog}, we have the following inner product identities between rows
\begin{align}
    \mathbf{r}_{1,i}\mathbf{r}_{2,j}^\top = -\mathbf{r}_{2,i}\mathbf{r}_{1,j}^\top, \label{e:ortho_diff_rows}\\
    \mathbf{r}_{1,i}\mathbf{r}_{1,j}^\adjoint  = \delta_{ij} -  \mathbf{r}_{2,i}\mathbf{r}_{2,j}^\adjoint.\label{e:ortho_same_rows}
\end{align}
Property \cref{e:ortho_diff_rows} explains why we necessarily already have a zero in the right-hand block of \cref{e:Givens_first_step}.

Starting from \eqref{e:Givens_first_step}, we now explain how to build the matrix $\mathbf{O}$ in \eqref{eq:unitary_U_slater_ph}. 
More precisely, the next step consists in filling the remainder of the left block of \cref{e:Givens_first_step} with zeros, one row at a time. 
We start with a right-multiplication with the particle-hole transformation $\mathbf{B}$ to zero the top right element. 
Then right-multiplication with double Givens matrices are used to zero the second row of the left block, from left to right, until we reach the last element of the second row. 
If that last element is not zero, orthogonality of the rows implies that the last element of the row of the second block is zero, so that a right-multiplication with the particle-hole transformation $\mathbf{B}$ finishes our procedure of zeroing the second row of the left block. 
We repeat the operation row after row, using $\mathbf{B}$ to zero the right-most element of each row if it is not already zero.

To make things more visual, we show the right multiplication by $\mathbf{O}^*$ on our example \eqref{e:Givens_first_step}, starting with a particle-hole transformation on \cref{e:Givens_first_step}.
For simplicity of the display, we assume that particle-hole transformations are necessary for each row.
This gives
\begin{align*}
    &\bV\begin{pmatrix}
        \mathbf{W}_2  \vert \mathbf{W}_1
    \end{pmatrix}\\
    &\xrightarrow{\times \mathbf{B}} \left(\begin{array}{ccc|ccc}
        0 & 0 & \mathbf{0}   & * & * & *\\
        0 & * & *  &* & * & * \\
        * & * & * & * & * & *
    \end{array}\right)
    \xrightarrow{\times\Gamma^\adjoint} \left(\begin{array}{ccc|ccc}
        0 & 0 & 0   & * & * & \underline{0}\\
        0 & \mathbf{0} & *  &* & * & \underline{0} \\
        * & * & * & * & * & *
    \end{array}\right) \text{ by \cref{e:ortho_diff_rows}}\\
    &\xrightarrow{\times\Gamma^\adjoint} \left(\begin{array}{ccc|ccc}
        0 & 0 & 0   & \underline{\alpha}_{(2)} & \underline{0}_{(1)} & 0\\
        0 & 0 & *  &* & * & \underline{0}_{(3)} \\
        \mathbf{0} & * & * & * & * & *
    \end{array}\right) \text{ by \cref{e:ortho_diff_rows} at (1), \cref{e:ortho_same_rows} at (2), \cref{e:ortho_diff_rows} at (3)}\\
    &\xrightarrow{\times\mathbf{B}} \left(\begin{array}{ccc|ccc}
        0 & 0 & 0   & \alpha & 0 & 0\\
        0 & 0 & \mathbf{0}  &\underline{0} & * & \bm{*} \\
        0 & * & * & \underline{0} & * & *
    \end{array}\right) \text{ by \cref{e:ortho_same_rows}}\\
    &\xrightarrow{\times\Gamma^\adjoint}\left(\begin{array}{ccc|ccc}
        0 & 0 & 0   & \alpha & 0 & 0\\
        0 & 0 & 0  &0 & \underline{\beta}_{(3)} & 0 \\
        0 & \mathbf{0} & \underline{\gamma}_{(4)} & 0 & \underline{0}_{(2)} & \underline{0}_{(1)}
    \end{array}\right) \text{ by \cref{e:ortho_diff_rows} at (1), by \cref{e:ortho_same_rows} at (2), (3), (4)}\\
    &\xrightarrow{\times\mathbf{B}}\left(\begin{array}{ccc|ccc}
        0 & 0 & 0   & \alpha & 0 & 0\\
        0 & 0 & 0  &0 & \beta & 0 \\
        0 & 0 & 0 & 0 & 0 & \gamma
    \end{array}\right) = \bV\begin{pmatrix}
        \mathbf{W}_2  \vert \mathbf{W}_1
    \end{pmatrix}\mathbf{O}^\adjoint,
\end{align*}
with unit modulus complex numbers $\alpha,\beta,\gamma$.
To help the reader understand which of properties  \cref{e:ortho_diff_rows} and \cref{e:ortho_same_rows} applied, an index was added to matrix elements which are automatically set to zero.
\begin{proposition}[Formalization of a result of \cite{JSKSB18}]\label{lem:nb_Givens_PH}
    The number of double Givens gates as defined in \cref{lem:lower_block} to achieve the factorization \cref{eq:unitary_U_slater_ph} is {at most} $N(N-1)/2$, whereas the number of particle-hole transformations is {at most} $N$.
    {Furthermore,  if $\det\mathbf{W} = 1$, the number particle-hole transformations is even, else if $\det \mathbf{W} = -1$ it is odd.}
\end{proposition}
\cref{lem:lower_block} is now leveraged to yield the factorization of $\mathbf{W}$ in \eqref{e:w_matrix}.
\begin{lemma}\label{lem:full_from lower}
    The matrix $\mathbf{O}$ in \cref{eq:unitary_U_slater_ph} is a unitary orthogonal complex matrix, and we have the factorization
    \[
        \begin{pmatrix}
            \mathbf{D}\overline{\bV} & \mathbf{0}\\
        \mathbf{0} & \overline{\mathbf{D}}\bV
        \end{pmatrix}
        \mathbf{W} \mathbf{O}^\adjoint =
        \begin{pmatrix}
        \bI & \mathbf{0}\\
        \mathbf{0} & \bI
        \end{pmatrix}.
    \]
\end{lemma}
\begin{proof}
    The discussion above yields the expression of the lower block.
    Since $\mathbf{O}$ and $\mathbf{W}$ are orthogonal complex matrices and unitary as well, the upper block is determined by the lower block in the light of \cref{def:orthogonal_complex}. 
\end{proof}
Furthermore, the $\mathbf{W}$ matrix of \cref{lem:full_from lower} yields a Bogoliubov transformation of the following factorized form
$$
\begin{pmatrix}
    \bm{b}^\adjoint\\
    \bm{b}
\end{pmatrix}
= 
\begin{pmatrix}
    \bV^\top\overline{\mathbf{D}} & \mathbf{0}\\
    \mathbf{0} & \bV^\adjoint \mathbf{D}
\end{pmatrix}
\mathbf{O}
\begin{pmatrix}
    \mathbf{a}^\adjoint\\
    \mathbf{a}
\end{pmatrix}  
$$
This factorization into two factors gives us a recipe to prepare the state $b_1^\adjoint\dots b_r^\adjoint\ket{\emptyset_b}$ by composing two transformations given in the following remarks.
\begin{remark}[particle number conserving transformation] 
    Define $\bQ = (\bV^\top\overline{\mathbf{D}})_{[r]:}$  as the matrix made of the $r$ first rows of $\bV^\top\overline{\mathbf{D}}$. 
    We can directly use the factorization given in \cref{s:circuit_proj_DPP} so that, for $k\in [r]$, we have
    $$
            \corr{\mathcal{U}(\bQ) a_k^\adjoint \mathcal{U}(\bQ)^\adjoint = \sum_{j=1}^{N}\bQ_{kj}
            a^{\adjoint}_j }
    $$
    where $\mathcal{U}(\bQ)$ realizes the product of Givens transformations factorizing $\bQ$.
\end{remark}
\begin{remark}[particle number non-conserving transformation]
    There is a unitary operator $\mathcal{U}_{\mathrm{ph}}(\mathbf{O}) : \mathbb{H}\to \mathbb{H}$ representing the multiplication by $\mathbf{O}$ 
    as follows
    \[
        \begin{pmatrix}
            \mathcal{U}_{\mathrm{ph}}\mathbf{a}^\adjoint\mathcal{U}_{\mathrm{ph}}^\adjoint \\
            \mathcal{U}_{\mathrm{ph}}\mathbf{a}\ \mathcal{U}_{\mathrm{ph}}^\adjoint
        \end{pmatrix}
        = 
        \mathbf{O}
        \begin{pmatrix}
            \mathbf{a}^{\adjoint} \\
            \mathbf{a}
        \end{pmatrix},
    \]
    where the action of $\mathcal{W}$ is entrywise on the vectors $\mathbf{a}^{\adjoint}$ and $\mathbf{a}$.
    By construction, we have
    \begin{align}
        \corr{\mathcal{U}_{\mathrm{ph}} = \mathcal{B}^{\mathtt{m}_1}\mathcal{M}_{1}\mathcal{B}^{\mathtt{m}_2}\mathcal{M}_{2}\mathcal{B}^{\mathtt{m}_3}\dots \mathcal{B}^{\mathtt{m}_{N-1}}\mathcal{M}_{N-1}\mathcal{B}^{\mathtt{m}_N},}\label{eq:W_Givens}
    \end{align}
    \corr{for some $(\mathtt{m}_1, \dots, \mathtt{m}_N) \in \{0,1\}^N$, and} 
    where $\mathcal{B}$ is such that $\mathcal{B}a_N^\adjoint\mathcal{B}^\adjoint = a_N$ while leaving all other creation-annihilation operators unchanged, while $\mathcal{M}_i$ is a composition of \corr{at most} $N-i$ Givens operators.
    Each Givens operator $\Gamma: \mathbb{H}\to \mathbb{H}$ simply represents multiplication by the Givens matrix $\bm{\Gamma}$ as follows
    \[
        \begin{pmatrix}
            \Gamma\mathbf{a}^\adjoint \Gamma^\adjoint\\
            \Gamma\mathbf{a} \Gamma^\adjoint
        \end{pmatrix} 
        = 
        \bm{\Gamma}\cdot
        \begin{pmatrix}
            \mathbf{a}^\adjoint\\
            \mathbf{a}
        \end{pmatrix}.
    \]
    Again, the Givens operators in \cref{eq:W_Givens} appear in reverse order compared with Givens matrices in \cref{eq:unitary_U_slater_ph}.
\end{remark}

\subsubsection{Quantum gates}
We now combine the two remarks above.
Let $\bV$ and $\mathbf{O}$ be given by \cref{lem:lower_block}.
For short, denote $\bQ = (\bV^\top\overline{\mathbf{D}})_{[r]:}$. The factorization of the fermionic Gaussian state \cref{e:fermionic_gaussian_state_W} as a composition of double Givens gates and $X$ gates reads 
\[
    b_1^\adjoint\dots b_r^\adjoint\ket{\emptyset_b} =  \mathcal{W}\cdot  a_1^{\adjoint}\dots a_r^{\adjoint}\ket{\emptyset_{a}}, \text{ with } \mathcal{W} = \corr{\mathcal{U}_{\mathrm{ph}}(\mathbf{O})\cdot \mathcal{U}(\bQ)}.
\]
This result is readily obtained by inserting several times the identity $\mathcal{W}^\adjoint\mathcal{W}$ in the above equation and by noting that $\ket{\emptyset_b} = \mathcal{W}\ket{\emptyset_{a}}$.


%

%% file: experiments.tex
In this section, we demonstrate the circuits discussed in \cref{s:quantum_circuits} on a Qiskit simulator \citep{Qiskit}, and on a few 5-qubit IBMQ machines \citep{IBMQ}.
Python code to reproduce the experiments in available on GitHub.\footnote{\url{https://github.com/For-a-few-DPPs-more/quantum-sampling-DPPs}}


\subsection{Projection DPPs}\label{s:proj_dpp_exp}

\corr{As discussed in \cref{s:dilation_argument}, any DPP can be sampled by constructing a projection DPP on a twice larger ground set; we thus focus here on the simulation of projection DPPs.}
\corr{We take $N=5$ and we} consider a projection DPP on $[N] = \{1, 2, 3, 4, 5\}$. 
The marginal kernel is $\bK=\bQ_{:,[k]}\bQ_{:,[k]}^\top$ \corr{with $k=3$}, where $\bQ$ is a matrix with orthonormal columns obtained by the Householder QR decomposition of a realization of an $N\times N$ random matrix full of independent and identically distributed real Gaussians.
\corr{Almost surely, $\bK$ is thus a rank-$3$ projector, and all the following statements are conditionally on the kernel.}
In particular, the corresponding DPP generates samples of $3$ points with probability one.

\corr{Probabilities of subsets are depicted in the blue histogram of \cref{f:results}.}
\corr{For completeness}, we display $\bK$ in Figure~\ref{f:dpp1_kernel} \corr{where labels of entries\footnote{\corr{In \textsc{python}, the first entry of an array has label zero.}} range from $0$ to $4$ to match Qiskit convention.
Note that it is not obvious from this figure how to guess the rank of the kernel or whether it is projective.
Nonetheless, we can gain some other insights about the DPP by inspection of the color map.
We see for example that item $\{1\}$ has a large inclusion probability $\bK_{1,1}$; see the entry with label $(0,0)$ in dark red in Figure~\ref{f:dpp1_kernel}.
Furthermore, by considering the entries with labels $(0,1)$ and $(0,2)$, we see that $|\bK_{1,2}|$ and $|\bK_{1,3}|$ have small values compared e.g.\ with $|\bK_{1,4}|$; see the darker red entry with label $(0,3)$.
This intuitively indicates that the item with label $0$ and the item with label $3$ have a large ``similarity'' w.r.t.\ this kernel and are not likely to be jointly sampled.
Accordingly, we observe in the histogram in \cref{f:results} that the set with labels $\{0,1,2\}$ has a larger probability to be sampled than the set with labels $\{0,1,3\}$.}

Following the circuit construction of \cite{JSKSB18} with 2-qubit gates acting only on neighbouring qubits on a line, we obtain the circuit shown in Figure~\ref{f:dpp_circuit}, where each gate labeled as ``XX+YY" is a Givens gate, i.e., a sequence of CNOT and $Z$ gates as discussed in \cref{s:gates_details}. 
\corr{The circuit of Figure~\ref{f:dpp_circuit} starts on the left side by three $X$ gates which excite three fermions on the Fock vacuum, namely $a_1^* a_2^* a_3^* \ket{0}$.
Next, Givens rotations are applied to this state to yield $b_1^* b_2^* b_3^* \ket{0}$ as a result of \cref{p:jiang}.
The parameters of the corresponding gates are obtained by the QR decomposition of \cref{s:parallel_QR} as implemented in Qiskit.}

Now, before it can be run on a particular machine, the circuit needs to be transpiled, i.e.\ written as an equivalent sequence of gates that correspond to what can be physically implemented on the machine. 
We show in Figure~\ref{f:calibration_data} the calibration data for three $5$-qubit IBMQ machines: \ibm{lima}, \ibm{quito}, and \ibm{manila}.
This calibration takes the form of a graph, where nodes represent qubits, and edges represent the possibility of applying a CNOT gate, the only two-qubit gate in our original circuit in \cref{f:dpp_circuit}.
While \ibm{manila} can implement CNOT gates between neighbouring qubits on a line, as implicitly assumed in the construction of \cite{JSKSB18}, the two other machines have a T-structured graph that will force the transpiled circuit to look quite different from \cref{f:dpp_circuit}.
Indeed, we show the transpiled circuits in the first three panels of \cref{f:transpiled_dpp_circuits}.
Note how \ibm{manila} uses CNOT gates between neighbouring qubits, as in the original circuit, resulting in overall fewer CNOT gates than on \ibm{quito} and \ibm{lima}. 
The latter two machines cannot afford CNOT gates jointly acting on qubits 2 and 3, for instance, and end up compensating for this by using more CNOT gates in total. 
Since CNOT gates are among the most error-prone manipulations, minimizing the number of CNOT gates is an important feature.
Intuitively, had we known in advance that we would run the circuit on a machine with a particular graph, we should have designed the circuit in \cref{f:dpp_circuit} differently, by rather running parallel QR with Givens rotations only along edges of the machine's graph. 

Before we observe the results, we note that the three machines come with different characteristics.
For instance, \ibm{manila} has overall the lowest readout errors, and \ibm{lima} the largest.
It is common to summarize the characteristics of a machine in a single number $2^{d_\text{max}}$, called \emph{quantum volume}, where $d_\text{max}$ is --loosely speaking-- the depth of a square circuit that we can expect to run reliably. 
The volume reported by IBM is obtained numerically, using a procedure known as randomized benchmarking; see Appendix~\ref{s:errors_in_quantum_computers}. 
The machines \ibm{manila}, \ibm{quito}, and \ibm{lima} respectively have reported volumes 32, 16, and 8.
Thus, even on \ibm{manila}, the transpiled circuit is much larger than the ``guaranteed" $5\times 5$ square circuit, and we should expect noise in our measurements, as we shall now see.

\begin{figure}
    \centering
    \begin{subfigure}{0.48\textwidth}
        \includegraphics[width=\textwidth]{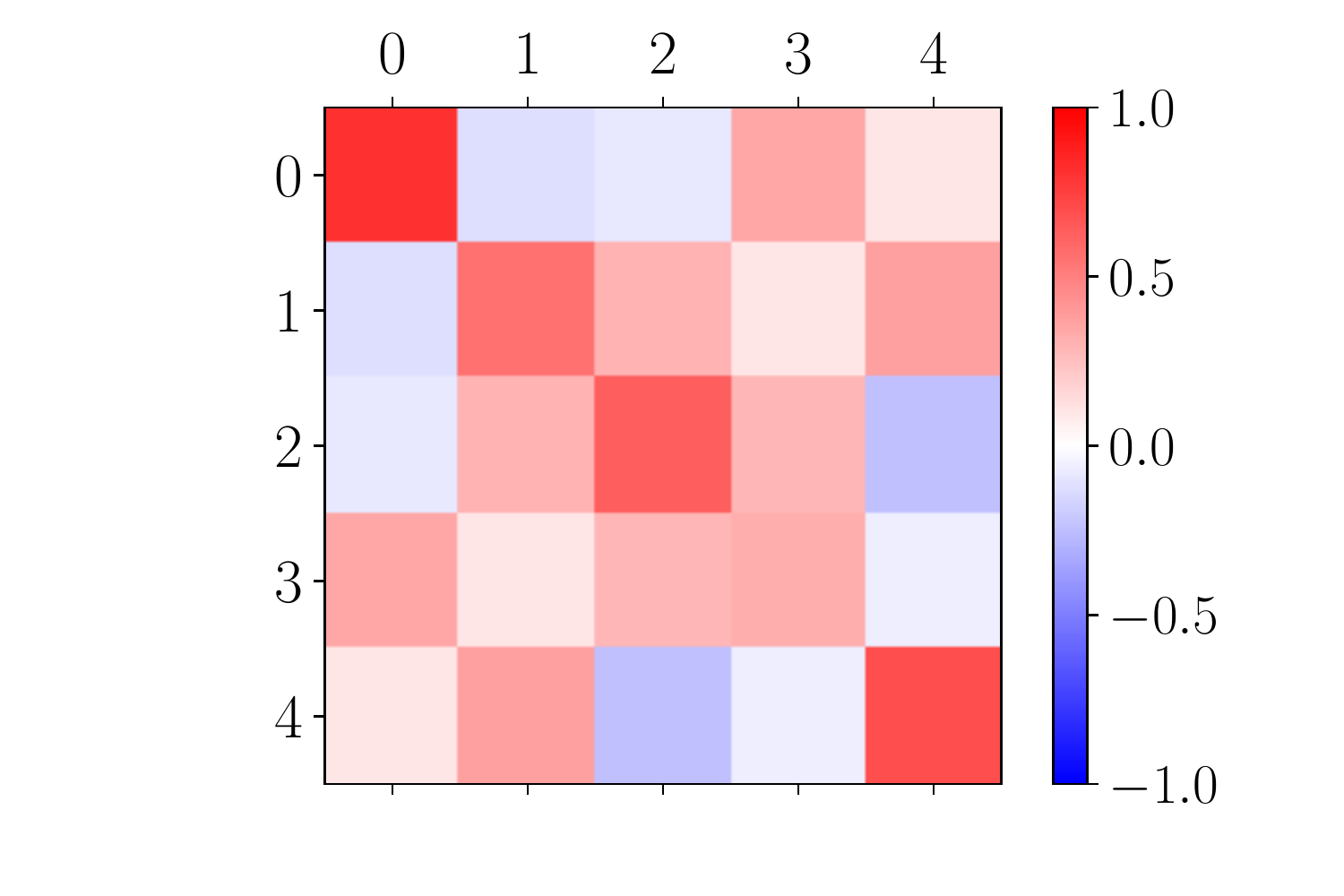}
        \caption{DPP kernel}
        \label{f:dpp1_kernel}
    \end{subfigure}
    \begin{subfigure}{0.48\textwidth}
        \includegraphics[width=\textwidth]{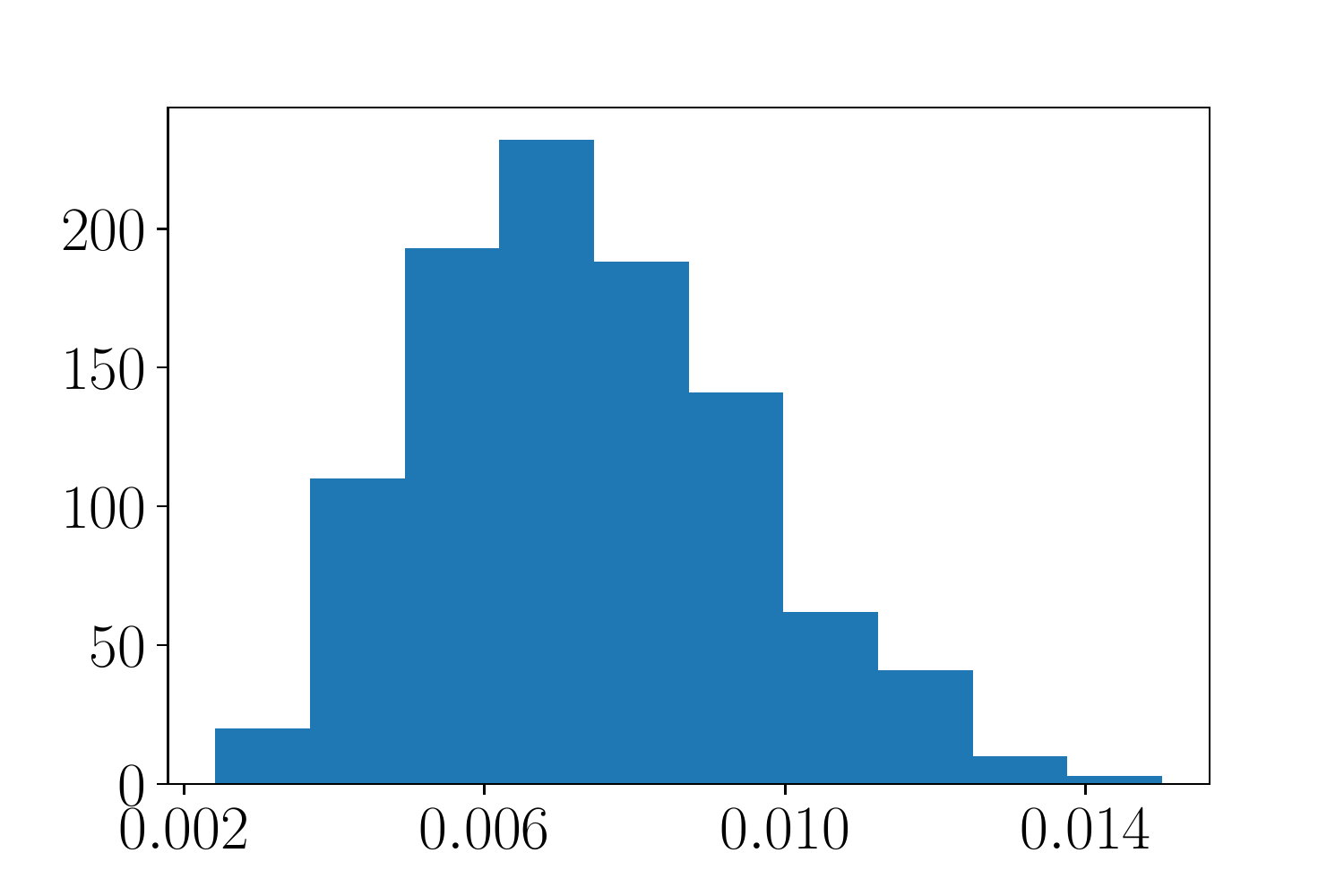}
        \caption{TV distance under the null}
        \label{f:dpp1_tv}
    \end{subfigure}
    \caption{\corr{Illustration of the DPP kernel matrix associated with the circuit in \cref{f:dpp_circuit} (left-hand side) and histogram of the total variation distance between the distribution of the simulated DPP and the exactly known distribution (right-hand side).}}
\end{figure}

\begin{figure}
    \centering
    \begin{subfigure}{0.235\textwidth}
        \includegraphics[width=\textwidth]{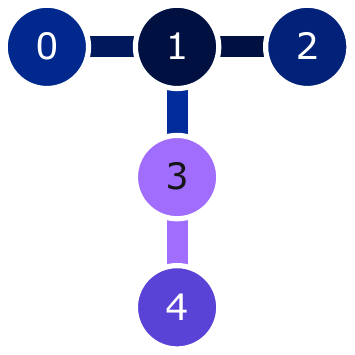}
        \caption{\ibm{lima}}
        \label{f:calibration_data_lima}
    \end{subfigure}
    \hfill
    \begin{subfigure}{0.235\textwidth}
        \includegraphics[width=\textwidth]{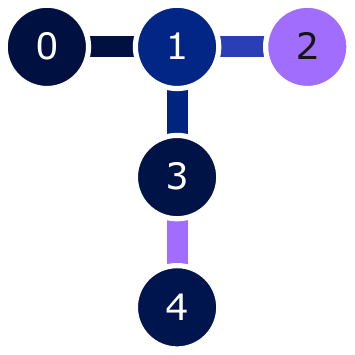}
        \caption{\ibm{quito}}
    \end{subfigure}
    \hfill
    \begin{subfigure}{0.4\textwidth}
        \includegraphics[width=\textwidth]{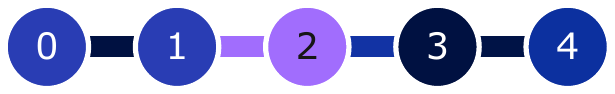}
        \caption{\ibm{manila}\label{f:calibration_data_manila}}
    \end{subfigure}
    \caption{Calibration data for three IBMQ machines: \ibm{lima}, \ibm{quito}, and \ibm{manila}. 
    Dark blue is low, light violet is high, but colors are not comparable across subfigures, although the orders of magnitude are similar. 
    Node colors indicate readout errors (all of the order of $10^{-2}$), edge colors indicate CNOT error, of the same $10^{-2}$ order of magnitude for \ibm{lima} and \ibm{quito}, but going down to $5\times 10^{-3}$ for \ibm{manila}. 
    All machines have 5 qubits, but their volumes are respectively 8, 16, and 32.\label{f:calibration_data}}
\end{figure}

\begin{figure}
    \centering
    \begin{subfigure}{0.24\textwidth}
        \includegraphics[angle=90,origin=c,height=10cm]{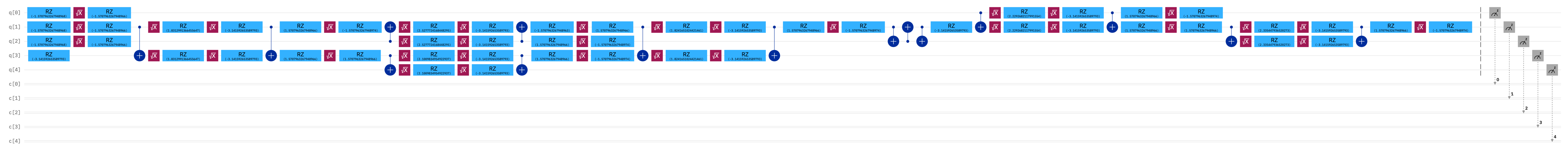}
        \caption{\ibm{lima}\label{f:transpiled_dpp_circuits_lima}}
    \end{subfigure}
    \begin{subfigure}{0.24\textwidth}
        \includegraphics[angle=90,origin=c,height=10cm]{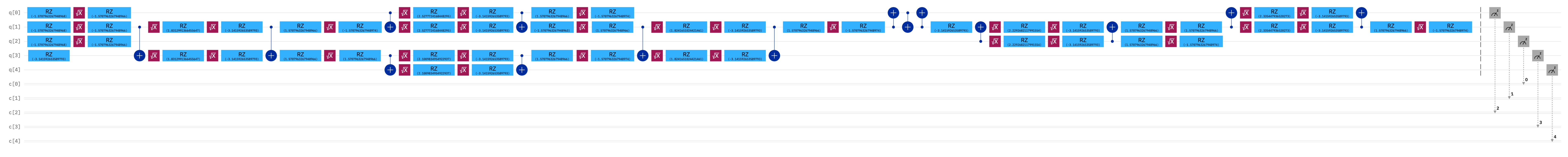}
        \caption{\ibm{quito}\label{f:transpiled_dpp_circuits_quito}}
    \end{subfigure}
    \begin{subfigure}{0.24\textwidth}
        \includegraphics[angle=90,origin=c,height=10.2cm]{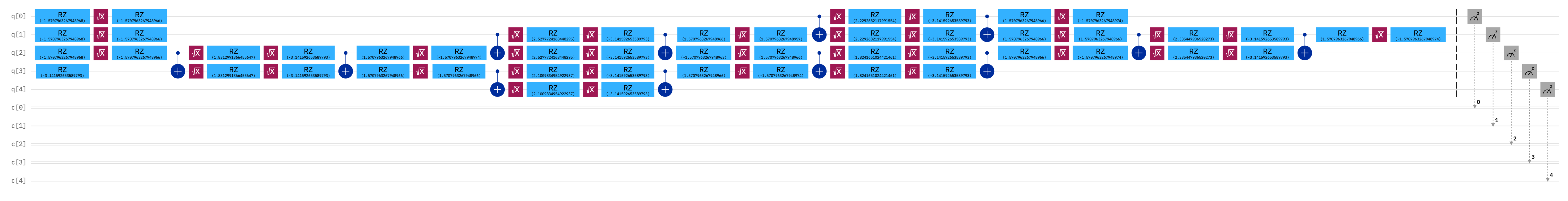}
        \caption{\ibm{manila}\label{f:transpiled_dpp_circuits_manila}}
    \end{subfigure}
    \begin{subfigure}{0.24\textwidth}
        \includegraphics[angle=90, height=18cm]{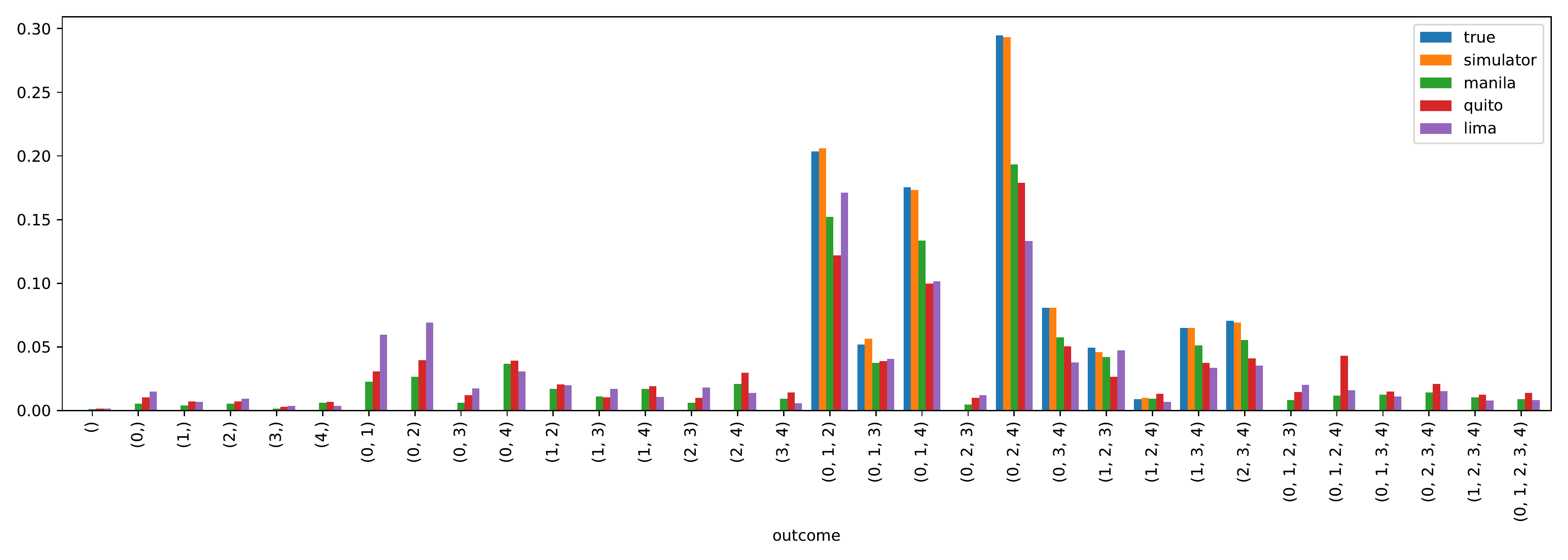}  
        \caption{Results}
        \label{f:results}
    \end{subfigure}
    \caption{Transpiled circuits corresponding to the input in Figure~\ref{f:dpp_circuit} for three IBMQ machines: \ibm{lima}, \ibm{quito}, and \ibm{manila}.
    \cref{f:results} shows the corresponding empirical distribution.}
    \label{f:transpiled_dpp_circuits}
\end{figure}

\cref{f:results} shows the empirical distributions corresponding to independently preparing the input and measuring the output of the transpiled circuits in \cref{f:transpiled_dpp_circuits}, $20,000$ times each.
In blue and orange, we show for reference the probability under DPP($\bK$) of the corresponding subset, as well as the empirical frequencies coming from sampling the output of the circuit using a simulator, which amounts to independently drawing $20,000$ samples of the DPP on a classical computer.

The empirical measure of the classical samples in orange is close to the underlying DPP, as testified by the total variation (TV) distance\footnote{Recall that the TV is the maximum absolute difference between the probabilities assigned to a subset, where the maximum is taken across subsets.} between the two distributions, which is below $10^{-2}$. 
Actually, if we repeatedly and independently draw sets of $20,000$ samples, we obtain the empirical distribution of the TV distance shown in Figure~\ref{f:dpp1_tv}. 
In contrast, the TV distance between the empirical measure obtained from $20,000$ runs on \ibm{manila} of the corresponding transpiled circuit is $0.27$, while it is $0.39$ and $0.40$ for the T-structured machines \ibm{lima} and \ibm{quito}, respectively. 
Looking at \cref{f:dpp1_tv} again, a test based on the estimated TV distance would easily reject even the hypothesis that at least one of the quantum circuits samples from the correct distribution. 
Moreover, we confirm (not shown) that the difference in TV between \ibm{manila} and its competitors is significant, which confirms the intuition coming from its larger volume and smaller number of CNOT gates, due to a QR decomposition adapted to its qubit communication graph.
Finally, although a test would reject that the quantum circuits sample from the correct DPP, the resulting distributions are still close to the target DPP, especially for \ibm{manila}, as confirmed by their respective TV.
Interestingly, the quantum circuits actually yield point processes that are supported on (almost) all subsets of $\{1, \dots, 5\}$, while the target DPP, being of projection, only charges subsets of cardinality 3.
The noise does respect the structure of the DPP, somehow: two-item subsets that (wrongly) appear in the support of the empirical measures correspond to items that are marginally favoured by the DPP, as can be seen on the diagonal of the kernel in \cref{f:dpp1_kernel}. 
Intuitively, the appearance of a subset of cardinality $2$ is partly due to readout error on one of the qubits supposed to output $1$. 
This intuition is confirmed by the calibration data: for \ibm{lima}, for instance, \cref{f:calibration_data_lima} shows that the qubit labeled `3' has a large readout error; simultaneously, there is a deficit of appearance of the triplet with labels $\{0,3,4\}$ in its empirical distribution in Figure~\ref{f:results}, while $\{0,4\}$ wrongly appears.

\paragraph*{Optimal QR for a T-structured communication graph.}
The QR-inspired fermionic circuits implemented in Qiskit follow \citep{JSKSB18}, and thus use Givens rotations between neighbouring columns.
While this suits the calibration data of \ibm{manila}, which has a linear qubit-communication graph, \ibm{lima} and \ibm{quito} rather have a communication graph shaped as a $T$.
As a result, transpilation is less straighforward and yields a bigger circuit than for \ibm{manila}.
In particular, the transpiled \ibm{quito} circuit in \cref{f:transpiled_dpp_circuits_quito} has 15 CNOT gates, while both the original circuit in \cref{f:dpp_circuit} and the transpiled \ibm{manila} circuit have only $12$ CNOT gates.

As we discussed in \cref{s:quantum_circuits}, a shorter (and less error-prone) transpiled circuit would intuitively result from a QR decomposition that respects the qubit-communication graph. 
For concreteness, we give a QR decomposition that is better suited to the communication graph of \ibm{quito}, with the mapping to the columns of $\mathbf{Q}$ given in \cref{f:col_to_qubit_T_structured}.
Assuming no preprocessing, we find
\begin{align}
    \bQ 
    &\rightarrow 
    \begin{pmatrix}
        * & * & \bzero&  *  & \bzero \\
        * & * & * & * & *  \\
        * & * & * & * & *  \\
    \end{pmatrix}
    \rightarrow  
    \begin{pmatrix}
        * & * & 0&  \bzero  & 0 \\
        * & * & * & * & *  \\
        * & * & * & * & *  \\
    \end{pmatrix}\nonumber\\
    &\rightarrow 
    \begin{pmatrix}
        \underline{\lambda_1} & \bzero & 0&  0  & 0 \\
        \underline{0} & * & * & * & \bzero  \\
        \underline{0} & * & * & * & *  \\
    \end{pmatrix}
    \rightarrow 
    \begin{pmatrix}
        \lambda_1 & 0 & 0&  0  & 0 \\
        0 & * & * & \bzero & 0  \\
        0 & * & * & * & *  \\
    \end{pmatrix}\nonumber\\
    &\rightarrow 
    \begin{pmatrix}
        \lambda_1 & 0 & 0&  0  & 0 \\
        0 & \bzero & \underline{\lambda_2} & 0 & 0  \\
        0 & * & \underline{0} & * & \bzero  \\
    \end{pmatrix}
    \rightarrow 
    \begin{pmatrix}
        \lambda_1 & 0 & 0&  0  & 0 \\
        0 & 0 & \lambda_2 & 0 & 0  \\
        0 & \underline{\lambda_3} & 0 & \bzero & 0  \\
    \end{pmatrix}.\label{e:QR_for_T_graph}
\end{align}
While not necessary optimal in any sense, our guiding principle for the decomposition \cref{e:QR_for_T_graph} is to fill the matrix with zeros such that, for each row, the final complex phase appears at a node which, if removed from the graph, leaves the resulting graph connected.
Note that to fall back onto a ``diagonal" matrix, although this last step is unnecessary for DPP sampling, a final permutation between the second and third columns can be realized by an extra Givens gate with $\theta=0$ and $\phi = \pi/2$ (i.e., a so-called ISwap gate) between qubit $1$ and qubit $2$.
Overall, the sequence of Givens rotations corresponding to \cref{e:QR_for_T_graph} can be transpiled on \ibm{quito} or \ibm{lima} with only the expected 2-per-rotation CNOT gates, since all rotations are applied to neighbours in the graph.
We leave the characterization and benchmarking of the optimal QR decomposition for a given qubit communication graph for future work.
\begin{figure}[h]
    \centering
    \begin{tikzpicture}
        \node[shape=circle,draw=black] (A) at (0,0) {0};
        \node[shape=circle,draw=black] (B) at (1,0) {1};
        \node[shape=circle,draw=black] (C) at (2,0) {2};
        \node[shape=circle,draw=black] (D) at (3,0) {3};
        \node[shape=circle,draw=black] (E) at (4,0) {4};
        \path [-] (A) edge (B);
        \path [-] (B) edge (C);
        \path [-] (D) edge (E);
        \draw[-] (B) to [bend left=40] (D);
        \node (A1) at (0,-1.1) {$\begin{pmatrix}
            * \\
            *  \\
            *
        \end{pmatrix}$};
        \node (B1) at (1,-1.1) {$\begin{pmatrix}
            * \\
            *  \\
            *
        \end{pmatrix}$};
        \node (C1) at (2,-1.1) {$\begin{pmatrix}
            * \\
            *  \\
            *
        \end{pmatrix}$};
        \node (D1) at (3,-1.1) {$\begin{pmatrix}
            * \\
            *  \\
            *
        \end{pmatrix}$};
        \node (E1) at (4,-1.1) {$\begin{pmatrix}
            * \\
            *  \\
            *
        \end{pmatrix}$};
        \node (A2) at (0,-2.1) {$\mathbf{Q}_{:,1}$};
        \node (B2) at (1,-2.1) {$\mathbf{Q}_{:,2}$};
        \node (C2) at (2,-2.1) {$\mathbf{Q}_{:,3}$};
        \node (D2) at (3,-2.1) {$\mathbf{Q}_{:,4}$};
        \node (E2) at (4,-2.1) {$\mathbf{Q}_{:,5}$};
        \node[right] (F1) at (5,-1.1) {columns of $\mathbf{Q}$};
        \node[right] (G1) at (5,0) {qubit communication graph};
    \end{tikzpicture}
\caption{Allowed CNOT gates between qubits in the T-structured communication graph of \ibm{lima} or \ibm{quito}, for a map sending column $1$ $\mapsto$ qubit $0$, column $2$ $\mapsto$ qubit $1$, etc.
\label{f:col_to_qubit_T_structured}}
\end{figure}
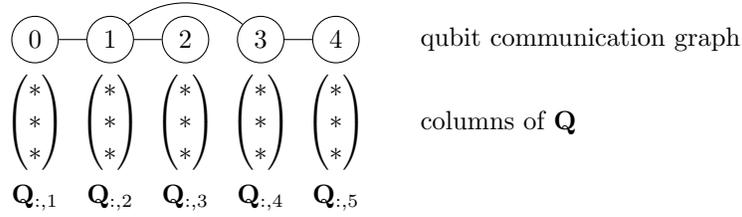
\subsection{Pfaffian PPs}
In this section, we illustrate the second step of the algorithm of \cref{s:quantum_circuit_PfPPs} to sample PfPPs of the type described by \citet{koshida21} and associated with a pure state of the form $\prod_{i\in \mathcal{C}}b_i^\adjoint\ket{\emptyset_b}$, namely, for which the matrix $\bS$ is projective as explained in \cref{rem:projective_S_for_PfPP}.

To construct $\bS$, we consider the quadratic form $\mathbf{H}_{\mathrm{BdG}}$ given by the block matrix in \cref{e:H_BdG} with $N = 5$ with the Hermitian and skew-symmetric part given respectively by
\[
    \mathbf{M}
    =
    \begin{pmatrix}
        1 & 0.5& 0.2& 0.2& 0.2\\
        0.5& 1& 0.5& 0.2& 0.2\\
        0.2& 0.5& 1& 0.5& 0.2\\
        0.2& 0.2& 0.5& 1& 0.5\\
        0.2& 0.2& 0.2& 0.5& 1
    \end{pmatrix}
    \text{ and }
    \bm{\Delta} 
    =
    \begin{pmatrix}
        0& 1& 0& 0& 0\\
        -1& 0& 1& 0& 0\\
        0& -1& 0& 1& 0\\
        0& 0& -1& 0& 1\\
        0& 0& 0& -1& 0
    \end{pmatrix}.
\]
\corr{Let $0\leq \epsilon_1\leq \dots\leq  \epsilon_N$ be eigenvalues of $H_{\mathrm{BdG}}$ determined by the eigendecomposition \cref{e:eigen_decomp_HBDG} and let $\mathbf{W}$ be the orthogonal complex matrix associated with the diagonalization.}
Next, we select the subset $\mathcal{C}$ of the $3$ smallest positive eigenvalues of $\mathbf{H}_{\mathrm{BdG}}$.
Let $\mathbb{K}$ be the kernel associated with the projection
\begin{equation}
    \mathbf{S} = \overline{\mathbf{W}}^{\adjoint}
    \begin{pmatrix}
         \mathbf{1}_{\bar{\mathcal{C}}}& \mathbf{0}\\
        \mathbf{0} & \mathbf{1}_{\mathcal{C}}
    \end{pmatrix}
    \overline{\mathbf{W}} \text{ with } \mathcal{C}= \{1,2,3\},\label{eq:S_for_exp}
\end{equation}
by virtue of \cref{prop:existence_Koshida}.
The Hamiltonian quadratic form and the associated Pfaffian kernel are displayed in \cref{f:M_Delta}, whereas the corresponding circuit is displayed in \cref{fig:circuit_pfpp}.
The probability mass function of $Y\sim \pfPP(\mathbb{K})$ has the following simple expression:
\begin{align}
    \mathbb{P}(Y = S)= (-1)^{|\overline{S}|}
    \pf(\mathbb{K} - \mathbb{J}_{\overline{S}}),\label{e:pf_L_ensemble}
\end{align}
where $(\mathbb{J}_{\overline{S}})_{ij} = \delta_{ij}1(i\in\overline{S})\left(\begin{smallmatrix}
    0 & 1\\
    -1 & 0
\end{smallmatrix}\right)$, with $\overline{S}$ being the complement of $S$ in $[N]$; we refer to \cref{a:proof_pmf} for a short proof.
\begin{figure}[h]
    \centering
    \begin{subfigure}[b]{0.45\textwidth}
        \includegraphics[width=\textwidth]{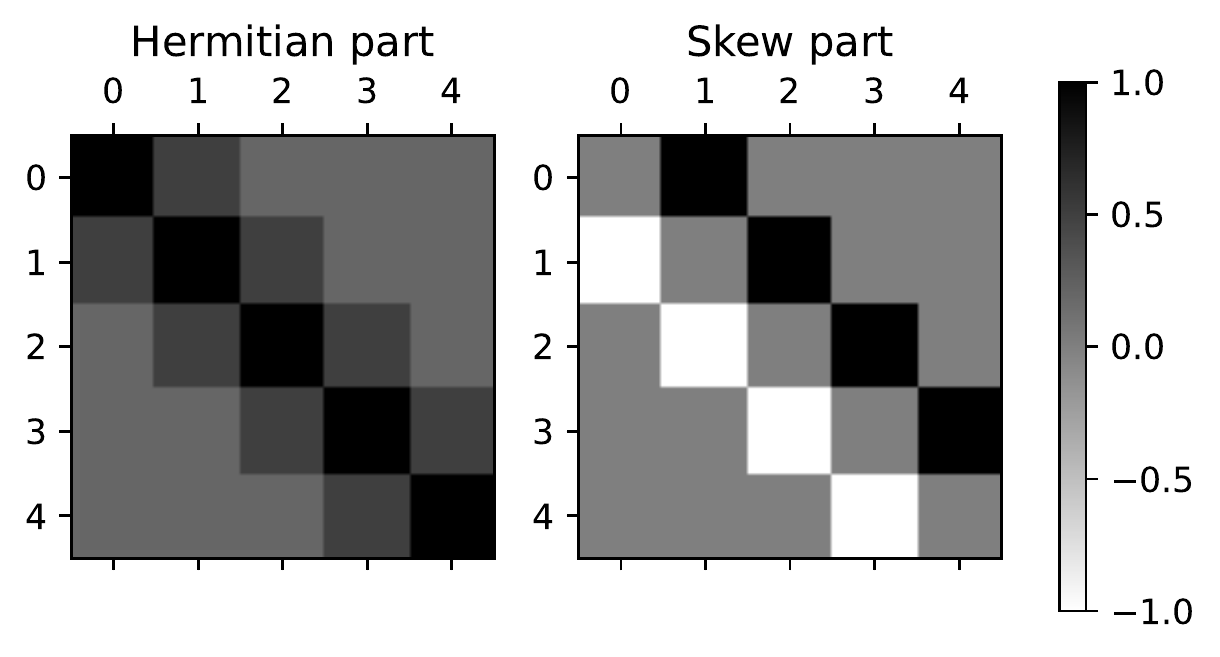}
        \caption{$\mathbf{M}$ and $\bm{\Delta}$.}
    \end{subfigure}
    \hfill
    \begin{subfigure}[b]{0.45\textwidth}
    \includegraphics[width=\textwidth]{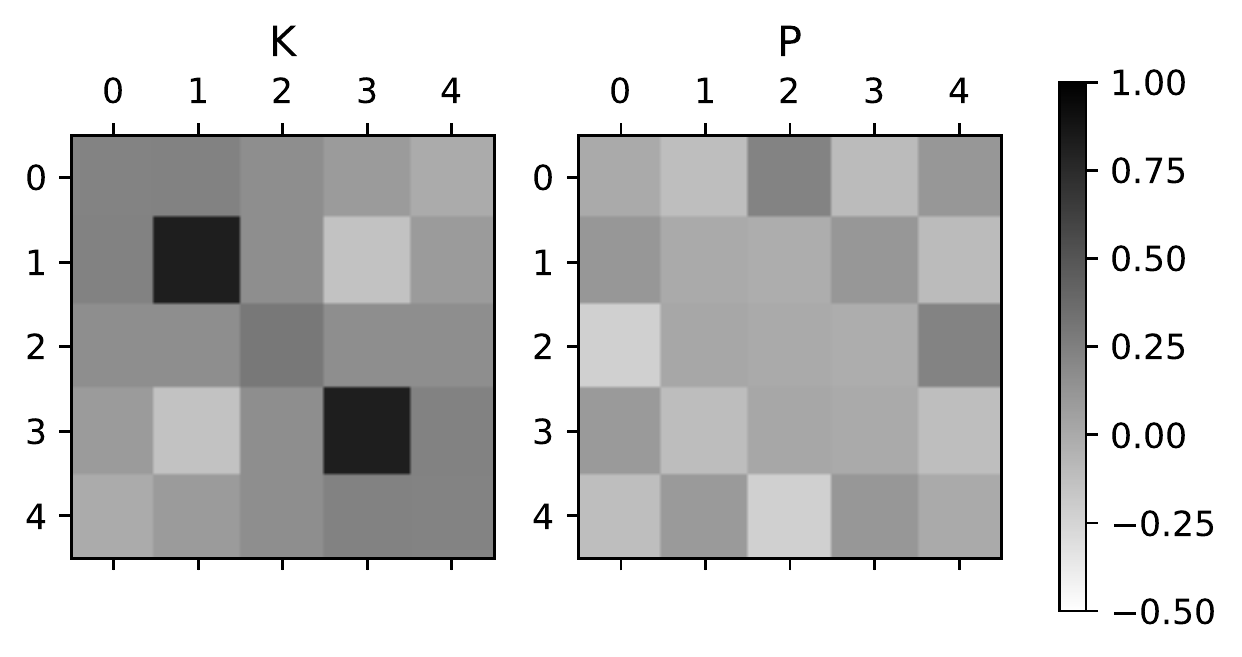}
    \caption{$\bK$ and $\mathbf{P}$.\label{f:K_P_pfPP}}
    \end{subfigure}
    \caption{Visualization of the Hamiltonian $\mathbf{H}_{\mathrm{BdG}}$ and the components defining the Pfaffian kernel $\mathbb{K}$ \corr{given in \cref{eq:Pfaffian_kernel} as a function of $\mathbf{K}$ and $\mathbf{P}$.}.}
    \label{f:M_Delta}
\end{figure}
\begin{figure}[h]
    \centering
    \includegraphics[width=\textwidth]{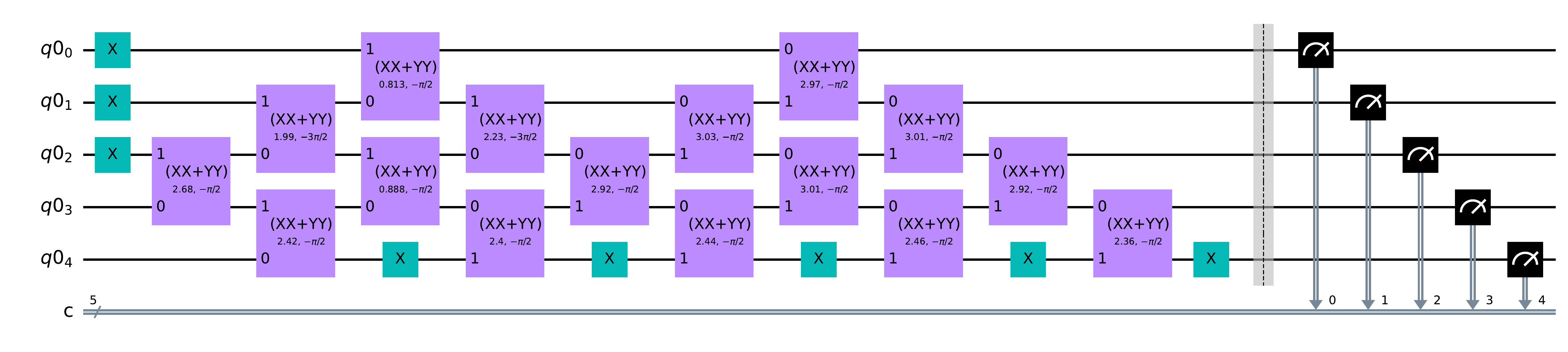}
    \caption{\corr{Circuit corresponding to $\pfPP(\pfK)$ associated to the projection matrix \cref{eq:S_for_exp}. 
    The parity of the number of $X$ gates determines the parity of the number of points in samples of $\pfPP(\pfK)$.}\label{fig:circuit_pfpp}}
\end{figure}
For these numerical experiments, we only compare (classical) simulations of the quantum circuits to the true point process.
As in \cref{s:proj_dpp_exp}, $20,000$ samples are drawn independently from $Y\sim\pfPP(\pfK)$ thanks to the Qiskit simulator.
Next, the empirical frequencies of subsets of $[N]$ are computed and compared with the probabilities \cref{e:pf_L_ensemble}.

In \cref{f:results_pfpp}, we observe that the empirical probabilities match the expected case, with a TV distance of $0.009$.
Furthermore, it is manifest from \cref{f:K_P_pfPP} that there is a weak repulsion between items $2$ and item $4$ -- see the entry $(1,3)$ in the grayscale matrix -- and that each of these two elements has a large marginal probability.
Hence, they can be expected to be sampled together.
This is confirmed in \cref{f:results_pfpp} where the subset $\{2,4\}$ (corresponding to the label $(1,3)$ \corr{in the histogram}) corresponds to a large mass under the empirical measure.
Also, note that all subsets naturally have the same parity.
\begin{figure}[h]
    \centering
    \includegraphics[width=0.8\textwidth]{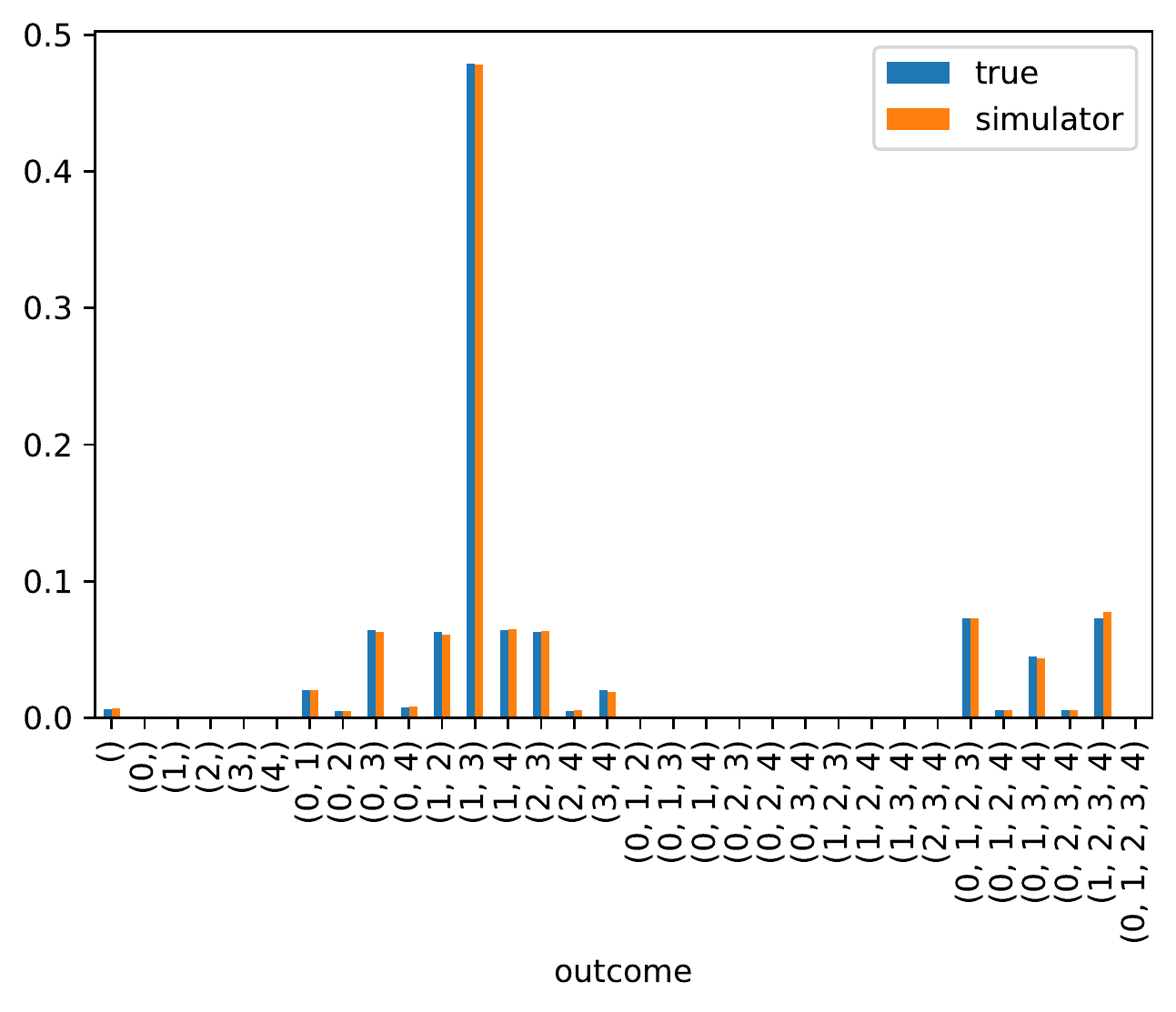}
    \caption{Histogram of the empirical frequencies of the subsets sampled by $\pfPP(\pfK)$ \corr{associated to the projection matrix \cref{eq:S_for_exp}} \emph{versus} subset probabilities \corr{as given by the Pfaffian formula} \cref{e:pf_L_ensemble}.}
    \label{f:results_pfpp}
\end{figure}
For completeness, the histogram of the total variation distance between the ground truth distribution and the estimated distribution, computed over $1000$ runs, is displayed in \cref{f:tv_pfpp}.
\begin{figure}[h]
    \centering
    \includegraphics[width=0.4\textwidth]{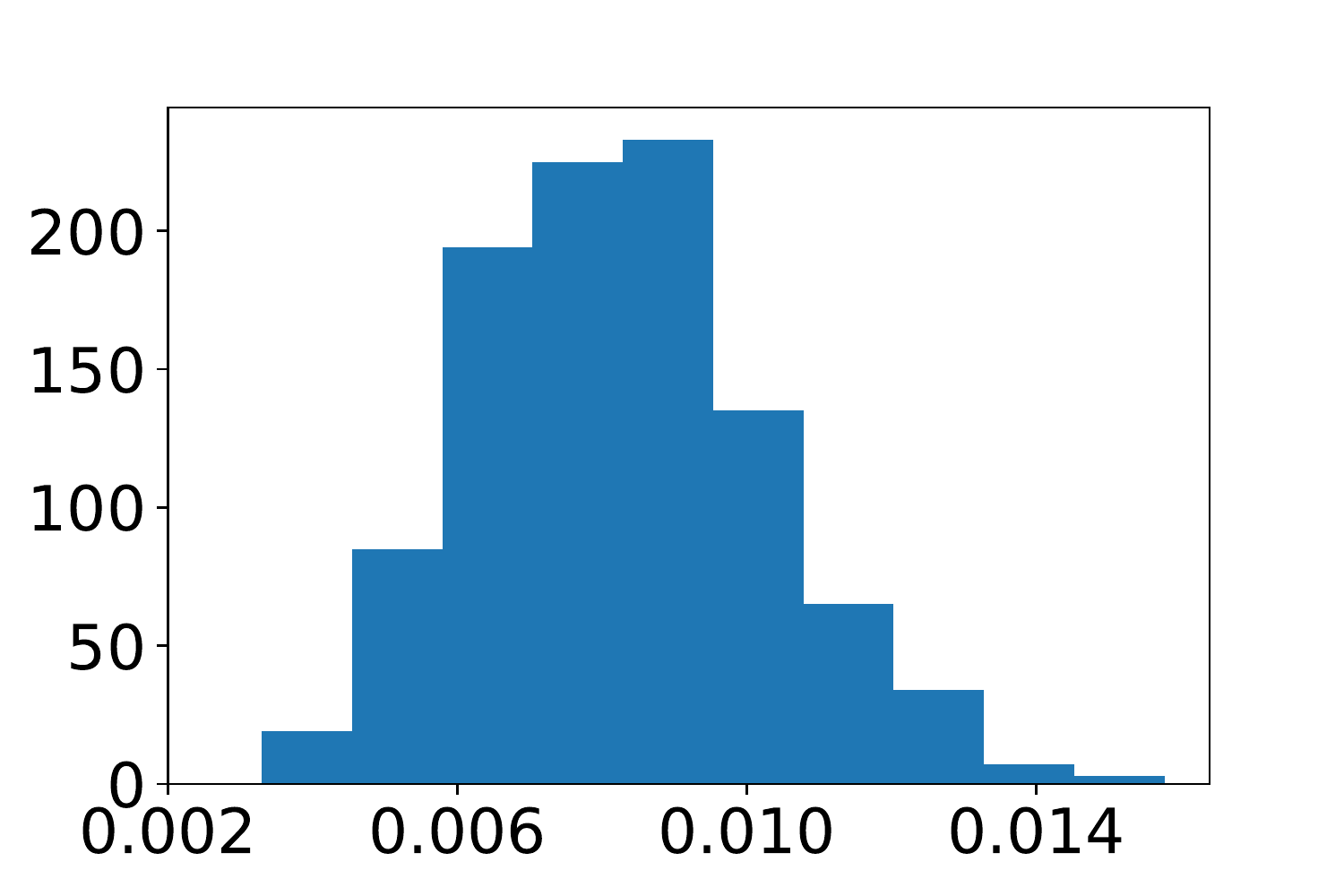}
    \caption{Empirical distribution of the total variation distance between the true probability mass function of t$\pfPP(\pfK)$ \corr{associated to the projection matrix  \cref{eq:S_for_exp}} and its simulation.}
    \label{f:tv_pfpp}
\end{figure}

%% file: discussion.tex
Inspired by the pioneering work of \citet{Mac72} on point processes in quantum physics, we have studied quantum algorithms for DPP sampling.
We did so by reducing DPP sampling to a fermion sampling problem, and then leveraging and modifying recent algorithms for fermion sampling \citep{WHWCNT15,JSKSB18}.
While many of the steps are either common lore in one or the other field, or recently published material, we believe that there is value in a self-contained survey of how to reduce a finite DPP to a fermion sampling problem, all the way from the mathematical object to the implementation on the quantum machine. 
We hope that this paper can help spark further cross-disciplinary work. 
Moreover, writing down all the steps from a DPP as specified in machine learning to its quantum implementation has allowed us to make contributions on top of the survey, like the extension of the argument to a class of Pfaffian PPs, and the first steps in adapting the QR-decomposition behind fermion sampling algorithms to qubit communication constraints.
This opens several research avenues.




First, as mentioned in the introduction, projective DPPs can also be sampled thanks to the Clifford loaders defined in \citep{KerPra22}, introduced independently of the physics literature that we cover in this paper.
Yet the structure of the arguments is related, and it would be enlightening to explicitly compare, on the one hand, the parallel implementation we give in \cref{s:parallel_QR} with a circuit depth logarithmic in $N$ and, on the other hand, data loaders of \citet{KerPra22} with a similar depth.

Second and in the same spirit, an interesting extension of this work would be to develop an algorithm optimally matching any qubit communication graph to a QR decomposition scheme, where by \emph{optimal} we mean minimizing e.g.\ the total variation distance between the output of the circuit and the original point process. 
This would generalize the case of the $T$ graph described in \cref{s:proj_dpp_exp}, and lead to transpilers with fewer noisy gates.
A potential strategy would be to follow the approach of \cite{FreBru19}, who optimize a sparsity-inducing objective to approximate a unitary matrix as a product of Givens rotations.

Third, a natural improvement of the proposed classical preprocessing followed by circuit construction would be the inclusion of the kernel diagonalization \emph{in} the quantum circuit, using for instance the recent developments about quantum SVD \citep{RSML18}.
The combination of this quantum preprocessing and a QR-based circuit would constitute a turn-key sampling pipeline.

A fourth and maybe more speculative research perspective would be to leverage our knowledge of all the correlation functions of point processes such as DPPs, PfPPs, and permanental PPs \citep{JAQK20} to develop a statistical test of quantum decoherence in a given machine.
In particular, our aim would be to design statistics of PPs which are sensitive to the different kinds of errors affecting a quantum computer.

Finally, from a mathematical perspective, we think it is worth exploring in more depth the structure of PfPPs.
While potentially offering more modeling power in machine learning applications, they have received little attention, likely due to their high sampling cost on a classical computer. 
Since the sampling overhead on a quantum computer is minor, they are likely to become a valuable addition to the machine learner's toolbox. 
For starters, we are unaware of a formula for the probability mass function \cref{e:pf_L_ensemble} taking the form of a pfaffian of a likelihood matrix.
Such a construction would generalize the \emph{extended} L-ensemble construction in \citep{TBUA23}.

%% file: appendix.tex
\section{Mathematical details}
\label{s:mathematical_details}
\subsection{More about Givens operators}
\label{a:more_about_Bogoliubov}
We briefly explain here how the unitary operator $\mathcal G$ representing a Givens rotation in \cref{e:givens_operator_definition} can be constructed.
A unitary operator $\mathcal{V}: \mathbb{H} \to \mathbb{H}$ such that
\begin{align*}
    \begin{pmatrix}
        \mathcal{V} a_{1}^\adjoint\mathcal{V}^\adjoint\\
        \mathcal{V} a_{2}^\adjoint\mathcal{V}^\adjoint
    \end{pmatrix}
    = 
    \begin{pmatrix}
        \cos \theta   & \sin \theta\\
         -\sin \theta &  \cos \theta 
    \end{pmatrix}
    \begin{pmatrix}
        a_{1}^\adjoint\\
        a_{2}^\adjoint
    \end{pmatrix}
\end{align*}
is given by
$
    \mathcal{V} = \exp \left(\theta (a_{2}^\adjoint a_{1} - a_{1}^\adjoint a_{2})\right),
$
where we used the BCH formula \citep{Hal10}.
Thanks to this observation, we readily have that 
$$
\mathcal{G} = \mathcal{T}\mathcal{V} \mathcal{T}^*,
$$ 
with $\mathcal{T} = \exp(\i \phi (a_{1}^\adjoint a_{1} - a_{2}^\adjoint a_{2})/2)$, realizes the Givens transformation
\begin{align*}
    \begin{pmatrix}
        \mathcal{G} a_{1}^\adjoint\mathcal{G}^\adjoint\\
        \mathcal{G} a_{2}^\adjoint\mathcal{G}^\adjoint
    \end{pmatrix}
    = 
    \begin{pmatrix}
        \cos \theta   & e^{-\i \phi}\sin \theta\\
         -e^{\i \phi}\sin \theta &  \cos \theta 
    \end{pmatrix}
    \begin{pmatrix}
        a_{1}^\adjoint\\
        a_{2}^\adjoint
    \end{pmatrix}.
\end{align*}

\subsection{Proof of \cref{lem:bilinears_BdG}} \label{proof:bilinears_BdG}
    We begin by a remark about the diagonalization of $H_{\mathrm{BdG}}$.
    Let $\bF = \Big(\begin{smallmatrix}
        \mathbf{0} & \bI \\
        \bI & \mathbf{0}
    \end{smallmatrix}\Big)$.
    As a consequence of~\cref{lem:diag_H_BdG}, we can write
    \begin{align*}
        H_{\mathrm{BdG}} = \frac{1}{2}\begin{pmatrix}
            \mathbf{c}^\adjoint \\
            \mathbf{c}
        \end{pmatrix}^\adjoint 
        \begin{pmatrix}
            -\overline{\mathbf{M}}& -\overline{\bm{\Delta}}\\
            \bm{\Delta} &  \mathbf{M}
        \end{pmatrix}
        \begin{pmatrix}
            \mathbf{c}^\adjoint \\
            \mathbf{c}
        \end{pmatrix}.
    \end{align*}
    The diagonalization of $H_{\mathrm{BdG}}$ reads
    \[
        H_{\mathrm{BdG}} = \frac{1}{2}
        \begin{pmatrix}
            \mathbf{b}^\adjoint \\
            \mathbf{b}
        \end{pmatrix}^\adjoint
        \begin{pmatrix}
            -\Diag(\epsilon_k) & \mathbf{0}\\
            \mathbf{0} & \Diag(\epsilon_k)
        \end{pmatrix}
        \begin{pmatrix}
            \mathbf{b}^\adjoint \\
            \mathbf{b}
        \end{pmatrix}
        .
    \]
    Note that the expectation value of the bilinears are
    \begin{align*}
        \begin{pmatrix}
            (\expval**{c_i c_j^\adjoint}_{\rho} )_{i,j} &(\expval**{c_i c_j}_{\rho} )_{i,j} \\
            (\expval**{c_i^\adjoint c_j^\adjoint}_{\rho} )_{i,j}&(\expval**{c_i^\adjoint c_j}_{\rho} )_{i,j} 
    \end{pmatrix}
        &= \mathbf{W}^\top
        \begin{pmatrix}
            (\expval**{b_i b_j^\adjoint} )_{i,j} & (\expval**{b_i b_j} )_{i,j}\\
            (\expval**{b_i^\adjoint b_j^\adjoint} )_{i,j} & (\expval**{b_i^\adjoint b_j} )_{i,j}
        \end{pmatrix}
        \overline{\mathbf{W}}\\
        &= \mathbf{W}^\top
        \begin{pmatrix}
            ((1-d_i) \delta_{ij} )_{i,j} & (0)_{i,j}\\
            (0)_{i,j} & (d_i\delta_{ij} )_{i,j}
        \end{pmatrix}
        \overline{\mathbf{W}}\\
        &= \mathbf{W}^\top
        \begin{pmatrix}
            \bI - \Diag(d_k) & \mathbf{0}\\
            \mathbf{0} & \Diag(d_k)
        \end{pmatrix}
        \overline{\mathbf{W}},
    \end{align*}
    with $d_k = \exp(-\beta \epsilon_k)/(1+\exp(-\beta \epsilon_k))$ for all $1\leq k \leq N$.
    The proof is completed if we recall that the sigmoid $\sigma$ satisfies $\sigma(x) = 1 -\sigma(-x)$.
\subsection{Proof of \cref{prop:Pfaffian_PP}} \label{proof:Pfaffian_PP}
    Assume that $i_1, \dots, i_k$ are pairwisely distinct.
    The object of interest is 
    \begin{equation*}
        \expval{N_{i_1}\dots N_{i_k}}_{\rho} = \expval{c^\adjoint_{i_1}c_{i_1}\dots c^\adjoint_{i_k}c_{i_k}}_{\rho}.
    \end{equation*}
    We now use a direct consequence of Wick's theorem for expectations in a Gaussian state, see \citep[Theorem 3]{BFBDHRRSW22Sub}: let $m$ even and let $\beta_1, \dots, \beta_m$ be linear combinations of the $c_i$'s and $c_i^\adjoint$'s for $1\leq i \leq N$, then
    \begin{align}
        \expval{\beta_1\dots\beta_m}_{\rho} = \sum_{\sigma \text{ contraction}} \mathrm{sgn}(\sigma) \expval{\beta_{\sigma(1)}\beta_{\sigma(2)}}_{\rho} \dots \expval{\beta_{\sigma(m-1)}\beta_{\sigma(m)}}_{\rho}. \label{eq:wick_contraction}
    \end{align}
    Take $m=2k$ and use the following definition: $\beta_{2\ell -1} = c^\adjoint_{i_{\ell}}$ and $\beta_{2\ell}= c_{i_{\ell}}$ for $1\leq \ell \leq k$. 
    Thus, we now show that \eqref{eq:wick_contraction} is the Pfaffian of a skewsymmetric matrix made of the $2\times 2$ blocks.
    
    Let us construct this skewsymmetric matrix.
    In details, for the pair $(\ell,\ell^\prime)$ with $1\leq \ell < \ell^\prime \leq k$,  we denote  the $2\times 2$ block by $\pfK_{\ell\ell^\prime}$.
    In order to make any block matrix $\left(\pfK_{i_mi_n}\right)_{1\leq m,n \leq k}$ skewsymmetric, we need to have 
    \begin{equation}
        \pfK_{\ell\ell^\prime}^\top = - \pfK_{\ell^\prime\ell}.\label{eq:block_skew}
    \end{equation}
    The form of the $(\ell,\ell^\prime)$ block with $1\leq \ell < \ell^\prime \leq k$ is
    \[  
        \pfK_{\ell\ell^\prime}
        = 
        \begin{pmatrix}
            \expval**{\beta_{2\ell -1} \beta_{2\ell^\prime -1} }_{\rho} & \expval**{\beta_{2\ell -1} \beta_{2\ell^\prime} }_{\rho}\\
            \star & \expval**{\beta_{2\ell} \beta_{2\ell^\prime} }_{\rho}
        \end{pmatrix}
        =
        \begin{pmatrix}
            \expval**{c_{i_\ell}^\adjoint c_{i_{\ell^\prime}}^\adjoint }_{\rho} & \expval**{c_{i_{\ell}}^\adjoint c_{i_{\ell^\prime}} }_{\rho}\\
            \star & \expval**{c_{i_{\ell}} c_{i_{\ell^\prime}}}_{\rho}
        \end{pmatrix}.
    \]Note that $\expval**{\beta_{2\ell -1} \beta_{2\ell^\prime} }_{\rho}$ does not appear in \eqref{eq:wick_contraction} since $\ell < \ell^\prime$.
    Now, the $\star$ is completed by $ - \expval**{c_{i_{\ell^\prime}}^\adjoint c_{i_{\ell}} }_{\rho}$ to ensure the property \eqref{eq:block_skew} which guarantees skewsymmetry of the block kernel matrix.
    Recalling \eqref{e:CAR}, the Pfaffian kernel is defined as the skewsymmetric matrix whose $(\ell,\ell^\prime)$-block is
    \[
        \pfK_{\ell\ell^\prime}
        =
        \begin{pmatrix}
            \expval**{c_{i_\ell}^\adjoint c_{i_{\ell^\prime}}^\adjoint }_{\rho} & \expval**{c_{i_{\ell}}^\adjoint c_{i_{\ell^\prime}} }_{\rho}\\
            \expval**{c_{i_{\ell}} c_{i_{\ell^\prime}}^\adjoint }_{\rho} - \delta_{i_{\ell}i_{\ell^\prime}} & \expval**{c_{i_{\ell}} c_{i_{\ell^\prime}}}_{\rho}
        \end{pmatrix}.
    \]
    Now, we define the kernel $\pfK(i_{\ell}, i_{\ell^\prime}) = \pfK_{\ell\ell^\prime}$.
    Then, the expression \eqref{eq:wick_contraction} matches the definition of the Pfaffian, i.e.,
    \[
        \expval{\beta_1\dots\beta_m}_{\rho} = \pf\left(\pfK(i,j)\right)_{1\leq i,j\leq m}.
    \]
    By using~\cref{lem:bilinears_BdG}, we obtain the expression \eqref{eq:Pfaffian_kernel} and \eqref{eq:correlation_Pfaffian} follows.

    \subsection{Proof of \cref{prop:Pfaffian_sample_parity}}
    \label{proof:Pfaffian_sample_parity}
    \paragraph*{Key identities.}
    The two main ingredients for this proof are 
    \begin{align}
        &\pf(\mathbf{B}^\top \mathbf{A}\mathbf{B}) = \pf(\mathbf{A}) \det(\mathbf{B}),\label{e:BAB}\\
         &\pf   \Big(\begin{smallmatrix}
            \mathbf{0} & \mathbf{M}\\
            -\mathbf{M} & \mathbf{0}
        \end{smallmatrix}\Big)= (-1)^{N(N-1)/2}\det \mathbf{M},\label{e:pf_block}
    \end{align}
    where $\mathbf{M}$ is an $N\times N$ matrix and $\mathbf{A}$ is $2N\times 2N$ skewsymmetric.
    \paragraph*{Main part of the proof.}
    We begin by using \cref{lem:PfPP_parity}.  
    The expected parity of a sample of $Y\sim \pfPP(\pfK)$ is
    $
        \E_{Y}(-1)^{|Y|} = \pf\left(\mathbb{J} - 2 \mathbb{K}\right).
    $
    Next, by applying a suitable permutation matrix $\mathbf{B}$ of signature $(-1)^{N(N-1)/2}$ on rows and columns of $\mathbb{J} - 2 \mathbb{K}$, we find
    \begin{align}
        \pf(\mathbb{J} - 2 \mathbb{K}) &= (-1)^{N(N-1)/2} \pf\left(\begin{pmatrix}
            \mathbf{0} & \mathbf{I}\\
            -\mathbf{I} & \mathbf{0}
        \end{pmatrix}
        - 2
        \begin{pmatrix}
            \mathbf{P} & \mathbf{K}\\
            -\mathbf{K}^\top & -\overline{\mathbf{P}}
        \end{pmatrix}
        \right)\nonumber\\
        &=(-1)^{N(N-1)/2} (-2)^N \pf\left(\begin{pmatrix}
            \mathbf{0} & -\mathbf{I}/2\\
            \mathbf{I}/2 & \mathbf{0}
        \end{pmatrix}
        +
        \begin{pmatrix}
            \mathbf{P} & \mathbf{K}\\
            -\mathbf{K}^\top & -\overline{\mathbf{P}}
        \end{pmatrix}
        \right),\label{e:sym_to_be_added}
    \end{align}
    where we used \cref{e:BAB} to obtain the first equality.
    At this point, we simply can add the symmetric matrix 
    $\frac{1}{2}\left(\begin{smallmatrix}
        \mathbf{0} & \mathbf{I}\\
        \mathbf{I} & \mathbf{0}
    \end{smallmatrix}\right)$ to the argument of the Pfaffian at the right-hand side of \cref{e:sym_to_be_added}.
    This is valid since the pfaffian of a matrix is the Pfaffian of its skew-symmetric part; see \cref{s:PfPP_basics}.
    We then obtain
    \begin{align}
        \pf(\mathbb{J} - 2 \mathbb{K}) &=(-1)^{N(N-1)/2} (-2)^N \pf\left(\begin{pmatrix}
            \mathbf{0} & \mathbf{0}\\
            \mathbf{I} & \mathbf{0}
        \end{pmatrix}
        +
        \begin{pmatrix}
            \mathbf{P} & \mathbf{K}\\
            -\mathbf{K}^\top & -\overline{\mathbf{P}}
        \end{pmatrix}
        \right)\nonumber\\
        &=(-1)^{N(N-1)/2} (-2)^N \pf\left(\mathbf{C}\mathbf{S}
        \right).\label{e:pf2-K}
    \end{align}
    Recall that 
    $\mathbf{C} 
    = 
    \left(
    \begin{smallmatrix}
            \mathbf{0} & \mathbf{I}\\
            \mathbf{I} & \mathbf{0}
    \end{smallmatrix}\right)$
    and 
    $\mathbf{S} 
    = 
    \left(\begin{smallmatrix}
            \mathbf{I}-\mathbf{K}^\top & -\overline{\mathbf{P}}\\
            \mathbf{P} & \mathbf{K}
    \end{smallmatrix}\right)$.
    Thanks to \cref{lem:bilinears_BdG} and the definition of orthogonal complex matrices, we have
    \begin{align*}
        \mathbf{C}\mathbf{S} = \overline{\mathbf{W}}^{\top}
        \begin{pmatrix}
            \mathbf{0} & \mathbf{D}_{-}\\
            \mathbf{D}_{+} &\mathbf{0} 
        \end{pmatrix}
        \overline{\mathbf{W}},
    \end{align*}
    with $\mathbf{D}_{\pm} = \Diag(\sig(\pm\beta \epsilon_k))$.
    Using that $\pf\left(\mathbf{C}\mathbf{S}
    \right) = \pf\left(\frac{1}{2}\left(\mathbf{C}\mathbf{S} - \mathbf{S}^\top\mathbf{C}^\top
    \right)\right)$ and once again \cref{e:BAB}, we find
    \begin{align*}
        \pf\left(\mathbf{C}\mathbf{S}
    \right) &= \det(\overline{\mathbf{W}})\pf\begin{pmatrix}
        \mathbf{0} & \frac{1}{2}(\mathbf{D}_{-} - \mathbf{D}_{+})\\
        \frac{1}{2}(\mathbf{D}_{+} - \mathbf{D}_{-}) &\mathbf{0} 
    \end{pmatrix} \\
    &= \det(\overline{\mathbf{W}}) (-1)^{N(N-1)/2} \left(\frac{-1}{2}\right)^N\det(\mathbf{D}_{+} - \mathbf{D}_{-}).
    \end{align*}
    Now, we substitue this identity into \cref{e:pf2-K} and
    we conclude the first part of the proof by remarking that 
    \[
        \sig(\beta \epsilon_k) - \sig(-\beta \epsilon_k) = \frac{1 - e^{-\beta\epsilon_k}}{1 + e^{-\beta\epsilon_k}} = \tanh(\beta\epsilon_k/2).
    \]
    Now we show that $
    \det \mathbf{W} = \sign \pf(\mathbf{A}_M)$.
    Recalling the results of \cref{lem:diag_H_BdG}, we have
    \[
    H_{\mathrm{BdG}} =\frac{\i}{2} \mathbf{f}^\top \mathbf{A} \mathbf{f} =  \frac{\i}{2} \bm{\gamma}^\top \mathbf{A}_{M} \bm{\gamma},
    \]
    where $\mathbf{A}$ and $\mathbf{A}_M$ are related by a conjugation with a permutation matrix of signature $(-1)^{N(N-1)/2}$, and therefore, thanks to \cref{e:BAB}
    \[
        \pf(\mathbf{A}) = (-1)^{N(N-1)/2}  \pf(\mathbf{A}_M).
    \]
    Still as a consequence of \cref{lem:diag_H_BdG}, it holds 
    $$
    \mathbf{R}\mathbf{A}\mathbf{R}^\top = 
    \Big(\begin{smallmatrix}
        0 & \Diag(\epsilon_k)\\
        -\Diag(\epsilon_k) & 0
    \end{smallmatrix}\Big).
    $$
    with $\mathbf{R}$ real orthogonal as well as  $\mathbf{W} = \bm{\Omega}^\adjoint \mathbf{R} \bm{\Omega}$ with $\Omega$ unitary.
    Thus, $\det(\mathbf{W}) = \det(\mathbf{R}) \in \{-1,1\}$, and 
    \[
        \pf(\mathbf{A}) = \det(\mathbf{R})\pf   \Big(\begin{smallmatrix}
            0 & \Diag(\epsilon_k)\\
            -\Diag(\epsilon_k) & 0
        \end{smallmatrix}\Big) = \det(\mathbf{W}) (-1)^{N(N-1)/2} \prod_{k}\epsilon_k,
    \]
    where we used \cref{e:pf_block}.
    Hence, by accounting for the permutation signature,
    \[
        \pf(\mathbf{A}_M) = \det(\mathbf{W}) \prod_{k}\epsilon_k.
    \]
    Since we assumed that $\epsilon_k$'s are all \corr{strictly} positive, we obtain $\sign\pf\mathbf{A}_M = \det\mathbf{W}$.

    \subsection{Proof of \cref{e:pf_L_ensemble}}
    \label{a:proof_pmf}
    This proof uses the so-called inclusion-exclusion principle
    \[
        \mathbb{P}(Y = S) = \sum_{J: S\subseteq J} (-1)^{|J\setminus S|} \mathbb{P}(S \subseteq Y),
    \]
    see e.g.\ \citep[Eq (17)]{kaLe22}, and the identity
    \[
        \sum_{J: S\subseteq J} 
        \pf(\mathbb{A}_{J})
        =
          \pf(\mathbb{J}_{\bar{S}} + \mathbb{A}), 
    \]
    for any kernel $\mathbb{A}$ such that $\mathbb{A}(i,j) = - \mathbb{A}(j,i)$; see \citep[p 298]{BoEy05}.
    By combining these two formulae, we find
    \begin{align*}
        \mathbb{P}(Y = S) &= \sum_{J: S\subseteq J} (-1)^{|J\setminus S|} \pf(\mathbb{K}_J)\\
        &= (-1)^{|S|} \sum_{J: S\subseteq J}  \pf(-\mathbb{K}_J)\\
        &=  (-1)^{|S|}\pf(\mathbb{J}_{\bar{S}} - \mathbb{K})\\
        &= (-1)^{N-|S|}\pf(\mathbb{K}-\mathbb{J}_{\bar{S}}),
    \end{align*}
    which is the desired result since $|\bar{S}| = N - |S|$. 

    \subsection{Example of Givens and particle-hole decomposition in a rank-deficient case}
    \label{app:rank_defficient_ph}
    We provide a complementary example for the discussion of \cref{s:double_Givens}, for which we assumed that $\mathbf{W}_2$ was full-rank. 
    In the case $N=3$, we consider a specific orthogonal complex matrix 
    \[
        \mathbf{W}
        = 
        \begin{pmatrix}
            \overline{\mathbf{W}}_1& \overline{\mathbf{W}}_2\\
            \mathbf{W}_2 & \mathbf{W}_1\
        \end{pmatrix}
        = \bm{\Omega}^* 
        \mathbf{R}
        \bm{\Omega}
        \text{ with } \mathbf{R} = \begin{pmatrix}
            0 & 1 & 0 & 0 & 0 & 0\\
            1 & 0 & 0 & 0 & 0 & 0\\
            0 & 0 & 1 & 0 & 0 & 0\\
            0 & 0 & 0 & 1 & 0 & 0\\
            0 & 0 & 0 & 0 & 1 & 0\\
            0 & 0 & 0 & 0 & 0 & 1\\
            \end{pmatrix},
    \]
    which is of the form considered in \cref{lem:diag_H_BdG}.
    By construction, we have $\det \mathbf{W} = -1$.
    Thanks to a short calculation, we see that $\mathbf{W}_1$ is the projector onto the orthogonal of the vector $\left(\begin{smallmatrix}1 & 1 & 0\end{smallmatrix}\right)^\top$, whereas $\mathbf{W}_2 = - \bI + \mathbf{W}_1$. 
    The explicit expression 
    \begin{align*}
        \begin{pmatrix}
            \mathbf{W}_2  \vert \mathbf{W}_1
        \end{pmatrix}
        = 
        \left(\begin{array}{ccc|ccc}
            -\frac{1}{2}  & \frac{1}{2} & 0&  \frac{1}{2} &\frac{1}{2} & 0\\
            \frac{1}{2}  & -\frac{1}{2} & 0&  \frac{1}{2} &\frac{1}{2} & 0\\
            0  & 0 & 0&  0 & 0 & 1\\
        \end{array}
        \right),
    \end{align*}
    makes manifest that $\mathbf{W}_2$ is rank-deficient.
    To factorize $\mathbf{W}$, we first apply a similar preprocessing to \cref{e:Givens_first_step}, i.e., a multiplication  with a product of real Givens matrices $\bV$ on the left, to put a triangle of zeros at the top left of the left block. 
    Next, we multiply on the right by double Givens rotations and particle-hole transformations (see \cref{e:W_rank_deff_Gamma} and \cref{e:W_rank_deff_B}), i.e.,
    \begin{align}
        &\begin{pmatrix}
            \mathbf{W}_2  \vert \mathbf{W}_1
        \end{pmatrix}\nonumber\\
        &\xrightarrow{\mathbf{G}\times}
        \left(\begin{array}{ccc|ccc}
            0  & 0 & 0 &  \frac{1}{\sqrt{2}} & \frac{1}{\sqrt{2}} &0\\
            \frac{1}{\sqrt{2}}  & -\frac{1}{\sqrt{2}} &0 & 0 & 0 &0\\
            0 & 0 & 0 & 0 & 0 &1
        \end{array}
        \right)
        \xrightarrow{\mathbf{G}\times}
        \left(\begin{array}{ccc|ccc}
            0  & 0 & 0 &  \frac{1}{\sqrt{2}} & \frac{1}{\sqrt{2}} &0\\
            0  & 0 & 0 & 0 & 0 &-1\\
            \frac{1}{\sqrt{2}} & -\frac{1}{\sqrt{2}} & 0 & 0 & 0 &0
        \end{array}
        \right)\nonumber \\
        &\xrightarrow{\times\Gamma^\adjoint}\left(\begin{array}{ccc|ccc}
            0  & 0 & 0 &  1 & 0 &0\\
            0  & 0 & 0 & 0 & 0 &-1\\
            0 & -1 & 0 & 0 & 0 &0
        \end{array}
        \right)
        \xrightarrow{\times\Gamma^\adjoint}
        \left(\begin{array}{ccc|ccc}
            0  & 0 & 0 &  1 & 0 & 0\\
            0  & 0 &0 & 0 & 1 & 0\\
            0 & 0 & -1 & 0 & 0 & 0
        \end{array}\right)\label{e:W_rank_deff_Gamma}\\
        &\xrightarrow{\times \mathbf{B}}
        \left(\begin{array}{ccc|ccc}
            0  & 0 & 0 &  1 & 0 & 0\\
            0  & 0 &0 & 0 & 1 & 0\\
            0 & 0 & 0 & 0 & 0 & -1
        \end{array}\right).\label{e:W_rank_deff_B}
    \end{align}
    This yields
    \[
        \bV
        \begin{pmatrix}
            \mathbf{W}_2  \vert \mathbf{W}_1
        \end{pmatrix} 
        \mathbf{O}^* =         \begin{pmatrix}
            \mathbf{0}  \vert \mathbf{D}
        \end{pmatrix}
        \text{ with } \mathbf{O} = \mathbf{B}\mathbf{M},
    \]
    where $\mathbf{M}$ is a product of two double Givens rotations in \cref{e:W_rank_deff_Gamma}.
    Therefore, we here have an example for which $\mathbf{O}$ contains only one particle-hole transformation.
    \subsection{Derivation of the determinantal formula for occupation number correlations without Wick's theorem\label{a:proof_without_wick}}
    This section derives directly the results of \cref{s:projDPP} about projection DPPs without using \cref{prop:DPP_from_rho}.
    While redundant in the paper, we believe this appendix might be useful to readers who want to focus on projection DPPs and skip the need for Wick's theorem. 

    Recall that the $i$th number operator reads $N_i = c_i^\adjoint c_i$.
    Let $\rho$ be the orthogonal projector onto the line generated by $\ket{\psi} = b_1^\adjoint \dots b_r^\adjoint \ket{\emptyset}$ where the fermionic operators $b^*_k,b_k$ are defined in \cref{e:new_annihilation_operators}, to wit $b_k = \sum_{j=1}^N \mathbf u_{kj}^* c_j$ and $b_k^* = \sum_{j=1}^N \mathbf u_{kj} c_j^*$.
    Here, as in \cref{prop:DPP_from_rho}, $\mathbf{u}_{k}$ is the $k$th column of the unitary matrix $\mathbf{U}$.
    In this section, we derive the expression
    \[
        \expval{N_{i_1}\dots N_{i_k}}_\rho = \det (\bK_{\{i_1, \dots, i_k\}}),
    \]
    where $\bK = \mathbf{U}_{[r],:} \mathbf{U}_{[r],:}^*$.

    By expressing the trace with respect to the Fock basis $(\ket{\bn})$ of the operators $c_i$; see \cref{l:representation_theory}, we find
    \begin{align}
        T \triangleq \tr \left(\rho N_{i_1}\dots N_{i_k}\right) &= \sum_{\bn} \expval{N_{i_1}\dots N_{i_k} \rho}{\bn}\nonumber\\
        &= \sum_{\bn} n_{i_1}\dots n_{i_k}\expval{\rho}{\bn}\nonumber\\
        &= \sum_{\bn} n_{i_1}\dots n_{i_k}|\braket{\bn}{\psi}|^2.\label{e:sum_over_n_exp}
    \end{align}
    We know recall that $\ket{\psi} = b_1^\adjoint \dots b_r^\adjoint \ket{\emptyset}$ and
    $
        b_k^* = \sum_{j=1}^N \mathbf u_{kj} c_j^*,   
    $
    so that 
    \begin{align*}
        \braket{\bn}{\psi} &= \expval{(c_N)^{n_N}\dots (c_1)^{n_1}b_1^\adjoint \dots b_r^\adjoint}{\emptyset}\\
        &= \sum_{j_1, \dots, j_r} \mathbf u_{1 j_1} \dots \mathbf u_{r j_r} \expval{(c_N)^{n_N}\dots (c_1)^{n_1}c_{j_1}^\adjoint \dots c_{j_r}^\adjoint}{\emptyset}.
    \end{align*}
    Note that the sequence of Booleans $n_1, \dots, n_N $ has to sum to $r$ for the sum above not to vanish. 
    Accounting for the constraint $n_{i_1} \dots n_{i_k} \neq 0$ in \cref{e:sum_over_n_exp}, this yields
    \begin{align*}
        T = \sum_{\substack{1\leq m_1 \leq \dots \leq m_r \leq N \\ \{i_1, \dots, i_k\} \subseteq \{m_1, \dots,m_r\}} }\left| \sum_{j_1, \dots, j_r} \mathbf u_{1 j_1} \dots \mathbf u_{r j_r} \expval{c_{m_r}\dots c_{m_1} c_{j_1}^\adjoint \dots c_{j_r}^\adjoint}{\emptyset} \right|^2.
    \end{align*}
    At this point, we note that
    \[
        \expval{c_{m_r}\dots c_{m_1} c_{j_1}^\adjoint \dots c_{j_r}^\adjoint}{\emptyset} \neq 0
    \]
    if and only if there is a permutation $\sigma$ such that $j_1 = \sigma(m_1), \dots, j_r = \sigma(m_r)$.
    Also, it is easy to check that $\expval{c_{m_r}\dots c_{m_1} c_{m_1}^\adjoint \dots c_{m_r}^\adjoint}{\emptyset} = 1$ and therefore, by using \cref{e:CAR}, we have
    \[
        \expval{c_{m_r}\dots c_{m_1} c_{\sigma(m_1)}^\adjoint \dots c_{\sigma(m_r)}^\adjoint}{\emptyset} = \sign \sigma.
    \]
    Consequently,
    \begin{align*}
         \sum_{j_1, \dots, j_r} \mathbf u_{1 j_1} \dots \mathbf u_{r j_r} \expval{c_{m_r}\dots c_{m_1} c_{j_1}^\adjoint \dots c_{j_r}^\adjoint}{\emptyset}  &= \sum_{\sigma}  \mathbf u_{1 \sigma(m_1)} \dots \mathbf u_{r \sigma(m_r)}\sign \sigma\\
        & = \det(\mathbf{U}_{[r],\{m_1,\dots,m_r\}}),
    \end{align*}
    where we used the notation $[r] = \{1, \dots, r\}$.
    The expression of $T$ becomes
    \begin{align*}
        T &= \sum_{\substack{1\leq m_1 \leq \dots \leq m_r \leq N \\ \{i_1, \dots, i_k\} \subseteq \{m_1, \dots,m_r\}} } \det \mathbf{U}_{[r],\{m_1,\dots,m_r\}} \det \mathbf{U}^*_{\{m_1,\dots,m_r\},[r]} \\
        &= \sum_{\substack{\mathcal{C}: \{i_1, \dots, i_k\}\subseteq \mathcal{C}\\ |\mathcal{C}| = r}} \det \mathbf{U}_{[r],\mathcal{C}} \det \mathbf{U}^*_{\mathcal{C},[r]} \\
        &= \sum_{\substack{\mathcal{C}: \{i_1, \dots, i_k\}\subseteq \mathcal{C}\\ |\mathcal{C}| = r}} \det (\mathbf{U}^*_{\mathcal{C},[r]}\mathbf{U}_{[r],\mathcal{C}}).
    \end{align*}
    For the last step of this derivation, we will use a well-known formula for projective DPPs.
    More precisely, we first let $\mathbf{Q}= \mathbf{U}_{[r],:}$. Since $\mathbf{Q}$ is a matrix with orthonormal rows, the kernel $\mathbf{K} = \mathbf{Q}^* \mathbf{Q}$ is an orthogonal projector of rank $r$. 
    For $Y\sim \mathrm{DPP}(\bK)$ and $\mathcal{C}$ such that $|\mathcal{C}| = r$, it holds that $\mathbb{P}(Y = \mathcal{C}) = \det (\bK_{\mathcal{C}})$. Furthermore, by definition, we also have
    $
        \mathbb{P}(S \subseteq Y) = \det (\bK_{S})
    $
    for any subset $S$.
    Now, we conclude by 
    \begin{align*}
        T
        &= \sum_{\substack{\mathcal{C}: \{i_1, \dots, i_k\}\subseteq \mathcal{C}\\ |\mathcal{C}| = r}} \det (\bK_{\mathcal{C}}) \\
        &= \sum_{\substack{\mathcal{C}: \{i_1, \dots, i_k\}\subseteq \mathcal{C}\\ |\mathcal{C}| = r}} \mathbb{P}(Y = \mathcal{C})\\
        &= \mathbb{P}(\{i_1, \dots, i_k\}\subseteq Y) = \det (\bK_{\{i_1, \dots, i_k\}}),
    \end{align*}
    which is the desired result.
    \section{Gate details}
    \label{s:gates_details}
    We give here details about the implementation of a Givens rotation using elementary gates.
    Note first that zeroing out a matrix entry $y\in\mathbb{C}\setminus\{0\}$ can be done with a Givens rotation matrix as follows 
    \[
        \begin{pmatrix}
            \cos \theta & \e^{-\i \phi} \sin \theta\\
            -\e^{\i \phi} \sin \theta & \cos \theta
        \end{pmatrix}
        \begin{pmatrix}
            x\\
            y
        \end{pmatrix}
        = 
        \begin{pmatrix}
            r\\
            0
        \end{pmatrix},
    \]
    with $\exp\i \phi = x^* y/|xy|$, $\cos\theta = |x|/(|x|^2 + |y|^2)^{1/2}$ and $\sin\theta = |y|/(|x|^2 + |y|^2)^{1/2}$.
    
    Let us list a few elementary gates to implement the corresponding operation on a pair of qubits.
    For ease of comparison with the online documentation of Qiskit, we denote in this section the Pauli matrices \eqref{e:pauli_matrices} by $\mathbf{X}$, $\mathbf{Y}$, and $\mathbf{Z}$.

    \paragraph*{Controlled NOT gate.} The CNOT gate is a two-qubit gate, i.e., a linear operator on $\mathbb{C}^2\otimes \mathbb{C}^2$.
    In the computational basis $(\ket{00},
    \ket{01},
    \ket{10},
    \ket{11})$, it is described by the matrix
    \begin{align*}
        \begin{array}{c}
            \Qcircuit @C=1em @R=.7em {
            & \targ &  \qw\\
            & \ctrl{-1} & \qw
            }
        \end{array}
        \equiv
        \begin{pmatrix}
            1 &0 &0 &0\\
            0 &0 &0 &1\\
            0 &0 &1 &0\\
            0 &1 &0 &0 
        \end{pmatrix},
    \end{align*}
    where the left-hand side is the graphical depiction of the operator in circuits.
    The CNOT gates thus permutes $\ket{01}$ and $\ket{11}$ while leaving $\ket{00}$ and $\ket{10}$ unchanged.
    In other words, it flips labels $0$ and $1$ in the first factor of the tensor product, provided that the second factor (the ``control" qubit) is $\ket 1$.

    \paragraph*{Controlled unitary gate.}
    Similarly, we can define the following gate, controlled this time by the first qubit,
    \begin{align*}
        \begin{array}{c}
            \Qcircuit @C=1em @R=.7em {
            & \qw &  \ctrl{1} &  \qw \\
            & \qw & \gate{e^{\i \theta \mathbf{Y}}}  & \qw
            }
        \end{array}
        \equiv
        \begin{pmatrix}
            1 &0 &0 &0\\
            0 &1 &0 &0\\
            0 &0 &c \theta & - s \theta\\
            0 &0 &s \theta & c \theta 
           \end{pmatrix},
    \end{align*}
    where we denoted for simplicity $c\theta = \cos \theta$ and $s\theta = \sin \theta$.
    Explicitly, this gate rotates $\ket{0}$ and $\ket{1}$ in the second factor provided the first factor is $\ket{1}$, namely it rotates $\ket{10}$ and $\ket{11}$ and leaves untouched $\ket{00}$ and $\ket{01}$.

    \subsection{Givens rotation using a controlled unitary gate.}
    This implementation is inspired by \citet{JSKSB18}, although our definition of the Givens rotation slightly differs in order to match the definition \cref{e:givens}.
    The $2$-qubit gate representing a Givens rotation 
    is implemented as
    \begin{align*}
        \begin{array}{c}
            \Qcircuit @C=1em @R=.7em {
            & \qw & \targ &  \ctrl{1} & \targ & \qw & \qw \\
            & \gate{e^{\i \phi \mathbf{Z}/2}} & \ctrl{-1} & \gate{e^{\i \theta \mathbf{Y}}}  & \ctrl{-1} 
            & \gate{e^{-\i \phi \mathbf{Z}/2}}& \qw
            }
        \end{array}\equiv\begin{pmatrix}
    1 &0 &0 &0\\
    0  &c \theta &  e^{-\i\phi}s \theta & 0\\
    0  &- e^{\i\phi}s \theta & c \theta & 0\\
    0 &0 &0 &1
   \end{pmatrix}.
    \end{align*}
    We now give calculation details.
    First, we compute the inner bock of three gates as follows
    \begin{align*}
        &\begin{array}{c}
            \Qcircuit @C=1em @R=.7em {
            & \qw & \targ &  \ctrl{1} & \targ & \qw  \\
            & \qw & \ctrl{-1} & \gate{e^{\i \theta \mathbf{Y}}}  & \ctrl{-1} 
            & \qw
            }
        \end{array}
        \equiv
        \begin{pmatrix}
         1 &0 &0 &0\\
         0 &0 &0 &1\\
         0 &0 &1 &0\\
         0 &1 &0 &0 
        \end{pmatrix}
        \begin{pmatrix}
         1 &0 &0 &0\\
         0 &1 &0 &0\\
         0 &0 &c \theta & - s \theta\\
         0 &0 &s \theta & c \theta 
        \end{pmatrix}
        \begin{pmatrix}
         1 &0 &0 &0\\
         0 &0 &0 &1\\
         0 &0 &1 &0\\
         0 &1 &0 &0 
        \end{pmatrix}
        \\
        &= \begin{pmatrix}
            1 &0 &0 &0\\
            0  &c \theta & s \theta & 0\\
            0  &-s \theta & c \theta & 0\\
            0 &0 &0 & 1
           \end{pmatrix}.
    \end{align*}
Now, the result follows by noting that 
\[
    \begin{array}{c}
            \Qcircuit @C=1em @R=.7em {
            & \qw & \qw  \\
            & \gate{e^{\i \phi \mathbf{Z}/2}} & \qw
            }
        \end{array}
        \equiv
    \begin{pmatrix}
        e^{\i\phi/2} &0 &0 &0\\
         0 &e^{-\i\phi/2} &0 &0\\
         0 &0 &e^{\i\phi/2} & 0\\
         0 &0 &0 & e^{-\i\phi/2} 
        \end{pmatrix}.
\]

\subsection{Qiskit 0.42.1 implementation without controlled unitary gate.}

The Givens rotation is implemented in Qiskit, as in \cref{f:dpp_circuit},  thanks to a gate
\[
    R_{XX+YY}(2\theta,\phi -\pi/2) = \begin{pmatrix}
        1 &0 &0 &0\\
        0  &c \theta &  e^{-\i\phi}s \theta & 0\\
        0  &-e^{\i\phi}s \theta & c \theta & 0\\
        0 &0 &0 &1
       \end{pmatrix}
\] where the 2-qubit parameterized $XX+YY$ interaction is defined as 
\begin{align*}
&R_{XX+YY}(\theta,\beta)\\
&= \left(\exp(-\i \beta \mathbf{Z}/2) \otimes \mathbb{I}\right) \exp\left(-\i \frac{\theta}{2}\frac{\mathbf{X}\otimes \mathbf{X} + \mathbf{Y} \otimes \mathbf{Y}}{2} \right)\left(\exp(\i \beta \mathbf{Z}/2) \otimes \mathbb{I}\right)\\
&=\begin{pmatrix}
    1 &0 &0 &0\\
    0  &c \theta/2 & -\i e^{-\i\beta}s \theta/2 & 0\\
    0  &-\i e^{\i\beta}s \theta/2 & c \theta/2 & 0\\
    0 &0 &0 &1
   \end{pmatrix}.
\end{align*}
This gate is implemented in Qiskit 0.42.1 as a composition of elementary gates containing $12$ single-qubit gates and $2$ CNOT gates.
Also, note that some details in their gate definitions may vary from the one used in this section.
\section{Sources of error in quantum computers}
\label{s:errors_in_quantum_computers}
Different quantities are usually specified by the constructor to assess
the efficiency of a quantum machine. 
First, the \emph{error rate} (or \emph{readout error}) is an estimated frequency of getting an undesired measurement (a ``bit flip") when measuring a state drawn independently from a fixed, user-defined prior distribution.
Second, the de-excitation (or \emph{relaxation}) of a qubit prepared in its excited
state, and thus supposed to represent $\ket{1}$, is a natural physical process. 
It is usually understood to be caused by the coupling of the qubit to electromagnetic radiation or, more abstractly, its environment.
The typical timescale of the relaxation processes is called the relaxation time, and usually denoted by $T_1$.
Moreover, the presence of an environment is not only the source of relaxation (bit-flip error), but also an additional source of phase errors.
This type of error destroys the quantum coherence properties in and between the qubits, i.e., it modifies the correlation structure in Boolean vectors built using Born's rule \eqref{e:born}, possibly to the point of making the qubits behave as simple classical bits.
This \emph{decoherence} process is usually characterized by a typical timescale denoted $T_2$.
The output of any circuit with depth longer than $T_1$ or $T_2$ is likely to be strongly contaminated with the corresponding noises.

To give concrete numbers, the IBM 127 qubits Washington platform\footnote{See 
\url{https://quantum-computing.ibm.com/services/resources?services=systems}}
has a median $T_1$ (the median is over all qubits) of
$T_1 = \SI{95,22}{\us}$, and the median $T_2$ is $T_2 = \SI{92,17}{\us}$.
This means that, in practice, after about $\SI{100}{\us}$, more than half
of the qubits have either decohered (become classical bits) or have
relaxed to their ground state. 
Moreover, the median readout error is $p_{\text{err, ro}} = \SI{1.350e-2}{}$, and the median CNOT gate error is $p_{\text{err, CNOT}} = \SI{1.287e-2 }{}$.
Ideally, characterizing the noise of all the gates would require a quantum process tomography as well as quantum state tomography, to characterize the state of all qubits. 
In practice, this is out of reach for large systems. 
To circumvent this issue, \cite{Preskill-2022} propose to partially characterize the noise with a scalable and robust algorithm called \emph{randomized benchmarking}. 
This is how the cited numbers for IBM machines have been benchmarked.\footnote{\url{https://qiskit.org/textbook/ch-quantum-hardware/randomized-benchmarking.html}}

The long-term goal of the quantum computer race is to build a 
fault-tolerant quantum computer, based on quantum error
corrections codes. 
These are techniques that build up on
redundancies of the logical qubits to be robust to noise-induced errors according; see the so-called \emph{threshold theorem} \citep{Kitaev-2003,Knill-1998}. 
As in the classical case, the required redundancy
depends of the actual values of the error rates, which are to this day still too high to have a realistic implementation.
In the meantime, other techniques are being developed to alleviate
the influence or errors, either by directly eliminating errors on the hardware itself, like with so-called \emph{spin echos} or \emph{dynamic decoupling} techniques \citep{Preskill-1998}, or by statistical postprocessing, with so-called \emph{error mitigation techniques} \citep{Cai-2022}.

%% file: biblio.bib
@article{DLMS19,
  title={Noninteracting fermions in a trap and random matrix theory},
  author={Dean, D. S. and Le Doussal, P. and Majumdar, S. N. and Schehr, G.},
  journal={Journal of Physics A: Mathematical and Theoretical},
  volume={52},
  number={14},
  pages={144006},
  year={2019},
  publisher={IOP Publishing}
}

@inproceedings{DCMW19,
  title={Minimax experimental design: {B}ridging the gap between statistical and worst-case approaches to least squares regression},
  author={Derezi{\'n}ski, M. and Clarkson, K. L. and Mahoney, M. W. and Warmuth, M. K.},
  booktitle={Conference on Learning Theory},
  pages={1050--1069},
  year={2019},
  organization={PMLR}
}

@article{grabsch2019pfaffian,
  title={Pfaffian formula for fermion parity fluctuations in a superconductor and application to Majorana fusion detection},
  author={Grabsch, A. and Cheipesh, Y. and Beenakker, C.W.J.},
  journal={Annalen der Physik},
  volume={531},
  number={10},
  pages={1900129},
  year={2019},
}

@book{CTDL77,
      author        = "Cohen-Tannoudji, C. and Diu, B. and Laloë,
                       F.",
      title         = "{Quantum mechanics; 1st ed.}",
      publisher     = "Wiley",
      address       = "New York, NY",
      year          = "1977",
}

@misc{Qiskit,
    author = {{Qiskit contributors}},
    title = {{Qiskit: An Open-source Framework for Quantum Computing}},
    year = {2023},
    doi = {10.5281/zenodo.2573505}
}

@misc{IBMQ,
    author = {{IBM Quantum}},
    url = {https://quantum-computing.ibm.com/},
    year = {2021}
}

@article{DFP08,
    author = {Dierckx, B. and Fannes, M. and Pogorzelska, M.},
    title = "{Fermionic quasifree states and maps in information theory}",
    journal = {Journal of Mathematical Physics},
    volume = {49},
    number = {3},
    year = {2008},
    url = {https://doi.org/10.1063/1.2841326},
}

@article{lyons2003,
  title={Determinantal probability measures},
  author={Lyons, R.},
  journal={Publications Mathematiques de l'IHES},
  volume={98},
  pages={167--212},
  year={2003},
  url={https://doi.org/10.1007/s10240-003-0016-0}
}

@article{PoiBa23,
  title={On proportional volume sampling for experimental design in general spaces},
  author={Poinas, A. and Bardenet, R.},
  journal={Statistics and Computing},
  volume={33},
  number={1},
  pages={29},
  year={2023},
  publisher={Springer},
  url={https://doi.org/10.1007/s11222-022-10115-0}
}

@InProceedings{DLM20,
  title = 	 {Bayesian experimental design using regularized determinantal point processes},
  author =       {Derezinski, M. and Liang, F. and Mahoney, M.},
  booktitle = 	 {Proceedings of the Twenty Third International Conference on Artificial Intelligence and Statistics},
  pages = 	 {3197--3207},
  year = 	 {2020},
  volume = 	 {108},
  series = 	 {Proceedings of Machine Learning Research},
  publisher =    {PMLR},
  url = 	 {https://proceedings.mlr.press/v108/derezinski20a.html},
}

@article{KerPra22,
  title={Quantum machine learning with subspace states},
  author={Kerenidis, I. and Prakash, A.},
  journal={arXiv preprint arXiv:2202.00054},
  url={https://arxiv.org/abs/2202.00054},
  year={2022}
}

@article{Preskill-1998,
  title={Lecture notes for physics 229: Quantum information and computation},
  author={Preskill, J.},
  journal={California Institute of Technology},
  volume={16},
  number={1},
  pages={1--8},
  year={1998},
  url={http://theory.caltech.edu/~preskill/ph229/}
}

@article{Preskill-2022,
  title={The randomized measurement toolbox},
  author={Elben, A. and Flammia, S. T. and Huang, H.-Y. and Kueng, R. and Preskill, J. and Vermersch, B. and Zoller, P.},
  journal={Nature Reviews Physics},
  pages={1--16},
  year={2022},
  publisher={Nature Publishing Group UK London},
  url={https://arxiv.org/abs/2203.11374}
}

@article{Knill-1998,
  title={Resilient quantum computation: error models and thresholds},
  author={Knill, E. and Laflamme, R. and Zurek, W. H.},
  journal={Proceedings of the Royal Society of London. Series A: Mathematical, Physical and Engineering Sciences},
  volume={454},
  number={1969},
  pages={365--384},
  year={1998},
  publisher={The Royal Society},
  url={https://www.jstor.org/stable/i203157},
}

@article{Kitaev-2003,
  title={Fault-tolerant quantum computation by anyons},
  author={Kitaev, A. Yu.},
  journal={Annals of physics},
  volume={303},
  number={1},
  pages={2--30},
  year={2003},
  publisher={Elsevier},
  Url={https://doi.org/10.1016/S0003-4916(02)00018-0}
}

@article{Cai-2022,
  title={Quantum error mitigation},
  author={Cai, Z. and Babbush, R. and Benjamin, S. C. and Endo, S. and Huggins, W. J. and Li, Y. and McClean, J. R. and O'Brien, T. E.},
  journal={arXiv preprint arXiv:2210.00921},
  year={2022},
  url={https://arxiv.org/abs/2210.00921}
}


%% file: stats.bib
@article{Kitaev01,
url = {https://dx.doi.org/10.1070/1063-7869/44/10S/S29},
year = {2001},
publisher = {},
volume = {44},
number = {10S},
pages = {131},
author = {Kitaev, A.Y.},
title = {Unpaired Majorana fermions in quantum
wires},
journal = {Physics-Uspekhi},
}

@article{RSML18,
  title = {Quantum singular-value decomposition of nonsparse low-rank matrices},
  author = {Rebentrost, P. and Steffens, A. and Marvian, I. and Lloyd, S.},
  journal = {Phys. Rev. A},
  volume = {97},
  issue = {1},
  pages = {012327},
  numpages = {6},
  year = {2018},
  publisher = {American Physical Society},
  url = {https://link.aps.org/doi/10.1103/PhysRevA.97.012327}
}

@article{TBUA23,
  title={Extended L-ensembles: a new representation for Determinantal Point Processes},
  author={Tremblay, N. and Barthelm{\'e}, S. and Usevich, K. and Amblard, P.-O.},
  journal={The Annals of Applied Probability},
  volume={33},
  number={1},
  pages={613--640},
  year={2023},
  publisher={Institute of Mathematical Statistics}
}

@article{JAQK20,
  title = {Point processes with Gaussian boson sampling},
  author = {Jahangiri, S. and Arrazola, J.M. and Quesada, N. and Killoran, N.},
  journal = {Phys. Rev. E},
  volume = {101},
  issue = {2},
  pages = {022134},
  numpages = {13},
  year = {2020},
  url = {https://link.aps.org/doi/10.1103/PhysRevE.101.022134}
}

@InProceedings{FreBru19,
  title = 	 {Approximating Orthogonal Matrices with Effective Givens Factorization},
  author =       {Frerix, T. and Bruna, J.},
  booktitle = 	 {Proceedings of the 36th International Conference on Machine Learning},
  pages = 	 {1993--2001},
  year = 	 {2019},
  volume = 	 {97},
  series = 	 {Proceedings of Machine Learning Research},
  month = 	 {09--15 Jun},
  publisher =    {PMLR},
  url = 	 {https://proceedings.mlr.press/v97/frerix19a.html},
}

@BOOK{Hal10,
      author       = {Hall, B. C.},
      title        = {{L}ie {G}roups, {L}ie {A}lgebras, and {R}epresentations:
                      {A}n {E}lementary {I}ntroduction},
      volume       = {222},
      address      = {New York},
      publisher    = {Springer},
      series       = {Graduate Texts in Mathematics},
      year         = {2010},
}

@article{SRFL08,
  title = {Classification of topological insulators and superconductors in three spatial dimensions},
  author = {Schnyder, A.P. and Ryu, S. and Furusaki, A. and Ludwig, A.W.W.},
  journal = {Phys. Rev. B},
  volume = {78},
  issue = {19},
  pages = {195125},
  numpages = {22},
  year = {2008},
  publisher = {American Physical Society},
  url = {https://link.aps.org/doi/10.1103/PhysRevB.78.195125}
}

@inproceedings{DeKaMa20,
 author = {Derezinski, M. and Khanna, R. and Mahoney, M. W.},
 booktitle = {Advances in Neural Information Processing Systems},
 pages = {4953--4964},
 publisher = {Curran Associates, Inc.},
 title = {Improved guarantees and a multiple-descent curve for Column Subset Selection and the Nystrom method},
 url = {https://tinyurl.com/nhj5wexa},
 volume = {33},
 year = {2020}
}

@article{FSS21,
author = {Fanuel, M. and Schreurs, J. and Suykens, J.A.K.},
title = {Diversity Sampling is an Implicit Regularization for Kernel Methods},
journal = {SIAM Journal on Mathematics of Data Science},
volume = {3},
number = {1},
pages = {280-297},
year = {2021},
URL = {https://doi.org/10.1137/20M1320031}
}

@article{TBA19,
  author  = {Tremblay, N. and Barthelm{{\'e}}, S. and Amblard, P.-O.},
  title   = {Determinantal Point Processes for Coresets},
  journal = {Journal of Machine Learning Research},
  year    = {2019},
  volume  = {20},
  number  = {168},
  pages   = {1--70},
  url     = {http://jmlr.org/papers/v20/18-167.html}
}

@article{KMWGACGB18,
  title = {Quantum Simulation of Electronic Structure with Linear Depth and Connectivity},
  author = {Kivlichan, I. D. and McClean, J. and Wiebe, N. and Gidney, C. and Aspuru-Guzik, A. and Chan, G. K.-L. and Babbush, R.},
  journal = {Phys. Rev. Lett.},
  volume = {120},
  issue = {11},
  pages = {110501},
  numpages = {6},
  year = {2018},
  publisher = {American Physical Society},
  url = {https://link.aps.org/doi/10.1103/PhysRevLett.120.110501}
}

@article{SaKu78,
	Author = {Sameh, A. H. and Kuck, D. J.},
	Journal = {Journal of the ACM (JACM)},
	Number = {1},
	Pages = {81--91},
	Publisher = {ACM New York, NY, USA},
	Title = {On stable parallel linear system solvers},
	Volume = {25},
	Year = {1978},
	Url={https://doi.org/10.1145/322047.322054}}

@article{WHWCNT15,
	Author = {Wecker, D. and Hastings, M. B. and Wiebe, N. and Clark, B. K. and Nayak, C. and Troyer, M.},
	Journal = {Physical Review A},
	Number = {6},
	Pages = {062318},
	Publisher = {APS},
	Title = {Solving strongly correlated electron models on a quantum computer},
	Volume = {92},
	Year = {2015},
	url = {https://link.aps.org/doi/10.1103/PhysRevA.92.062318}}

@article{OGKL01,
	Author = {Ortiz, G. and Gubernatis, J. E. and Knill, E. and Laflamme, R.},
	Journal = {Physical Review A},
	Number = {2},
	Pages = {022319},
	Publisher = {APS},
	Title = {Quantum algorithms for fermionic simulations},
	Volume = {64},
	Year = {2001},
	Url={https://doi.org/10.1103/PhysRevA.64.022319}}

@article{DGHL12,
	Author = {Demmel, J. and Grigori, L. and Hoemmen, M. and Langou, J.},
	Journal = {SIAM Journal on Scientific Computing},
	Number = {1},
	Pages = {A206--A239},
	Publisher = {SIAM},
	Title = {Communication-optimal parallel and sequential {QR} and {LU} factorizations},
	Volume = {34},
	Year = {2012},
	url={https://arxiv.org/abs/0806.2159}}

@article{BVJ07,
	Author = {Bouten, L. and Van Handel, R. and James, M. R.},
	Journal = {SIAM Journal on Control and Optimization},
	Number = {6},
	Pages = {2199--2241},
	Publisher = {SIAM},
	Title = {An introduction to quantum filtering},
	Volume = {46},
	Year = {2007},
	url={https://arxiv.org/abs/math/0601741}}

@article{Rains2000,
	Author = {Rains, E. M.},
	Journal = {arXiv preprint math/0006097},
	Title = {Correlation functions for symmetrized increasing subsequences},
	Year = {2000},
	url={https://arxiv.org/abs/math/0006097}}

@unpublished{Moore2014,
	Author = {Moore, G. W.},
	Title = {Quantum symmetries and compatible hamiltonians},
	Url = {http://www.physics.rutgers.edu/gmoore/QuantumSymmetryBook.pdf},
	Year = {2014}}

@inproceedings{KuTa11,
	Address = {Arlington, Virginia, USA},
	Author = {Kulesza, A. and Taskar, B.},
	Booktitle = {Proceedings of the Twenty-Seventh Conference on Uncertainty in Artificial Intelligence},
	Isbn = {9780974903972},
	Location = {Barcelona, Spain},
	Numpages = {9},
	Pages = {419--427},
	Publisher = {AUAI Press},
	Series = {UAI'11},
	Title = {Learning Determinantal Point Processes},
	Year = {2011},
	Url={https://arxiv.org/abs/1202.3738}}

@InProceedings{BTA22,
  title = 	 {A Faster Sampler for Discrete Determinantal Point Processes},
  author =       {Barthelm\'e, S. and Tremblay, N. and Amblard, P.-O.},
  booktitle = 	 {Proceedings of The 26th International Conference on Artificial Intelligence and Statistics},
  pages = 	 {5582--5592},
  year = 	 {2023},
  volume = 	 {206},
  series = 	 {Proceedings of Machine Learning Research},
  publisher =    {PMLR},
  url = 	 {https://proceedings.mlr.press/v206/barthelme23a.html},
}

@inproceedings{Wilson96,
	Address = {New York, NY, USA},
	Author = {Wilson, D. B.},
	Booktitle = {Proceedings of the Twenty-Eighth Annual ACM Symposium on Theory of Computing},
	Location = {Philadelphia, Pennsylvania, USA},
	Numpages = {8},
	Pages = {296--303},
	Publisher = {Association for Computing Machinery},
	Series = {STOC '96},
	Title = {{Generating Random Spanning Trees More Quickly than the Cover Time}},
	Url = {https://doi.org/10.1145/237814.237880},
	Year = {1996}}

@article{DeMa20,
	Author = {Derezinski, M. and Mahoney, M. W.},
	Journal = {Notices of the American Mathematical Society},
	Number = {1},
	Pages = {34--45},
	Title = {Determinantal point processes in randomized numerical linear algebra},
	Volume = {68},
	Year = {2021},
	url={https://arxiv.org/abs/2005.03185}}

@inproceedings{KySo18,
	Author = {Kyng, R. and Song, Z.},
	Booktitle = {2018 IEEE 59th Annual Symposium on Foundations of Computer Science (FOCS)},
	Pages = {373--384},
	Title = {{A Matrix Chernoff Bound for Strongly Rayleigh Distributions and Spectral Sparsifiers from a few Random Spanning Trees}},
	Url = {https://arxiv.org/abs/1810.08345},
	Year = {2018}}

@article{Pemantle91,
	Author = {Pemantle, R.},
	Journal = {The Annals of Probability},
	Number = {4},
	Pages = {1559--1574},
	Publisher = {Institute of Mathematical Statistics},
	Title = {{Choosing a Spanning Tree for the Integer Lattice Uniformly}},
	Url = {http://www.jstor.org/stable/2244527},
	Volume = {19},
	Year = {1991}}

@article{BraKit2002,
	Author = {Bravyi, S.B. and Kitaev, A.Y.},
	Journal = {Annals of Physics},
	Number = {1},
	Pages = {210-226},
	Title = {Fermionic Quantum Computation},
	Url = {https://www.sciencedirect.com/science/article/pii/S0003491602962548},
	Volume = {298},
	Year = {2002}}

@article{Ball2005,
	Author = {Ball, R. C.},
	Issue = {17},
	Journal = {Phys. Rev. Lett.},
	Numpages = {4},
	Pages = {176407},
	Title = {Fermions without Fermion Fields},
	Url = {https://link.aps.org/doi/10.1103/PhysRevLett.95.176407},
	Volume = {95},
	Year = {2005}}

@article{VerCir2005,
	Author = {Verstraete, F. and Cirac, J.I.},
	Journal = {Journal of Statistical Mechanics: Theory and Experiment},
	Number = {09},
	Pages = {P09012},
	Title = {Mapping local Hamiltonians of fermions to local Hamiltonians of spins},
	Url = {https://dx.doi.org/10.1088/1742-5468/2005/09/P09012},
	Volume = {2005},
	Year = {2005}}

@article{olshanski2020,
	Author = {Olshanski, G.},
	Journal = {Communications in Mathematical Physics},
	Number = {1},
	Pages = {507--555},
	Publisher = {Springer},
	Title = {Determinantal point processes and fermion quasifree states},
	Volume = {378},
	Year = {2020},
	Url={https://doi.org/10.1007/s00220-020-03716-1}}

@article{Soshnikov03,
	Author = {Soshnikov, A.},
	Journal = {Journal of statistical physics},
	Pages = {611--622},
	Title = {Janossy Densities. II. Pfaffian Ensembles},
	Volume = {113},
	Year = {2003},
	url={https://doi.org/10.1023/A:1026077020147}}

@article{BoEy05,
	Author = {Borodin, A. and Rains, E. M.},
	Journal = {Journal of statistical physics},
	Number = {3},
	Pages = {291--317},
	Title = {Eynard--Mehta theorem, Schur process, and their Pfaffian analogs},
	Volume = {121},
	Year = {2005}}

@article{kaLe22,
	Author = {Kassel, A. and L{\'e}vy, T.},
	Journal = {Annales de l'Institut Henri Poincar{\'e} D},
	Title = {Determinantal probability measures on Grassmannians},
	Year = {2022},
	Url={https://doi.org/10.4171/aihpd/152}}

@article{BDQ21,
	Author = {Bufetov, A. I. and Deelan Cunden, F. and Qiu, Y.},
	Doi = {10.1214/20-AIHP1099},
	Journal = {Annales de l'Institut Henri Poincar{\'e}, Probabilit{\'e}s et Statistiques},
	Number = {2},
	Pages = {856 -- 871},
	Publisher = {Institut Henri Poincar{\'e}},
	Title = {{Conditional measures for Pfaffian point processes: Conditioning on a bounded domain}},
	Url = {https://doi.org/10.1214/20-AIHP1099},
	Volume = {57},
	Year = {2021}}

@article{Kargin2014,
	Author = {Kargin, V.},
	Journal = {Journal of Statistical Physics},
	Number = {3},
	Pages = {681--704},
	Publisher = {Springer},
	Title = {On Pfaffian random point fields},
	Volume = {154},
	Year = {2014},
	Url={https://doi.org/10.1007/s10955-013-0900-z}}

@article{koshida21,
	Author = {Koshida, S.},
	Journal = {SIGMA. Symmetry, Integrability and Geometry: Methods and Applications},
	Pages = {008},
	Publisher = {SIGMA. Symmetry, Integrability and Geometry: Methods and Applications},
	Title = {Pfaffian Point Processes from Free Fermion Algebras: Perfectness and Conditional Measures},
	Volume = {17},
	Year = {2021},
	Url={https://doi.org/10.3842/SIGMA.2021.008}}

@article{TerDiVi02,
	Author = {Terhal, B. M. and DiVincenzo, D. P.},
	Issue = {3},
	Journal = {Phys. Rev. A},
	Numpages = {10},
	Pages = {032325},
	Publisher = {American Physical Society},
	Title = {Classical simulation of noninteracting-fermion quantum circuits},
	Url = {https://link.aps.org/doi/10.1103/PhysRevA.65.032325},
	Volume = {65},
	Year = {2002}}

@book{Mac17,
	Author = {Macchi, O.},
	Publisher = {Walter Warmuth Verlag},
	Title = {Point processes and coincidences -- Contributions to the theory, with applications to statistical optics and optical communication, augmented with a scholion by Suren Poghosyan and Hans Zessin},
	Year = {2017}}

@phdthesis{Mac72,
	Author = {Macchi, O.},
	School = {Universit\'e Paris-Sud},
	Title = {Processus ponctuels et coincidences -- Contributions {\`a} l'{\'e}tude th{\'e}orique des processus ponctuels, avec applications {\`a} l'optique statistique et aux communications optiques},
	Year = {1972}}

@article{JSKSB18,
	Author = {Jiang, Z. and Sung, K. J. and Kechedzhi, K. and Smelyanskiy, V. N. and Boixo, S.},
	Journal = {Physical Review Applied},
	Number = {4},
	Pages = {044036},
	Publisher = {APS},
	Title = {Quantum algorithms to simulate many-body physics of correlated fermions},
	Volume = {9},
	Year = {2018},
	url={https://arxiv.org/abs/1711.05395}}

@book{NiCh10,
	Author = {Nielsen, M. A. and Chuang, I. L.},
	Edition = {10th anniversary edition},
	Publisher = {Cambridge University Press},
	Title = {Quantum computation and quantum information},
	Year = {2010},
	Url={https://doi.org/10.1017/CBO9780511976667}}

@techreport{Nie05,
	Author = {Nielsen, M. A.},
	Institution = {The University of Queensland},
	Title = {The fermionic canonical commutation relations and the {Jordan-Wigner} transform},
	Year = {2005},
	Url={https://futureofmatter.com/assets/fermions_and_jordan_wigner.pdf}}

@article{BFBDHRRSW22Sub,
	Author = {Bardenet, R. and Feller, A. and Bouttier, J. and Degiovanni, P. and Hardy, A. and Ran{\c{c}}on, A. and Roussel, B. and Schehr, G. and Westbrook, C. I.},
	Journal = {arXiv preprint arXiv:2210.05522},
	Title = {From point processes to quantum optics and back},
	Year = {2022},
	Url = {https://arxiv.org/abs/2210.05522},
}

@article{Sos00,
	Author = {Soshnikov, A.},
	Journal = {Russian Mathematical Surveys},
	Pages = {923--975},
	Title = {Determinantal random point fields},
	Volume = {55},
	Year = {2000},
	url = {https://dx.doi.org/10.1070/RM2000v055n05ABEH000321},}

@article{BeBaCh20b,
	Author = {Belhadji, A. and Bardenet, R. and Chainais, P.},
	Journal = {Journal of Machine Learning Research (JMLR)},
	Title = {A determinantal point process for column subset selection},
	Year = {2020},
	url= {http://jmlr.org/papers/v21/19-080.html}}

@article{Kassel15,
	Author = {Kassel, A.},
	Journal = {ESAIM: Proc.},
	Pages = {60-73},
	Title = {Learning About Critical Phenomena from Scribbles and Sandpiles},
	Url = {https://doi.org/10.1051/proc/201551004},
	Volume = 51,
	Year = 2015}

@article{BaHa20,
	Author = {Bardenet, R. and Hardy, A.},
	Journal = {Annals of Applied Probability},
	Title = {Monte {C}arlo with Determinantal Point Processes},
	Year = {2020},
	url ={https://doi.org/10.1214/19-AAP1504},}

@book{GoVa12,
	Author = {Golub, G. H. and Van Loan, C. F.},
	Publisher = {JHU Press},
	Title = {Matrix computations},
	Year = {2012}}

@article{HKPV06,
	Author = {Hough, J. B. and Krishnapur, M. and Peres, Y. and Vir\'ag, B.},
	Journal = {Probability surveys},
	Title = {Determinantal processes and independence},
	Year = {2006},
	url={https://arxiv.org/abs/math/0503110}}

@article{KuTa12,
	Author = {Kulesza, A. and Taskar, B.},
	Journal = {Foundations and Trends in Machine Learning},
	Title = {Determinantal point processes for machine learning},
	Year = {2012},
	url={http://dx.doi.org/10.1561/2200000044}}

@article{Joh05,
	Author = {Johansson, K.},
	Eprint = {math-ph/0510038},
	Journal = {ArXiv Mathematical Physics e-prints},
	Title = {{Random matrices and determinantal processes}},
	Year = 2005}

@article{LaMoRu15,
	Author = {Lavancier, F. and M{\o}ller, J. and Rubak, E.},
	Journal = {Journal of the Royal Statistical Society, Series B},
	Series = {B},
	Title = {Determinantal point process models and statistical inference},
	Year = {2015},
	Url={https://www.jstor.org/stable/24775312}}
